\documentclass[journal,letterpaper,onecolumn,twoside,nofonttune]{IEEEtran}

\usepackage[utf8]{inputenc} 
\usepackage[T1]{fontenc}
\usepackage{url}
\usepackage{ifthen}
\usepackage{cite}
\usepackage[cmex10]{amsmath}
\usepackage{amsmath}
\usepackage{amsthm}
\usepackage{amsfonts}
\usepackage{amssymb}
\usepackage{mathtools}
\usepackage{color}
\usepackage{verbatim}
\usepackage[normalem]{ulem}
\usepackage{enumitem}
\usepackage{hyperref}
\usepackage[ruled,vlined]{algorithm2e}
\usepackage{algpseudocode}

\usepackage{floatrow}
\usepackage{tikz}
\usetikzlibrary{matrix,arrows}

\definecolor{purple}{rgb}{0.5, 0.0, 0.5}
\definecolor{dark_green}{rgb}{0.0, 0.5, 0.0}
\definecolor{mygray}{gray}{0.6}

\DeclareMathAlphabet{\mathpzc}{OT1}{pzc}{m}{it}

\newcommand{\white}{\color{white}}

\newcommand{\Cc}{\mathcal{C}}
\newcommand{\ow}{\mathcal{O}}

\newcommand{\I}{\mathcal{I}}

\newcommand{\K}{\mathcal{K}}
\newcommand{\M}{\mathcal{M}}
\newcommand{\Mb}{\bold{M}}

\newcommand{\Qb}{\bold{Q}}

\newcommand{\Qcal}{\mathcal{Q}}

\newcommand{\Rb}{\bold{R}}
\newcommand{\Rt}{\tilde{R}}
\newcommand{\Otil}{\tilde{O}}
\newcommand{\wb}{\bold{w}}
\newcommand{\Wb}{\bold{W}}

\newcommand{\Ub}{\bold{U}}
\newcommand{\vb}{\bold{v}}
\newcommand{\Vb}{\bold{V}}
\newcommand{\xb}{\bold{x}}
\newcommand{\Xb}{\bold{X}}

\newcommand{\xbg}{{\grave{\bold{x}}}}

\newcommand{\xbh}{{\hat{\bold{x}}}}

\newcommand{\yb}{\bold{y}}

\newcommand{\ellb}{\bar{\ell}}
\newcommand{\ellg}{\grave{\ell}}
\newcommand{\ellt}{\tilde{\ell}}
\newcommand{\xih}{\hat{\xi}}

\newcommand{\Fb}{\bold{F}}
\newcommand{\Gb}{\bold{G}}

\newcommand{\zb}{\bold{z}}
\newcommand{\Zb}{\bold{Z}}
\newcommand{\Pibold}{\bold{\Pi}}
\newcommand{\Piba}{\acute{\Pibold}}
\newcommand{\Pibh}{\hat{\Pibold}}
\newcommand{\Pibt}{\tilde{\Pibold}}
\newcommand{\Omb}{\bold{\Omega}}

\newcommand{\Ombwt}{\widetilde{\Omb}}
\newcommand{\Phib}{\bold{\Phi}}
\newcommand{\D}{\mathcal{D}}

\newcommand{\Scal}{\mathcal{S}}

\newcommand{\Z}{\mathbb{Z}}
\newcommand{\R}{\mathbb{R}}
\newcommand{\C}{\mathbb{C}}
\newcommand{\F}{\mathbb{F}}
\newcommand{\E}{\mathbb{E}}
\newcommand{\N}{\mathbb{N}}

\newcommand{\gb}{\bold{g}}
\newcommand{\gt}{\tilde{g}}

\newcommand{\gh}{\hat{g}}
\newcommand{\Pb}{\bold{P}}
\newcommand{\Sb}{\bold{S}}

\newcommand{\Sbwt}{\widetilde{\Sb}}

\newcommand{\SbPi}{\Sb_{\Pibold}}
\newcommand{\SbPh}{\Sb_{\Pibh}}
\newcommand{\SbPt}{\Sb_{\Pibt}}
\newcommand{\Ab}{\bold{A}}
\newcommand{\Acal}{\mathcal{A}}

\newcommand{\Abh}{\widehat{\Ab}}
\newcommand{\Abt}{{\tilde{\bold{A}}}}

\newcommand{\Abc}{{\check{\bold{A}}}}
\newcommand{\Abg}{{\grave{\bold{A}}}}

\newcommand{\bb}{\bold{b}}

\newcommand{\bbh}{\widehat{\bold{b}}}
\newcommand{\bbt}{\tilde{\bold{b}}}

\newcommand{\bbbr}{\breve{\bold{b}}}
\newcommand{\bbc}{{\check{\bold{b}}}}
\newcommand{\bbg}{{\grave{\bold{b}}}}
\newcommand{\Bb}{\bold{B}}

\newcommand{\Bcal}{\mathcal{B}}

\newcommand{\ab}{\bold{a}}
\newcommand{\abt}{\tilde{\bold{a}}}
\newcommand{\abg}{\grave{\bold{a}}}
\newcommand{\Cb}{\bold{C}}

\newcommand{\Db}{\bold{D}}

\newcommand{\Eb}{\bold{E}}
\newcommand{\eb}{\bold{e}}
\newcommand{\Ib}{\bold{I}}
\newcommand{\Hb}{\bold{H}}
\newcommand{\Hbh}{\hat{\bold{H}}}
\newcommand{\Hbt}{\tilde{\bold{H}}}
\newcommand{\Ht}{\tilde{H}}
\newcommand{\Ubt}{\tilde{\bold{U}}}
\newcommand{\Vbh}{\hat{\bold{V}}}
\newcommand{\Sigb}{\bold{\Sigma}}

\newcommand{\err}{\mathrm{err}}
\newcommand{\image}{\mathrm{im}}

\newcommand{\rpm}{\raisebox{.2ex}{$\scriptstyle\pm$}}
\newcommand{\getsU}{\overset{\ U}{\gets}}

\newcommand{\sfsty}[1]{\ensuremath{\mathsf{#1}}}  % sf-style
\newcommand{\Enc}{\sfsty{Enc}}
\newcommand{\Dec}{\sfsty{Dec}}

\DeclarePairedDelimiter\ceil{\lceil}{\rceil}

\DeclareMathOperator{\GL}{GL}
\DeclareMathOperator{\tr}{tr}

\newcommand{\diag}{\mathrm{diag}}

\newtheorem{Thm}{Theorem}
\newtheorem{Cor}{Corollary}
\newtheorem{Prop}{Proposition}
\newtheorem{Lemma}{Lemma}
\newtheorem{Def}{Definition}
\newtheorem{Rmk}{Remark}

\newcommand{\ind}{\text{\color{white}.$\quad$}}

\begin{document}
\title{Iterative Sketching for Secure Coded Regression}

\author{\large%
Neophytos Charalambides, Hessam Mahdavifar, Mert Pilanci, and Alfred O. Hero III
  \vspace{-.25in}
  \thanks{Part of the material in this paper was presented at the 2022 IEEE International Symposium on Information Theory (ISIT), Espoo, Finland, June 2022 \cite{CMPH22}.
  \newline
  ${\white..}$ The authors N.C. and A.H. are with the Department of Electrical Engineering and Computer Science, University of Michigan, Ann Arbor, MI 48104 USA (email: neochara@umich.edu, hero@umich.edu). The author H.M. is with the Department of Electrical and Computer Engineering, Northeastern University, Boston, MA 02115 (email: h.mahdavifar@northeastern.edu). The author M.P. is with the Department of Electrical Engineering, Stanford University, Stanford, CA 94305 USA (e-mail: pilanci@stanford.edu).}
}

\maketitle

\begin{abstract}
Linear regression is a fundamental and primitive problem in supervised machine learning, with applications ranging from epidemiology to finance. In this work, we propose methods for speeding up distributed linear regression. We do so by leveraging randomized techniques, while also ensuring security and straggler resiliency in asynchronous distributed computing systems. Specifically, we randomly rotate the basis of the system of equations and then subsample \textit{blocks}, to simultaneously secure the information and reduce the dimension of the regression problem. In our setup, the basis rotation corresponds to an encoded encryption in an \textit{approximate gradient coding scheme}, and the subsampling corresponds to the responses of the non-straggling servers in the centralized coded computing framework. This results in a distributive \textit{iterative} stochastic approach for matrix compression and steepest descent.
\end{abstract}


% - - - - - - - - - - - - - - -
\section{Introduction}
\label{intro}

Applying a random projection\footnote{By `\textit{projections}' we refer to random matrices, not idempotent matrices.} is a classical way of performing dimensionality reduction, and is widely used in algorithmic and learning contexts \cite{Vem05,Woo14,DMMS11,DM16}. Distributed computations in the presence of \textit{stragglers} have gained a lot of attention in the information theory community. Coding-theoretic approaches have been adopted for this \cite{LLPPR17,reisizadeh2017coded,li2016coded,li2017coding,LSR17,dutta2016short,ramamoorthy2019universally,YSRKSA18,RRG20,CT19,CPH20c,CPH20b,OUG20,OBGU20,CMH21,RCHV23,CPH23}, and fall under the framework of \textit{coded computing} (CC). Data security is also an increasingly important issue in CC \cite{LA20}. The study of privacy and security in distributed computing networks, dates back to the notion of ``\textit{secure multi-party computation}'' \cite{BGV88,CCD88}. In this work, we propose methods to securely speed up distributed linear regression; a fundamental problem in supervised learning, by simultaneously leveraging random projections and sketching, and distributed computations.

Our results are presented in terms of the system model proposed in \cite{LLPPR17}; illustrated in Figure \ref{CMM_illustration}, though they extend to any centralized distributed model, \textit{i.e.} distributed systems with a central server which updates and communicates the parameters to computational nodes; referred to as servers. Furthermore, the desiderata of our approaches are to:
\begin{enumerate}[label=(\Roman*)]
  \item produce straggler-resilient \textit{accurate approximations}, in terms of the resulting $\ell_2$ error,
  \item \textit{secure} the data,
  \item carry out the computations \textit{efficiently} and \textit{distributively}.
\end{enumerate}
These properties will be clarified in Subsection \ref{approach_properties_subsec}.

\begin{figure}[h]
  \centering
    \includegraphics[scale=.14]{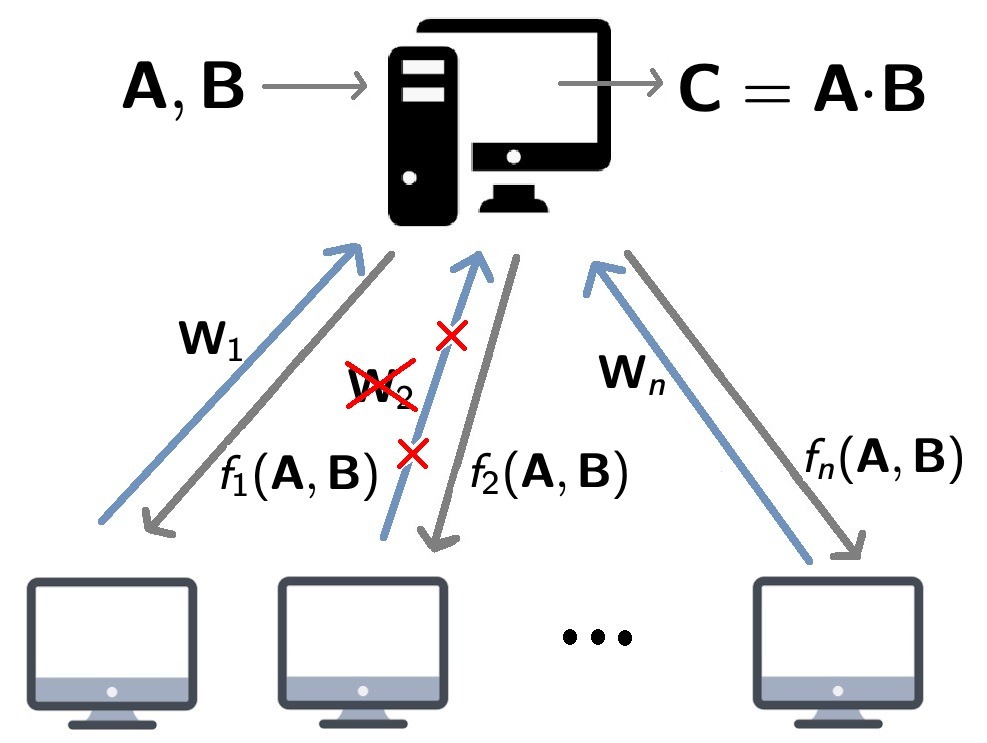}
    \caption{Diagram of an exact coded matrix multiplication scheme, under the model proposed in \cite{LLPPR17}. The central server sends an encoding $f_i(\Ab,\Bb)$ of matrices $\Ab,\Bb$; usually of their submatrices, to each of the $n$ computational nodes. Then, each server performs a computation on their encoded data to obtain $\Wb_i$, which are then delivered to the central server. Once a certain fraction of servers have responded, the central server obtains $\Cb=\Ab\cdot\Bb$, after performing a decoding step. In this example, the second server is a straggler.}
  \label{CMM_illustration}
\end{figure}

We focus on sketching for \textit{steepest descent} (SD) in the context of solving overdetermined linear systems. Part of our theoretical results are given for the sketch-and-solve paradigm \cite{Sar06}, which can be utilized by a single server; to also store a compressed version of the data. The application through CC results in iterative sketching methods for distributively solving overdetermined linear systems. We propose applying a random orthonormal projection to the linear system before distributing the data, and then performing SD distributively on the transformed system through \textit{approximate gradient coding}. Under the straggler scenario and the assumptions we make, this results in a \textit{mini-batch stochastic steepest descent} (SSD) procedure of the original system. A special case of such a projection is the \textit{Subsampled Randomized Hadamard Transform} (SRHT) \cite{DMMS11,Tro11,BG13}, which utilizes recursive Kronecker products of the Hadamard transform; and relates to the \textit{fast Johnson-Lindenstrauss transform} \cite{AC06,JL84}.

The benefit of applying an orthonormal matrix transformation is that we rotate and/or reflect the data's orthonormal basis, which \textit{cannot} be reversed without knowledge of the matrix. This is leveraged to give security guarantees,  while simultaneously ensuring that we recover well-approximated gradients, and an approximate solution of the linear system. Such sketching matrices are also referred to as \textit{randomized orthonormal systems} \cite{BP23}. We also discuss how one can use recursive Kronecker products of an orthonormal matrix of dimension greater than 2 in place of the Hadamard transform, to obtain a more efficient encoding and encryption than through a random and unstructured orthonormal matrix.

In the CC paradigm, the servers are assumed to be \textit{homogeneous} with the same expected response time. In the proposed methods, we stop receiving computations once a fixed fraction of the servers respond; producing a different induced sketch at each iteration. A predominant task which has been studied in the CC framework is the gradient computation of differentiable and additively separable objective functions \cite{TLDK17,HASH17,OGU19,CMH20,YA18,CT22,RTTD17,CPE17,KKR19,BWE19,SH22,GW21,CP18,WCP19,WLS19,HYKM19,CHZP18,CPH20a}. These schemes are collectively called \textit{gradient coding} (GC)\footnote{We abbreviate ‘gradient coding scheme/schemes’ to GCS/GCSs.}. We note that iterative sketching has proven to be a powerful tool for second-order methods \cite{PW16,LLDP20}, though it has not been explored in first-order methods. Since we indirectly consider a different underlying modified problem at each iteration, the methods we propose are \textit{approximate} GCSs. Related approaches have been proposed in \cite{RTTD17,KKR19,BWE19,GW21,SH22,CPE17,CP18,WCP19,WLS19,HYKM19,CHZP18,CPH20a}. Two important benefits of our approaches are that we do not require a decoding step, nor an encoding step by the servers; at each iteration, avoiding prevalent bottlenecks.

An advantage of using an updated sketch at each iteration, is that we do not have a bias towards the samples that would be selected/returned when the sketching takes place, which occurs in the sketch-and-solve approach. Specifically, we do not solve a modified problem which only accounts for a reduced dimension; determined at the beginning of the iterative process. Instead, we consider a different reduced system at each iteration. This is also justified numerically.

Another benefit of our approach, is that random projections secure the information from potential eavesdroppers, honest but curious; and colluding servers. We show information-theoretic security for the case where a random orthonormal projection is utilized in our sketching algorithm. Furthermore, the security of the SRHT, which is a crucial aspect, has not been extensively studied. Unfortunately, the SRHT is inherently insecure, which we show. We propose a modified projection which guarantees computational security of the SRHT.

There are related works to what we study. The work of \cite{BP23} focuses on parameter averaging for variance reduction, but only mentions a security guarantee for the Gaussian sketch, derived in \cite{ZWL08}. Another line of work \cite{KSD17,KSDY19}, focuses on introducing redundancy through equiangular tight frames (ETFs), partitioning the system into smaller linear systems, and then averaging the solutions of a fraction of them. A drawback of using ETFs, is that most of them are over $\C$. The authors of \cite{SKD18} study privacy of random projections, though make the assumption that the projections meet the `$\varepsilon$-MI-DP constraint'. Recently, the authors of \cite{LSM23} considered CC privacy guarantees through the lens of differential privacy, with a focus on matrix multiplication. Lastly, a secure GCS is studied in \cite{YA19}, though it does not utilize sketching. We also clarify that even though we guarantee cryptographic security, our methods may still be vulnerable to various privacy attacks, \textit{e.g.} membership inference attacks \cite{FJR15} and model inversion attacks \cite{SSSS17}. This is another interesting line of work, though is not a focus of our approach.

The paper is organized as follows. In \ref{coded_LR} we review the framework and background for coded linear regression, the notions of security we will be working with, the $\ell_2$-subspace embedding property, and list the main properties we seek to satisfy through our constructions; in order to meet the aforementioned desiderata. In \ref{bl_orth_sk_sec} we present the proposed iterative sketching algorithm, and in \ref{block_SRHT_sec} the special case where the projection is the randomized Hadamard transform; which we refer to as the ``\textit{block-SRHT}''. The subspace embedding results for the general algorithm and the block-SRHT are presented in the respective sections. We consider the case where the central server may adaptively change the step-size of its SD procedure in \ref{opt_ss_sec}, which can be viewed as an \textit{adaptive GCS}. In \ref{security_sec} we present the security guarantees of our algorithm and the modified version of the block-SRHT, as well as how our idea of encryption can be utilized in \textit{any} GCS. Finally, we present numerical experiments in \ref{exper_sec}, concluding remarks in \ref{concl_sec}, and how the proposed techniques can be utilized for logistic regression and other computations in Appendix \ref{orth_encr_distr_tasks}.% Due to the page constraint, we highlight the key technique or results used to prove our main statements. Most lemmas and corollaries are proven through algebraic manipulations.

% - - - - - - - - - - - - - - -
\section{Coded Linear Regression}
\label{coded_LR}

\subsection{Least Squares Approximation and Steepest Descent}
\label{LR_SD}

In linear least squares approximation \cite{DMMS11}, it is desired to approximate the solution
\begin{equation}
\label{x_star_pr_lr}
  \xb_{ls}^{\star} = \arg\min_{\xb\in\R^d}\Big\{L_{ls}(\Ab,\bb;\xb)\coloneqq\|\Ab\xb-\bb\|_2^2\Big\}
\end{equation}
where $\Ab\in\R^{N\times d}$ and $\bb\in\R^N$. This corresponds to the regression coefficients $\xb$ of the model $\bb=\Ab\xb+\vec{\varepsilon}$, which is determined by the dataset $\D=\left\{(\ab_i,b_i)\right\}_{i=1}^N\subsetneq \R^d\times\R$ of $N$ samples, where $(\ab_i,b_i)$ represent the features and label of the $i^{th}$ sample, \textit{i.e.} $\Ab=\big[\ab_1 \ \cdots \ \ab_N \big]^T$ and $\bb=\big[b_1 \ \cdots \ b_N \big]^T.$

To simplify our presentation, we first define some notational conventions. Row vectors of a matrix $\Mb$ are denoted by $\Mb_{(i)}$, and column vectors by $\Mb^{(j)}$. Our embedding results are presented in terms of an arbitrary partition $\N_N=\bigsqcup_{\iota=1}^K\K_\iota$, for $\N_N\coloneqq\{1,\dots,N\}$ the index set of $\Mb$'s rows. Each $\K_\iota$ corresponds to a \textit{block} of $\Mb$. The notation $\Mb_{(\K_\iota)}$ denotes the submatrix of $\Mb$ comprised of the rows indexed by $\K_\iota$. That is: $\Mb_{(\K_\iota)}=\Ib_{(\K_\iota)}\cdot\Mb$, for $\Ib_{(\K_\iota)}$ the corresponding submatrix of $\Ib_N$ of size $|\K_\iota|\times N$. We call $\Mb_{(\K_\iota)}$ the `$\iota^{th}$ \textit{block of} $\Mb$'. We abbreviate $(1-\epsilon)\cdot b\leqslant a \leqslant(1+\epsilon)\cdot b$, to $a\leqslant_\epsilon b$. The $k$-fold Kronecker product of $\Mb$ with itself is denoted by $\Mb^{\otimes k}=\Mb\otimes\cdots\otimes\Mb$. Lastly, `$\gets$' denotes a numerical assignment of a varying quantity, and `$\getsU$' a realization of a random variable through uniform sampling.

The orthogonal group in dimension $N$, is denoted by $O_N(\R)$, and consists of all orthonormal matrices in $\R^{N\times N}$. By $\Pibold$ we denote the random orthonormal matrix that is applied to the data matrix $\Ab$; which is drawn uniformly at random from a finite subgroup $\Otil_\Ab$ of $O_N(\R)$, \textit{i.e.} $\Pibold\getsU \Otil_\Ab\subsetneq O_N(\R)$. By $\Pibh$ and $\Pibt$ we denote the special cases where $\Pibold$ is an orthonormal matrix used for the \textit{block-SRHT} and \textit{garbled block-SRHT} respectively.

We address the overdetermined case where $N\gg d$. Existing exact methods find a solution vector $\xb_{ls}^{\star}$ in $\ow(Nd^2)$ time, where $\xb_{ls}^{\star}=\Ab^{\dagger}\bb$. A common way to approximate $\xb_{ls}^{\star}$ is through SD, which iteratively updates the gradient
\begin{equation}
\label{gr_ls}
  g_{ls}^{[t]}\coloneqq\nabla_{\xb}L_{ls}(\Ab,\bb;\xb^{[t]}) = 2\Ab^T(\Ab\xb^{[t]}-\bb)
\end{equation}
followed by updating the parameter vector
\begin{equation}
\label{par_upd}
  \xb^{[t+1]}\gets\xb^{[t]}-\xi_t\cdot g_{ls}^{[t]}.
\end{equation}
In our setting, the step-size $\xi_t>0$ is determined by the central server. The script $[t]$ indexes the iteration $t=0,1,2,\ldots$ which we drop when clear from the context. In \ref{opt_ss_sec}, we derive the optimal step-size $\xi_t^{\star}$ for \eqref{gr_ls} and the modified least squares problems we consider, given the updated gradients and parameters at iteration $t$.

\subsection{The Straggler Problem and Gradient Coding}
\label{str_problem}

Gradient coding is deployed in centralized computation networks, \textit{i.e.} a central server communicates $\xb^{[t]}$ to $m$ computational nodes; who perform computations and then communicate back their results. The central server distributes the dataset $\D$ among the $m$ nodes, to facilitate the solution of optimization problems with additively separable and differentiable objective functions. For linear regression \eqref{x_star_pr_lr}, the data is partitioned as
\begin{equation}
\label{part_data}
  \Ab=\Big[\Ab_1^T \ \cdots \ \Ab_K^T\Big]^T \quad \text{ and } \quad \bb=\Big[\bb_1^T \ \cdots \ \bb_K^T\Big]^T
\end{equation}
where $\Ab_i\in\R^{\tau\times d}$ and $\bb_i\in\R^{\tau}$ for all $i$, and $\tau=N/K$. For ease of exposition, we assume that $K|N$. Then, we have $L_{ls}(\Ab,\bb;\xb) = \sum_{i=1}^KL_{ls}(\Ab_i,\bb_i;\xb)$. A regularizer $\mu R(\xb)$ can also be added to $L_{ls}(\Ab,\bb;\xb)$ if desired.

In GC \cite{TLDK17}, the servers encode their computed \textit{partial gradients} $g_i\coloneqq\nabla_{\xb}L_{ls}(\Ab_i,\bb_i;\xb)$; which are then communicated to the central server. Once a certain fraction of encodings is received, the central server applies a decoding step to recover the gradient $g=\nabla_{\xb}L_{ls}(\Ab,\bb;\xb)=\sum_{i=1}^K g_i$. This can be computationally prohibitive, and takes place at every iteration. To the best of our knowledge, the lowest decoding complexity is $\ow\left((s+1)\cdot\ceil{\frac{m}{s+1}}\right)$; where $s$ is the number of stragglers \cite{CMH20}.

In our approach we trade time; by not requiring encoding nor decoding steps at each iteration, with accuracy of approximating $\xb_{ls}^{\star}$. Unlike conventional GCSs, in this paper the servers carry out the computation on the encoded data. The resulting gradient, is that of the modified least squares problem
\begin{equation}
\label{x_til_pr_lr}
  \xbh_{ls} = \arg\min_{\xb\in\R^d}\Big\{L_{\Sb^{[t]}}(\Ab,\bb;\xb)\coloneqq\big\|\Sb^{[t]}(\Ab\xb-\bb)\big\|_2^2\Big\}\
\end{equation}
for $\Sb^{[t]}\in\R^{r\times N}$ a sketching matrix, with $r\ll N$ and $r>d$. This is the core idea behind our approximation, in which we incorporate iterative sketching with orthonormal matrices and random sampling; and generalizations of the SRHT for $\Sb^{[t]}$, for our approximate GCSs. The sketching approach we take is to first apply a random projection, which also provides security against the servers and eavesdroppers, and then sample computations carried out on the blocks of the transformed data uniformly at random; which corresponds to the responses of the homogeneous non-stragglers. Furthermore, uniform sampling introduces the randomness in our proposed procedure, in such a way that the resulting GCS is a distributed SSD algorithm.

For $q$ the total number of responsive servers, we can mitigate up to $s=m-q$ stragglers. Specifically, the number of responsive servers $m-s$ in the CC model, corresponds to the number of sampling trials $q$ of our sketching algorithm, \textit{i.e.} $q=m-s$. At iteration $t$, a SD update of the modified least squares problem \eqref{x_til_pr_lr} is obtained distributively. Furthermore, we assume that the data is partitioned into as many blocks as there are servers, \textit{i.e.} $K=m$. The stragglers are assumed to be uniformly random and likely differ at each iteration. Thus, there is a different sketching matrix $\Sb^{[t]}$ at each epoch.

In conventional GCSs the objective is to construct an encoding matrix $\Gb\in\R^{m\times K}$ (can have $m\neq K$) and decoding vectors $\ab_\I\in\R^{1\times q}$, such that $\ab_\I\Gb_{(\I)}=\vec{\bold{1}}$ for any set of non-straggling servers $\I$. Furthermore, it is assumed that multiple replications of each encoded block are shared among the servers, such that $m\gneq q\geqslant K$.\footnote{As we mention in \ref{distr_grad_desc}, this can be done in order to mimic sampling \textit{with replacement} through the CC network. The reason we require $q\geqslant K$, is to define $\ab_\I^\star = \vec{\bold{1}}\Gb_{(\I)}^{\dagger}$. This idea has been extensively studied in \cite{CPH23b}.} From the fact that $\ab_\I\Gb_{(\I)}=\vec{\bold{1}}$ for any $\I$, in \textit{approximate} GC \cite{CPE17}, the optimal decoding vector for a set $\I$ of size $q=m-s$ is determined by
\begin{equation}
\label{opt_dec_vec}
  \ab_\I^\star = \arg\min_{\ab\in\R^{1\times q}}\big\{\big\|\ab\Gb_{(\I)}-\vec{\bold{1}}\big\|_2^2\big\} \quad \ \implies \quad \ab_\I^\star = \vec{\bold{1}}\Gb_{(\I)}^{\dagger}.
\end{equation}
The error in the approximated gradient $\grave{g}^{[t]}$ of an optimal approximate linear regression GCS $(\Gb,\ab_\I^\star)$, is then%for $\Gb_{(\I)}^{\dagger}$ the pseudoinverse of $\Gb_{(\I)}$.
\begin{equation}
\label{appr_GC_error}
  \big\|g^{[t]}-\grave{g}^{[t]}\big\|_2 \leqslant 2\sqrt{K}\cdot\err(\Gb_{(\I)})\cdot\|\Ab\|_2\cdot\|\Ab\xb^{[t]}-\bb\|_2,
\end{equation}
for $\err(\Gb_{(\I)}) \coloneqq \big\|\Ib_K-\Gb_{(\I)}^{\dagger}\Gb_{(\I)}\big\|_2$.

At this point, it is worth mentioning that there are similar approximate GCSs in terms of decoding, \textit{i.e.} which consider a fixed decoding vector at every iteration \cite{RTTD17,CPE17,KKR19,BWE19,SH22}, usually $\abt=\gamma\vec{\bold{1}}$ for all valid $\I$; where $\gamma$ is an appropriate constant. In our setting, we have $\gamma=1$. Similar to our GCS, these works use randomized constructions for their encoding matrix. Such schemes are referred to as \textit{fixed coefficient decoding}, while most other approximate GCSs are known as \textit{optimal coefficient decoding} \cite{GW21}, since they approximate or explicitly solve \eqref{opt_dec_vec}.

\subsection{Secure Coded Computing Schemes}
\label{security_subs}
Modern cryptography is split into two main categories, \textit{information-theoretic security} and \textit{computational security}. These categories are also referred to as \textit{Shannon secrecy} and \textit{asymptotic security} respectively. In this subsection, we give the definitions which will allow us to characterize the security level of our GCSs.

\begin{Def}[Ch.2 \cite{KL14}]
\label{Sh_secr}
An encryption scheme $\Enc$ with message, ciphertext and key spaces $\M$, $\Cc$ and $\K$ respectively, is \textbf{Shannon secret} w.r.t. a probability distribution $D$ over $\M$, if for all $\bar{m}\in\M$ and all $\bar{c}\in\Cc$:
\begin{equation}
  \Pr_{{\substack{m\sim D\\ k\getsU\K}}}\left[m=\bar{m}\mid\Enc_k(m)=\bar{c}\right] = \Pr_{m\sim D}\left[m=\bar{m}\right] .
\end{equation}
An equivalent condition is \textbf{perfect secrecy}, which states that for all $m_0,m_1\in\M$:
\begin{equation}
\label{perf_secrecy_id}
  \Pr_{k\getsU\K}\left[\Enc_k(m_0)=\bar{c}\right] = \Pr_{k\getsU\K}\left[\Enc_k(m_1)=\bar{c}\right] .
\end{equation}
\end{Def}

\begin{Def}[Ch.3 \cite{KL14}]
\label{comp_sec}
An encryption scheme $\Enc$ is \textbf{computationally secure}, if any probabilistic polynomial-time adversary $\Acal$ with decryption scheme $\Dec_\Acal$ succeeds in obtaining the message $m\in\M$ of length $n$ from the encrypted message $\Enc(m)$, with at most negligible\footnote{A function $\nu(n)$ is \textit{negligible} if it vanishes faster than the inverse of any polynomial in $n$, \textit{i.e.} $\lim_{n\to\infty}\nu(n)\cdot n^c=0$ for every constant $c>0$.} probability. That is, for any $m\in\M$, the decryption success probability of $\Acal$ is $\Pr\left[\Dec_\Acal(\Enc(m))=m\right]\leqslant\frac{1}{\mathrm{poly}(n)}$, where $\mathrm{poly}(\cdot)$ represents any positive polynomial.
\end{Def}
\begin{comment}
\begin{Def}[Ch.3 \cite{KL14}]
\label{comp_sec}
An encryption scheme is \textbf{computationally secure} if any probabilistic polynomial-time adversary succeeds in obtaining the information which was encrypted by the scheme, with at most negligible probability. By negligible, we mean it is asymptotically smaller than any inverse polynomial function.
\end{Def}
\end{comment}

\begin{Def}
\label{secure_CC_def}
A \textbf{secure CC scheme}, is a pair of CC encoding and decoding algorithms $(\Enc,\Dec)$, for which $\Enc(\Ab)$ also guarantees Shannon secrecy or computational security of $\Ab$, and $\Dec$ recovers the data which was encrypted; \textit{i.e.} $\Dec(\Enc(\Ab))=\Ab$.
\end{Def}

In our work, $\Enc$ corresponds to a linear transformation through a randomly selected orthonormal matrix $\Pibold$. By encryption, we refer to this linear transformation, which is utilized in our GCS. Furthermore, we do not require a decryption step by the central server, as it computes an unbiased estimate of the gradient at the end of each iteration. Also, since $\Pibold^T=\Pibold^{-1}$, it follows that $\Pibold^T$ meets the requirement of $\Dec$, so in the following definition we refer to the encoding-decoding pair by only referencing $\Enc$. Furthermore, $\Enc$ depends on a secret key $k$ which is randomly generated. In our case, this is simply $\Pibold$.

\subsection{The $\ell_2$-subspace embedding Property}
\label{embedding_subs}

For the analysis of the sketching matrices $\SbPi$ we propose in Algorithm \ref{alg_orthog_sketch}, we consider any orthonormal basis $\Ub\in\R^{N\times d}$ of the column-space of $\Ab$, \textit{i.e.} $\image(\Ab)=\image(\Ub)$. The subscript of $\SbPi$, indicates the dependence of the sketching matrix on $\Pibold$.

Recall that the \textit{$\ell_2$-subspace embedding} ($\ell_2$-s.e.) property \cite{Woo14,ERNM22} states that any $\yb\in\image(\Ub)$ satisfies with high probability:
\begin{align}
\label{subsp_emb_id}
  \big\|\SbPi\yb\big\|_2\leqslant_\epsilon\|\yb\|_2  \ \iff \ \big\|\Ib_d-(\SbPi\Ub)^T(\SbPi\Ub)\big\|_2\leqslant \epsilon
\end{align}
for $\epsilon>0$. In turn, this characterizes the approximation's error of the solution $\xbh_{ls}$ of \eqref{x_til_pr_lr} for $\Sb\gets\SbPi$, as
$$ \big\|\Ab\xbh_{ls}-\bb\big\|_2 \leqslant \frac{1+\epsilon}{1-\epsilon}\big\|\Ab\xb_{ls}^{\star}-\bb\big\|_2 \leqslant (1+\ow(\epsilon))\big\|\Ab\xb_{ls}^{\star}-\bb\big\|_2 $$
w.h.p., and $\big\|\Ab(\xb_{ls}^{\star}-\xbh_{ls})\big\|_2\leqslant\epsilon\big\|(\Ib_N-\Ub\Ub^T)\bb\big\|_2$.

\subsection{Properties of our Approach}
\label{approach_properties_subsec}

A key property in the construction of our sketching matrices, is to sample \textit{blocks} (\textit{i.e.} submatrices) of a transformation of the data matrix, which permits us to then perform the computations in parallel. The additional properties we seek to satisfy with our GCSs through block sampling are the following:
\begin{enumerate}[label=(\alph*)]%[label=(\roman*)]
  \item the underlying sketching matrix satisfies the $\ell_2$-s.e. property,
  \item the block leverage scores are flattened through the random projection $\Pibold$,
  \item the projection is over $\R$,
  \item the central server computes an unbiased gradient estimate at each iteration,
  \item do not require encoding/decoding at each iteration,
  \item guarantee security of the data from the servers and potential eavesdroppers,
  \item $\Pibold$ can be applied efficiently, \textit{i.e.} in $\ow(Nd\log N)$ operations.
\end{enumerate}

The seven properties listed above, are grouped together with respect to the desiderata mentioned in Section \ref{intro}. Specifically, desideratum (I) encompasses properties (a), (b), (c), (d), desideratum (II) corresponds to (f), and (III) encompasses (b), (c), (e), (g).

Property (a) is motivated by the sketch-and-solve approach, though through the iterative process, in practice we benefit by having fresh sketches. Leverage scores define the key structural non-uniformity that must be dealt with in developing fast randomized matrix algorithms; and are formally defined in \ref{subsp_emb_alg1}. If property (b) is met, we can then sample uniformly at random in order to guarantee (a). We require $\Pibold$ to be over $\R$, as if it were over $\C$, the communication cost from the central server to the servers; and the necessary storage space at the servers would double. Additionally, performing computations over $\C$ would result in further round-off errors and numerical instability. Properties (d) and (e) are met by requiring $\Pibold$ to be an orthonormal matrix. By allowing the projection to be random; we can secure the data, \textit{i.e.} satisfy (f). Furthermore, the action of applying an orthonormal projection for our encryption; is reversed through the computation of the partial gradients, hence no decryption step is required.

By considering a larger ensemble of orthonormal projections to select from, we can give stronger security guarantees. Specifically, by not restricting the structure of $\Pibold$, we can guarantee Shannon secrecy, though this prevents us from satisfying (g). On the other hand, if we let $\Pibold$ be structured, we can satisfy (g) at the cost of guaranteeing computational security; and not perfect secrecy.

We point out that even though Gaussian and random Rademacher sketches satisfy satisfy (a), (b), (c) and (f), they do not satisfy (d), (e) nor (g) in our CC setting. Experimentally, we observe that our proposed sketching matrices outperform their Gaussian and random Rademacher counterparts, primarily due to the fact that (d) is satisfied. Furthermore, for $\Pibold\in O_N(\R)$, our distributive procedure results in a SSD approach.

% - - - - - - - - - - - - - - -
\section{Block Subsampled Orthonormal Sketches}
\label{bl_orth_sk_sec}

Sampling blocks for sketching least squares has not been explored as extensively as sampling rows, though there has been interest in using ``block-iterative methods'' for solving systems of linear equations \cite{Elf80,Gut06,NT14,RN20}. Our interest in sampling blocks, is to invoke results and techniques from \textit{randomized numerical linear algebra} (RandNLA) to CC. Specifically, we apply the transformation before partitioning the data and sharing it between the servers, who will compute the respective partial gradients. Then, the slowest $s$ servers will be disregarded. The proposed sketching matrices are summarised in Algorithm \ref{alg_orthog_sketch}.

\begin{algorithm}[h]
\label{alg_orthog_sketch}
\SetAlgoLined
%{\small%\footnotesize
  \KwIn{$\Ab\in\R^{N\times d}$, $\bb\in\R^N$, $\xb^{[0]}\in\R^d$, $\tau=\frac{N}{K}$, $q=\frac{r}{\tau}>\frac{d}{\tau}$}
  \KwOut{approximate solution $\xbh\in\R^d$ to \eqref{x_star_pr_lr}}
  \textbf{Randomly Select:} $\Pibold\in O_N(\R)$, an orthonormal matrix\\
  \For{$t=0,1,2,\ldots$}
    {
    \textbf{Initialize:} $\Omb=\bold{0}_{q\times K}$\\
    \textbf{Select:} step-size $\xi_t>0$\\
    \For{$i=1$ to $q$}
      {
        uniformly sample with replacement $j_i$ from $\N_K$\\
        $\Omb_{i,j_i}=\sqrt{N/r}=\sqrt{K/q}$
      }
    $\Ombwt_{[t]}\gets\Omb\otimes\Ib_\tau$ \Comment{$\SbPi^{[t]}=\Ombwt_{[t]}\cdot\Pibold$}\\
    $\Abh_{[t]}\gets\Ombwt_{[t]}\cdot(\Pibold\Ab)=\SbPi^{[t]}\cdot\Ab$\\
    $\bbh_{[t]}\gets\Ombwt_{[t]}\cdot(\Pibold\bb)=\SbPi^{[t]}\cdot\bb$\\
    \textbf{Update:} $\xbh^{[t+1]}\gets\xbh^{[t]}-\xi_t\cdot\nabla_{\xb}L_{ls}\big(\Abh_{[t]},\bbh_{[t]};\xbh^{[t]}\big)$
    }
%}
\caption{Iterative Orthonormal Sketching}
\end{algorithm}

To construct the sketch $\Abh$, we first transform its orthonormal basis $\Ub$ by applying $\Pibold$ to $\Ab$. Then, we subsample $q$ blocks from $\Pibold\Ab$, to reduce dimension $N$. Finally, we normalize by $\sqrt{N/r}$ to reduce the variance of the estimator $\Abh$. Analogous steps are carried out on $\Pibold\bb$, to construct $\bbh$.

% - - - - - - - - - - - - - - -
\subsection{Distributed Steepest Descent and Iterative Sketching}
\label{distr_grad_desc}

We now discuss the servers' computational tasks of our proposed GCS, when SD is carried out distributively. The encoding corresponds to $\Abt=\Gb\cdot\Ab$ and $\bbt=\Gb\cdot\bb$ for $\Gb\coloneqq\sqrt{N/r}\cdot\Pibold$, which are then partitioned into $K$ encoded block pairs $(\Abt_i,\bbt_i)$; similar to \eqref{part_data}, and are sent to distinct servers. Specifically, $\Abt_i=\Ib_{(\K_i)}\cdot\Abt$ and $\bbt_i=\Ib_{(\K_i)}\cdot\bbt$. This differs from most GCSs, in that the encoding is usually done locally by the servers on the computed results; at each iteration.

If each server respectively computes $\nabla_{\xb}L_{ls}(\Abt_i,\bbt_i;\xb^{[t]})=2\Abt_i^T(\Abt_i^T\xb^{[t]}-\bbt_i)$ at iteration $t$, and the index set of the first $q$ responsive servers is $\Scal^{[t]}$, the aggregated gradient
\begin{equation}
\label{gr_update}
  \gh^{[t]} = 2\sum\limits_{j\in\Scal^{[t]}}\Abt_j^T\left(\Abt_j\xb^{[t]}-\bbt_j\right)
\end{equation}
is equal to the gradient of $L_{\Sb}(\Ab,\bb;\xb^{[t]})$ for $\Sb\gets\SbPi^{[t]}$ the induced sketching matrix at that iteration, \textit{i.e.} $\gh^{[t]}=\nabla_{\xb}L_{\SbPi^{[t]}}(\Ab,\bb;\xb^{[t]})$. The sampling matrix $\Ombwt_{[t]}$ and index set $\Scal^{[t]}$, correspond to the $q$ responsive servers. We clarify that the dimensionality reduction of Algorithm \ref{alg_orthog_sketch} is a direct consequence of performing sampling through $\Ombwt_{[t]}$, which reduces the linear system from $N$ equations to $r$. Consequently, at each iteration of our distributed approach, we are computing the gradient of the objective function corresponding to the smaller linear system $\Abh_{[t]}\xb=\bbh_{[t]}$. Another benefit of considering uniform sampling through $\Ombwt_{[t]}$ which changes at each iteration, is that our procedure results in a distributed SSD algorithm (Theorem \ref{GC_SGD_thm}). Our procedure is illustrated in Figure \ref{it_sketch_figure}.% \textit{i.e.} $\nabla_\xb  L_{ls}(\Abh_{[t]},\bbh_{[t]};\xb)$

\begin{figure}[h]
  \centering
    \includegraphics[scale=.17]{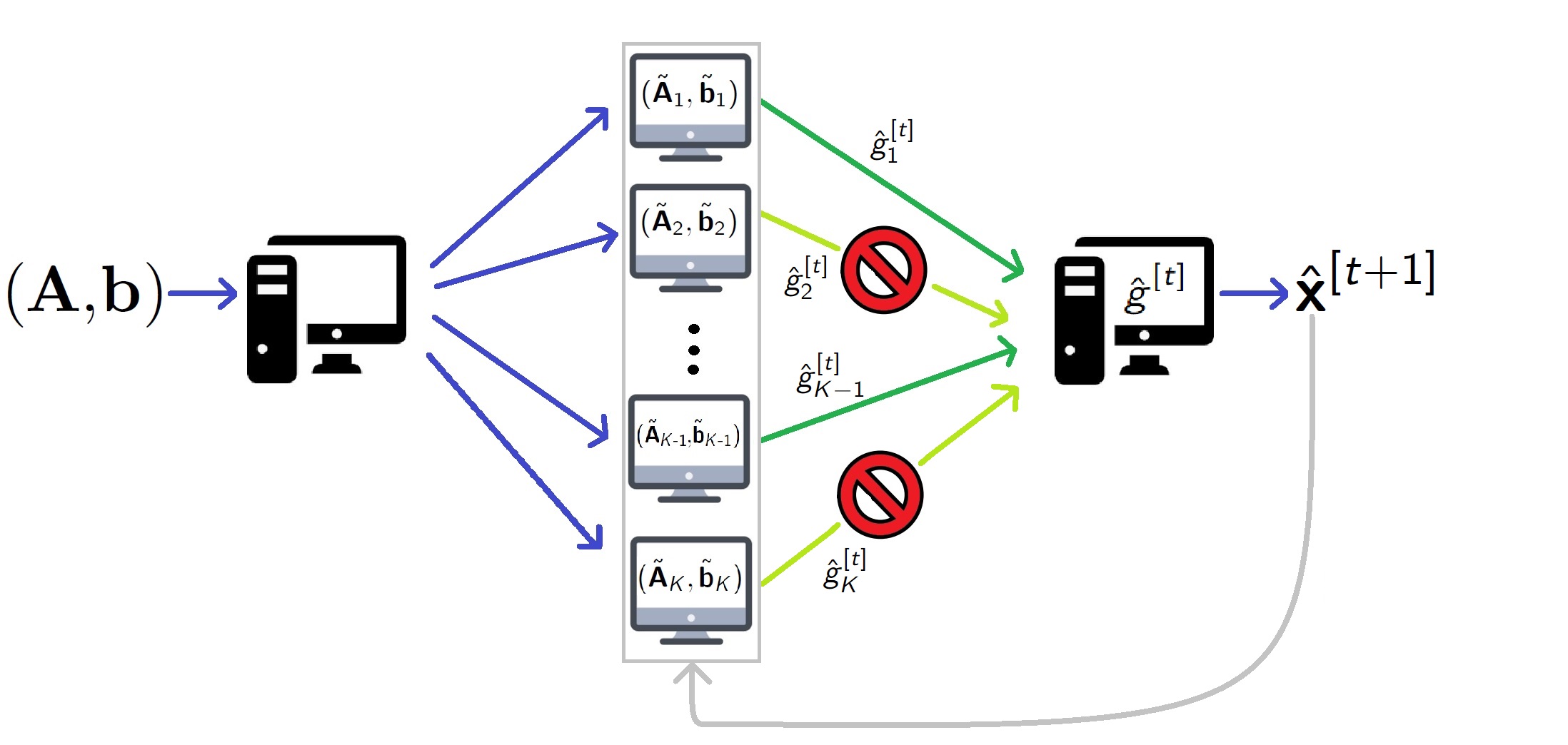}
    \caption{Illustration of our iterative sketching based GCS, at epoch $t+1$.}
  \label{it_sketch_figure}
\end{figure}

In Algorithm \ref{alg_orthog_sketch}, Theorems \ref{subsp_emb_thm_Unif} and \ref{subsp_emb_thm_SRHT}, we assume sampling with replacement. In what we just described, we used one replica of each block, thus $K=m$. To compensate for this, more than one replica of each block could be distributed. This is not a major concern with uniform sampling, as the probability that the $i^{th}$ block would be sampled more than once is $(q-1)/K^2$, which is negligible for large $K$.

\begin{Lemma}
\label{lemma_exp}
At any iteration $t$, with no replications of the blocks across the network, the resulting sketching matrix $\Sb_{[t]}$ satisfies $\E\left[\Sb_{[t]}^T\Sb_{[t]}\right]=\E\left[\Ombwt_{[t]}^T\Ombwt_{[t]}\right]=\Ib_N$.
\end{Lemma}

It is worth noting that by Lemma \ref{lemma_exp}, $\Sb_{[t]}$ in expectation satisfies the $\ell_2$-s.e. identity \eqref{subsp_emb_id} with $\epsilon=0$, as $$ \E\left[\Ub^T\big(\Sb_{[t]}^T\Sb_{[t]}\big)\Ub\right]=\Ub^T\E\left[\Sb_{[t]}^T\Sb_{[t]}\right]\Ub=\Ub^T\Ub=\Ib_d. $$

\begin{Thm}
\label{GC_SGD_thm}
The proposed GCS results in a mini-batch stochastic steepest descent procedure for
\begin{equation}
\label{mod_pr_G}
  \xbh = \arg\min_{\xb\in\R^d}\Big\{L_{\Gb}(\Ab,\bb;\xb)\coloneqq L_{ls}(\Gb\Ab,\Gb\bb;\xb)\Big\} .
\end{equation}
Moreover, $\E\left[\gh^{[t]}\right]=\frac{q}{K}\cdot g_{ls}^{[t]}$.
\end{Thm}

\begin{Lemma}
\label{eq_opt_sols}  % equal optimal solutions
The optimal solution of the modified least squares problem $L_{\Gb}(\Ab,\bb;\xb)$, is equal to the optimal solution $\xb_{ls}^{\star}$ of \eqref{x_star_pr_lr}.
\end{Lemma}

To prove Theorem \ref{GC_SGD_thm}, note that $\Ombwt_{[t]}$ corresponds to a uniform random selection of $q$ out of $K$ batches for each $t$; as in SSD, while in our procedure we consider the partial gradients of the $q$ fastest responses. When computing $\nabla_{\xb}L_\Gb(\Ab,\bb;\xb)$, the factor $\Pibold$ is annihilated; and the scaling factor $\sqrt{K/q}$ is squared. Since $\E\left[\gh^{[t]}\right]=\frac{q}{K} g_{ls}^{[t]}$, the estimate $\gh^{[t]}$ is unbiased after an appropriate rescaling; which could be incorporated in the step-size $\xi_t$.

\begin{Rmk}
By Theorem \ref{GC_SGD_thm} and Lemma \ref{eq_opt_sols}, it follows that with a diminishing step-size $\xi_t$ the expected regret of the least squares objective through our GCS converges to zero at a rate of $\ow(1/\sqrt{t}+r/t)$ \cite{DGSX12,Bub15}, and our updates $\xbh^{[t]}$ converge to $\xb_{ls}^{\star}$ in expectation at a rate of $\ow(1/t)$ \cite{RSS12,BWE19}.
\end{Rmk}

\begin{Cor}
\label{eq_SSD_dor}  % corollary on equivalence to SSD
Consider the problems \eqref{x_star_pr_lr} and \eqref{mod_pr_G}, which are respectively solved through SD and our iterative sketching based GCS. Assume that the two approaches have the same starting point $\xb^{[0]}$ and index set $\Scal^{[t]}$ at each $t$; and $\xih_t=\frac{K}{q}\xi_t$ the step-sizes used for our scheme. Then, in expectation, our approach through Algorithm \ref{alg_orthog_sketch} has the same update at each step $t$ as SD at the corresponding update, i.e $\E\left[\xbh^{[t]}\right]=\xb^{[t]}$.
\end{Cor}

By Lemma \ref{eq_opt_sols} and Corollary \ref{eq_SSD_dor}, the updated parameter estimates $\xbh^{[0]},\xbh^{[1]},\xbh^{[2]},\ldots$ of Algorithm \ref{alg_orthog_sketch} approach the optimal solution $\xb_{ls}^{\star}$ of \eqref{x_star_pr_lr}, by solving the modified regression problem \eqref{mod_pr_G} through SSD. It is also worth noting that the contraction rate of our GCS, in expectation is equal to that of regular SD. This can be shown through an analogous derivation of \cite[Theorem 6]{CPH23b}.%through the expected sketch $\SbPi^{[t]}$ at each iteration, is equal to that of regular SD.
%In the next subsection, we present our main $\ell_2$-s.e. result.

% - - - - - - - - - - - - - - -
\subsection{Subspace Embedding of Algorithm \ref{alg_orthog_sketch}}
\label{subsp_emb_alg1}

To give a $\ell_2$-s.e. guarantee for Algorithm \ref{alg_orthog_sketch}, we first show that the block leverage scores of $\Pibold\Ab$ are ``\textit{flattened}'', \textit{i.e.} they are all approximately equal. This is precisely what allows us to sample blocks for the construction of $\SbPi$; and in the distributed approach the computations of the partial gradients, \textit{uniformly} at random. Recall that the \textit{leverage scores} of $\Ubt\coloneqq\Pibold\Ub$ are $\ell_i\coloneqq\|\Ubt_{(i)}\|_2^2$ for each $i\in\N_N$, and the \textit{block leverage scores} are defined as $\tilde{\ell}_\iota\coloneqq\|\Ubt_{(\K_\iota)}\|_F^2=\sum_{j\in\K_\iota}\ell_j$ for all $\iota\in\N_K$ \cite{OJXE18,CPH20a}. A lot of work $\ell_2$-s.e. has utilized leverage scores as an importance sampling technique  \cite{DMM06,DMMS11,DMMW12,DM16,Woo14}. By generalizing these to sampling blocks, one can show analogous results (\textit{e.g.} \cite{CPH20a,CPH23b}).

Lemma \ref{bd_block_lvg_Unif} suggests that the \textit{normalized} block leverage scores $\ellg_i=\ellt_i/d$ of $\Ubt$ are approximately uniform for all $\iota$ with high probability. This is the key step to proving that each $\SbPi^{[t]}$ of Algorithm \ref{alg_orthog_sketch}, satisfies \eqref{subsp_emb_id}. We illustrate the flattening of the scores through the random projections considered in this paper in Figure \ref{flattening_t_distr}.%$\ellg_i=\frac{\ellt_i}{d}$

\begin{Lemma}
\label{bd_block_lvg_Unif}
For all $\iota\in\N_K$ and $\K_\iota\subsetneq\N_N$ of size $\tau=N/K$:
\begin{equation}
\label{bl_lvg_unif_flatten}
  \Pr\left[\ellg_\iota<_{N\rho}1/K\right] = \Pr\left[\big|\ellg_\iota-1/K\big|<\tau\rho\right] > 1-\delta,
\end{equation}
  for $\rho\geqslant\sqrt{\log(2\tau/\delta)/2}$.
\end{Lemma}

\begin{Thm}
\label{subsp_emb_thm_Unif}
Fix $\epsilon>0$ such that $\epsilon\ll1/N$. By Lemma \ref{bd_block_lvg_Unif}, we can then assume that $\ellg_\iota=1/K$ for all $\iota\in\N_K$. Then, $\SbPi$ of Algorithm \ref{alg_orthog_sketch} is a $\ell_2$-s.e. sketching matrix of $\Ab$. Specifically, for $\delta>0$ and $q=\Theta\left(\frac{d}{\tau}\log{(2d/\delta)}/\epsilon^2\right)$:%according to \eqref{subsp_emb_id}
\begin{equation*}
  \Pr\Big[\big\|\Ib_d-\Ub^T\SbPi^T\SbPi\Ub\big\|_2\leqslant\epsilon\Big]\geqslant 1-\delta.
\end{equation*}
\end{Thm}

The proof of Lemma \ref{bd_block_lvg_Unif} uses Hoeffding’s inequality to show that the individual leverage scores are flattened, which are then grouped together by an application of the binomial approximation. This is then directly applied to a generalized version of the leverage score sketching matrix which samples blocks instead of individual rows \cite[Theorem 1]{CPH23b}, to prove Theorem \ref{subsp_emb_thm_Unif}.

\begin{figure}[h]
  \centering
    \includegraphics[scale=.2]{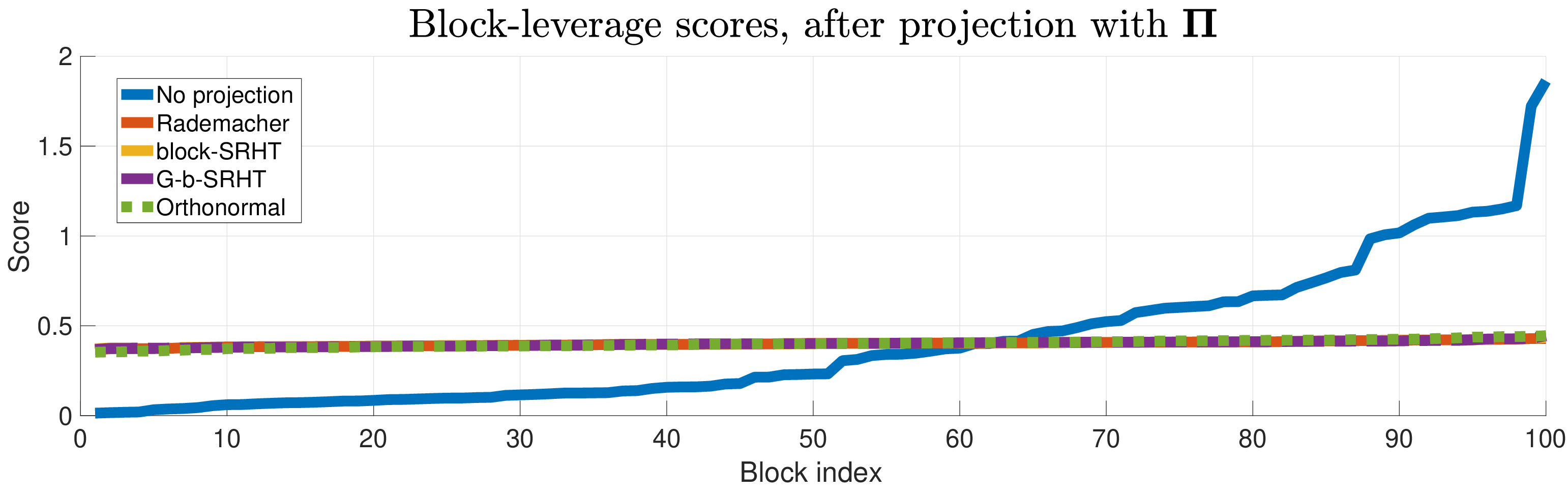}
    \caption{Flattening of block-scores, for $\Ab$ following a $t$-distribution. We abbreviate the garbled block-SRHT to `G-b-SRHT'.}
  \label{flattening_t_distr}
\end{figure}

We note that there is no benefit in considering an overlap across the block batches which are sent to the servers (\textit{e.g.} if a server receives $\big[\Abt_{i-1}^T\ \Abt_i^T\big]^T$ and another receives $\big[\Abt_i^T\ \Abt_{i+1}^T\big]^T$), in terms of sampling. The reason is that since the computations are received uniformly at random, there is still the same chance that $\gh_i$ and $\gh_j$ would be considered, for any $i\neq j$.

Before moving onto the block-SRHT, we show how our scheme compares to other approximate GCSs in terms of the approximation error \eqref{appr_GC_error}, when we consider multiple replications of each encoded block being shared among the servers. The result of Proposition \ref{appr_GC_prop} also applies to other sketching approaches, which satisfy \eqref{subsp_emb_id}.

\begin{Prop}
\label{appr_GC_prop}
By Theorem \ref{subsp_emb_thm_Unif}, $\SbPi$ satisfies \eqref{subsp_emb_id} (w.h.p.). Hence, the approximate gradients $\gh^{[t]}$ of Algorithm \ref{alg_orthog_sketch} satisfy \eqref{appr_GC_error} (w.h.p.), with $\err(\Gb_{(\I)})=\epsilon/\sqrt{K}$.
\end{Prop}

% - - - - - - - - - - - - - - -
\section{The Block-SRHT}
\label{block_SRHT_sec}

In this section, we focus on a special case of $\Pibold$ which can be utilized in Algorithm \ref{alg_orthog_sketch}; the \textit{randomized Hadamard transform}. By utilizing this transform we satisfy property (g), and avoid the extra computational cost needed to generate a random orthonormal matrix \cite{JHD21}.

The SRHT is comprised of three matrices: $\Omb\in\R^{r\times N}$ a uniform sampling and rescaling matrix of $r$ rows, $\Hbh_N\in\{\rpm1/\sqrt{N}\}^{N\times N}$ the normalized Hadamard matrix for $N=2^n$, and $\Db\in\{0,\rpm1\}^{N\times N}$ with i.i.d. diagonal Rademacher random entries; \textit{i.e.} it is a signature matrix. The main intuition of the projection is that it expresses the original signal or feature-row in the Walsh-Hadamard basis. Furthermore, $\Hbh_N$ can be applied efficiently due to its structure. As in the case where we transformed the left orthonormal basis and column-space of $\Ab$ by multiplying its columns with a random orthonormal matrix $\Pibold$, in the new basis $\Hbh_N\Db\Ub$; the block leverage scores are close to uniform. Hence, we perform \textit{uniform} sampling through $\Ombwt$ on the blocks of $\Hbh_N\Db\Ab$ to reduce dimension $N$, whilst maintaining the information of $\Ab$.

To exploit the SRHT in distributed GC for linear regression, we generalize it to subsampling blocks instead of rows; of the transformed data matrix, as in Algorithm \ref{alg_orthog_sketch}. We give a $\ell_2$-s.e. guarantee for the block-wise sampling version of SRHT, which characterizes the approximation of our proposed GCS for linear regression.

We refer to this special case as the ``\textit{block-SRHT}'', for which $\Pibh$ is taken from the subset $\hat{H}_N$ of $O_N(\R)$% set of orthonormal matrices
\begin{equation}
\label{set_Had}  % set of Hadamard transforms
  \hat{H}_N\coloneqq\left\{\Hbh_N\Db : \Db=\diag(\rpm1)\in\{0,\rpm1\}^{N\times N}\right\},
  %\hat{H}_N\coloneqq\left\{\Hbh_N\cdot\Db : \Db\gets\diag(\vec{D}), \text{for} \{\vec{D}_i\}_{i=1}^N\stackrel{\text{i.i.d.}}{\sim}\text{unif}\{-1,+1\} \right\},
\end{equation}
where $\Db$ is a random diagonal signature matrix with equiprobable entries of +1 and -1, and $\Hbh_N$ for $N=2^n$ is defined by
$$ \Hb_2 = \begin{pmatrix} 1 & 1 \\ 1 & -1 \end{pmatrix} \qquad \qquad \Hbh_N = \frac{1}{\sqrt{N}}\cdot\Hb_2^{\otimes \log_2(N)}. $$
%Equivalently, $\Hbh_N$ can be defined entry-wise through the $n$-bit binary representation of $i$,$j$ as $\Hbh_{ij}=(-1)^{\langle i,j\rangle_2}/\sqrt{N}$ for $\langle i,j\rangle_2 = \left(\sum_{l=0}^{n-1}i_l\cdot j_l\right)\bmod2$.
The SRHT introduced in \cite{DMMS11} corresponds to the case where we select $\tau=1$, \textit{i.e.} $K=N$. Henceforth, we drop the subscript $N$.%$$ \Hbh_{ij}=\frac{(-1)^{\langle i,j\rangle_2}}{\sqrt{N}}, $$

The main differences between the SRHT and the proposed block-SRHT $\SbPh$ for $\Pibh\getsU\hat{H}_N$, is the sampling matrix $\Ombwt$; and that $q=r/\tau$ sampling trials take place instead of $r$. The limiting computational step of applying $\SbPh$ in \eqref{x_til_pr_lr} is the multiplication by $\Hbh$. The recursive structure of $\Hbh$ permits us to compute $\SbPh\Ab$ in $\ow(Nd\log N)$ time, through Fourier methods \cite{Osg09}.

\subsection{Subspace Embedding of the Block-SRHT}

To show that $\SbPh$ with $\Pibh\getsU \hat{H}_N$ satisfies \eqref{subsp_emb_id}, we first present a key result, analogous to that of Lemma \ref{bd_block_lvg_Unif}. Considering the orthonormal basis $\Vbh\coloneqq\Hbh\Db\Ub$ of the transformed data $\Hbh\Db\Ab$ with individual leverage scores $\{\ell_i\}_{i=1}^N$, Lemma \ref{bd_block_lvg_Had} suggests that the resulting block leverage scores $\tilde{\ell}_\iota=\|\Vbh_{(\K_\iota)}\|_F^2=\sum_{j\in\K_\iota}\ell_j$ are approximately uniform for all $\iota\in\N_K$. Note that the diagonal entries of $\Db$ is the only place in which randomness takes place other than the sampling. This then allows us to prove Theorem \ref{subsp_emb_thm_SRHT}, our $\ell_2$-s.e. result regarding the block-SRHT.

\begin{Lemma}
\label{bd_block_lvg_Had}
For all $\iota\in\N_K$ and $\K_\iota\subsetneq\N_N$ of size $\tau=N/K$:
\begin{equation*}
\label{bl_lvg_had_flatten}
  \Pr\left[\tilde{\ell}_\iota\leqslant \eta d\cdot\log(Nd/\delta)/K\right]>1-\tau\delta/2,
\end{equation*}
for $0<\eta\leqslant 2+\log(16)/\log(Nd/\delta)$ a constant.
\end{Lemma}
%Furthermore, the transformation $\Hbh\Db$ also permits for a very sparse random projection to be applied, instead of $\Ombwt$ \cite{AC06}. Also note that the diagonal entries of $\Db$ is the only place in which randomness takes place other than the sampling.

\begin{Thm}
\label{subsp_emb_thm_SRHT}
The block-SRHT $\SbPh$ is a $\ell_2$-s.e. sketching matrix of $\Ab$. For $\delta>0$ and $q=\Theta\big(\frac{d}{\tau}\log(Nd/\delta)\cdot\log(2d/\delta)/\epsilon^2\big)$:
$$ \Pr\Big[\big\|\Ib_d-\Ub^T\SbPh^T\SbPh\Ub\big\|_2\leqslant\epsilon\Big]\geqslant 1-\delta . $$
\end{Thm}

Compared to Theorem \ref{subsp_emb_thm_Unif}, the above theorem has an additional logarithmic dependence on $N$. This is a consequence of applying the union bound, in order to show that the leverage scores of $\Hbh\Db\Ub$ are flattened (Lemma \ref{cor_fl_lem}). In the proof of Lemma \ref{bd_block_lvg_Unif}, we instead applied Hoeffding’s inequality, which removes such conditioning. Since Lemma \ref{bd_block_lvg_Unif} also holds for the block-SRHT, $\SbPh$ also satisfies the $\ell_2$-s.e. guarantee stated in Theorem \ref{subsp_emb_thm_Unif}.

In Subsection \ref{security_sec_SRHT} we alter the transformation $\Hbh\Db$ by permuting its rows. While the above $\ell_2$-s.e. result remains intact, under a mild but necessary assumption, this transformation now also guarantees computational security.

\subsection{Recursive Kronecker Products of Orthonormal Matrices}

One can consider more general sets of matrices to sample from, while still benefiting from the recursive structure leveraged in Fourier methods. For a fixed `base dimension' of $k\in\Z_{>2}$, let $\Pibold_k\in O_k(\R)$, and define $\Piba=\Pibold_k^{\otimes \ceil{\log_k(N)}}$. Performing the multiplication $\Piba\Ab$ now takes $\ow(Ndk^2\log_k N)$ time.

In the case where $k=2$, up to a permutation of the rows and columns; we have $O_2(\R)=\{\Ib_2,\Hbh_2\}$, which is limiting compared to $O_k(\R)$ for $k\geqslant3$. This allows more flexibility, as more `base matrices' $\Pibold_k\in O_k(\R)$ can be considered, and the security can therefore be improved, as now we do not rely only on applying a random permutation to $\Pibh$ (discussed in \ref{security_sec_SRHT}).

% - - - - - - - - - - - - - - -
\section{Optimal Step-Size and Adaptive GC}
\label{opt_ss_sec}

Recently, \textit{adaptive gradient coding} (AGC) was proposed in \cite{CT22}. The objective is to adaptively design an exact GCS without prior information about the behavior of potential persistent stragglers, in order to gradually minimize the communication load. This though comes at the cost of further delays due to intermediate designs of GC encoding-decoding pairs, as well as performing the encoding and decoding steps. Furthermore, the assumptions made in \cite{CT22} are more stringent compared to the ones we have made thus far.

In this section, we further accelerate our procedure, by adaptively selecting a step-size which reduces the total number of iterations required for convergence to the solutions of problems \eqref{x_star_pr_lr}, \eqref{x_til_pr_lr} and \eqref{mod_pr_G}, when SD is carried out. The proposed choice $\xi_t^{\star}$ for the step-size, is based on the latest gradient update of \eqref{x_star_pr_lr} and \eqref{x_til_pr_lr}. To determine $\xi_t^{\star}$, we solve 
\begin{align}
\label{ss_opt_prob}  % ss
  \xi_t^{\star} &= \arg\min_{\xi\in\R_{\geqslant0}}\Big\{\big\|\Ab\xb^{[t+1]}-\bb\big\|_2^2\Big\}\notag\\
  &= \arg\min_{\xi\in\R_{\geqslant0}}\Big\{\big\|\Ab(\xb^{[t]}-\xi\cdot g^{[t]})\bb\big\|_2^2\Big\}
  %\xi_t^{\star} = \arg\min_{\xi\in\R_{\geqslant0}}\Big\{\big\|\Ab\xb^{[t+1]}-\bb\big\|_2^2 = \big\|\Ab(\xb^{[t]}-\xi\cdot g^{[t]})\bb\big\|_2^2\Big\}
\end{align}
for each $t$. If $\xi_t=0$, we have reached the global optimum.

Since \eqref{ss_opt_prob} has a closed form solution, determining $\xi^{\star}_t$ at each iteration reduces to matrix-vector multiplications. In the distributed setting, this will be determined by the central server once sufficiently many servers have responded at iteration $t$, who will then update $\xb^{[t+1]}$ according to \eqref{par_upd}.

Compared to AGC, this is a more practical model, as we do not design and deploy multiple codes across the network. The authors of \cite{CT22} minimize the communication load of individual communication rounds. In contrast, we reduce the total number of iterations of the SD procedure, which leads to fewer communication rounds. Depending on the application and threshold parameters we pick for the two respective AGC methods, our proposed approach would likely have a lower overall communication load. This of course would also depend on the selected step-size used in the AGC for \cite{CT22}, and termination criterion. Furthermore, we are also flexible in tolerating a different number of stragglers $s$ at each iteration, which was a motivation for the design of AGC schemes.

\begin{Prop}
\label{prop_opt_ss}
  Given the respective gradient $g^{[t]}$ and update $\xb^{[t]}$ of the underlying objective function, the optimal step-size according to \eqref{ss_opt_prob} for $L_{ls}(\Ab,\bb;\xb^{[t]})$, $L_{\Gb}(\Ab,\bb;\xb^{[t]})$ and $L_{\Pibold}(\Ab,\bb;\xb^{[t]}) \coloneqq \big\|\Pibold(\Ab\xb-\bb)\big\|_2^2$, is:
  %$L_{ls}(\Ab,\bb;\xb^{[t]})$, $L_{\Pibold}(\Ab,\bb;\xb^{[t]})$ and $L_{\Gb}(\Ab,\bb;\xb^{[t]})$, is
  \begin{equation}
  \label{opt_ss}
    \xi_t^{\star} = \langle\Ab g^{[t]},\Ab\xb^{[t]}-\bb\rangle\big/\big\|\Ab g^{[t]}\big\|_2^2.
    %\xi_t^{\star} = \frac{\langle\Ab g^{[t]},\Ab\xb^{[t]}-\bb\rangle}{\|\Ab g^{[t]}\|_2^2}.
  \end{equation}
\end{Prop}

In our distributive stochastic procedure, one could select an adaptive step-size $\xi_{t+1}$ which minimizes $L_{\SbPi^{[t]}}(\Ab,\bb;\xb^{[t]})$; but the induced sketching matrix $\SbPi^{[t]}$ would need to be explicitly determined once $q$ servers have responded. This would result in further computations from the central server. Instead, we propose using the step-size \eqref{opt_ss}, as it is optimal in expectation.

The bottleneck in using $\xi_t^{\star}$, is that it can only be updated once $\gt^{[t]}\coloneqq\nabla_{\xb}L_{\Pibold}(\Ab,\bb;\xb^{[t]})$ has been determined, which causes a delay in updating $\xbh^{[t+1]}$. Even so, we significantly reduce the number of iterations, which is evident through our experiments in Section \ref{exper_sec}. The overall computation of the entire network is therefore also reduced. Furthermore, $\Ab^T\Ab$ and $\Ab^T\bb$ which appear in the expansion of $\xi_t^{\star}$ \eqref{opt_ss} can be computed beforehand, hence $\Ab^T\Ab\xb^{[t]}$ can be calculated by the central server while the servers are carrying out their tasks.%the $\gt^{[t]}$ of $L_{\Pibold}$ at iteration $t$ has been estimated,

\begin{Cor}
\label{cor_opt_ss}
  Assume that we know the parameter update $\xbh^{[t]}$, and the gradient $\gt^{[t]}$. Over the possible index sets $\Scal^{[t]}$ at iteration $t$, the optimal step-size according to
  \begin{equation}
  \label{opt_exp_ss}
    \arg\min_{\xi\in\R}\Big\{\E\left[\big\|\SbPi^{[t]}\big(\Ab\xbh^{[t+1]}-\bb\big)\big\|_2^2\right]\Big\}
  \end{equation}
  matches $\xi_t^{\star}$ of \eqref{opt_ss}.
\end{Cor}

Recall that other \textit{fixed coefficient decoding} GCSs exist \cite{RTTD17,CPE17,KKR19,BWE19,SH22}, \textit{i.e.} consider a fixed decoding vector $\abt=\gamma\vec{\bold{1}}$ at each iteration. In these schemes, the idea proposed in this subsection also applies; \textit{i.e.} consider decoding vectors $\abt^{[t]}=\xi_t^{\star}\gamma\vec{\bold{1}}$, in order to reduce the overall number of iterations needed for SD to converge. The optimal step-size $\xi_t^{\star}$ is determined by the central server at each iteration once it receives sufficiently many computations from the computational nodes. This computation replaces the decoding step performed by the central server, which is of high complexity. This of course can also be incorporated in \textit{optimal coefficient decoding} approximate or exact GCSs, though this leads to higher computation required from the central server.

% - - - - - - - - - - - - - - -
\section{Security of Orthonormal Sketches}
\label{security_sec}

In this section, we discuss the security of the proposed orthonormal-based sketching matrices and the block-SRHT. The idea behind securing the resulting sketches is that there is a large ensemble of orthonormal matrices $\Pibold$ to select from, making it near-impossible for adversaries to discover the inverse transformation.

To give information-theoretic security guarantees, we make some mild but necessary assumptions regarding Algorithm \ref{alg_orthog_sketch} and the data matrix $\Ab$. The message space $\M$ needs to be finite, which $\M$ in our case corresponds to the set of possible orthonormal bases of the column-space of $\Ab$. This is something we do not have control over, and it depends on the application and distribution from which we assume the data is gathered. Therefore, we assume that $\M$ is finite. For this reason, we consider a finite multiplicative subgroup $(\Otil_\Ab,\cdot)$ of $O_N(\R)$ (thus $\Ib_N\in\Otil_\Ab$, and if $\Qb\in\Otil_\Ab$ then $\Qb^T\in\Otil_\Ab$), which contains all potential orthonormal bases of $\Ab$.\footnote{In Appendix \ref{OTP_app}, we give an analogy between our approach and the \textit{one-time pad}.} Recall that $O_N(\R)$ is a regular submanifold of $\GL_N(\R)$. Hence, we can define a distribution on any subset of $O_N(\R)$.

We then let $\M=\Otil_\Ab$, and assume $\Ub_\Ab$ the $N\times N$ orthonormal basis of $\Ab$ is drawn from $\M$ w.r.t. a distribution $D$. For simplicity, we consider $D$ to be the uniform distribution. A simple method of generating a random matrix that follows the uniform distribution on the Stiefel manifold $V_n(\R^n)$ can be found in \cite[Theorem 2.2.1]{Yas12}. Alternatively, one could generate a random Gaussian matrix and then perform Gram–Schmidt to orthonormalize it. Furthermore, an inherent limitation of Shannon secrecy is that $|\K|\geqslant|\M|$.

\begin{Thm}
\label{Shan_secr_thm}
In Algorithm \ref{alg_orthog_sketch}, sample $\Pibold$ uniformly at random from $\Otil_\Ab$. The application of $\Pibold$ to $\Ab$ before partitioning the data, provides Shannon secrecy to $\Ab$ w.r.t. $D$ uniform, for $\K,\M,\Cc$ all equal to $\Otil_\Ab$.
\end{Thm}

% - - - - - - - - - - - - - - -
\subsection{Securing the SRHT}
\label{security_sec_SRHT}

Unfortunately, the guarantee of Theorem \ref{Shan_secr_thm} does not apply to the block-SRHT, as in this case it is restrictive to assume that $\Ub_\Ab\in \hat{H}_N$. A simple computation on a specific example also shows that this sketching approach does not provide Shannon secrecy.\footnote{Please check Appendix \ref{SRHT_counter_example} for the details.} For instance, if $\Ub_0=\Ib_2$, $\Ub_1=\Hbh_2$ and the observed transformed basis $\bar{\Cb}$ has two zero entries, then
$$ \Pr_{\Pibold\getsU \hat{H}_2}\left[\Pibold\cdot\Ub_1=\bar{\Cb}\right] > \Pr_{\Pibold\getsU \hat{H}_2}\left[\Pibold\cdot\Ub_0=\bar{\Cb}\right]=0. $$
Furthermore, since $\Hbh$ is a known orthonormal matrix, it is a trivial task to invert this projection and reveal $\Db\Ab$. This shows that the inherent security of the SRHT is relatively weak. Proposition \ref{insecurity_Prop} is proven by constructing a counterexample.
%That is $$ \Hbh\cdot(\Hbh\Db\Ab)=\Db\Ab $$ and an adversary then observes the features of $\Ab$ with roughly half of them having all their signs flipped. This shows that the inherent security of the SRHT is relatively weak, as it is not difficult to figure out which rows have their signs flipped.

\begin{Prop}
\label{insecurity_Prop}
  The SRHT does not provide Shannon secrecy.
\end{Prop}

To secure the SRHT and the block-SRHT, we randomly permute the rows of $\Hbh$, before applying it to $\Ab$. That is, for $\Pb\in S_N$ where $S_N\subsetneq\{0,1\}^{N\times N}$ is the permutation group on $N\times N$ matrices, we let $\Hbt\coloneqq\Pb\Hbh\in\{\rpm1/\sqrt{N}\}^{N\times N}$, and the new sketching matrix is
\begin{equation}
\label{Sbt_proj}  % projection
  \SbPt = \Ombwt\cdot(\Pb\cdot\Hbh)\cdot\Db =  \Ombwt\cdot\Hbt\cdot\Db = \Ombwt\cdot\Pibt,
\end{equation}
for which our flattening result (Corollary \ref{cor_fl_lem}) still holds. We ``garble'' $\Hbh$ so that the projection applied to $\Ab$ now inherently has more randomness, and allows us to draw from a larger ensemble. Specifically, for a fixed $N$, the block-SRHT has $2^N$ options for $\Pibh=\Hbh\Db$, while for $\Pibt=\Hbt\Db$ there are $2^NN!=\ow\big((2N/e)^N\sqrt{N}\big)$ options for $\Pibt$. Moreover, for
\begin{equation}
\label{set_G_Had}  % set of Garbled Hadamard transforms
  \Ht_N\coloneqq\left\{\Pb\Pibold : \Pb\in S_N \text{ and } \Pibold\in \hat{H}_N\right\}
\end{equation}
the set of all possible ``\textit{garbled Hadamard transforms}'', it follows that $(\Ht_N,\cdot)$ is a finite multiplicative subgroup of $O_N(\R)$. Hence, we can define a distribution on $\Ht_N$. We also get the benefits of permuting $\Hbh$'s columns without explicitly applying a second permutation, through $\Db$.% Note that $\Pibt^T\Pibt=\Ib_N$.

\begin{figure}[h]
  \centering
    \includegraphics[scale=.2]{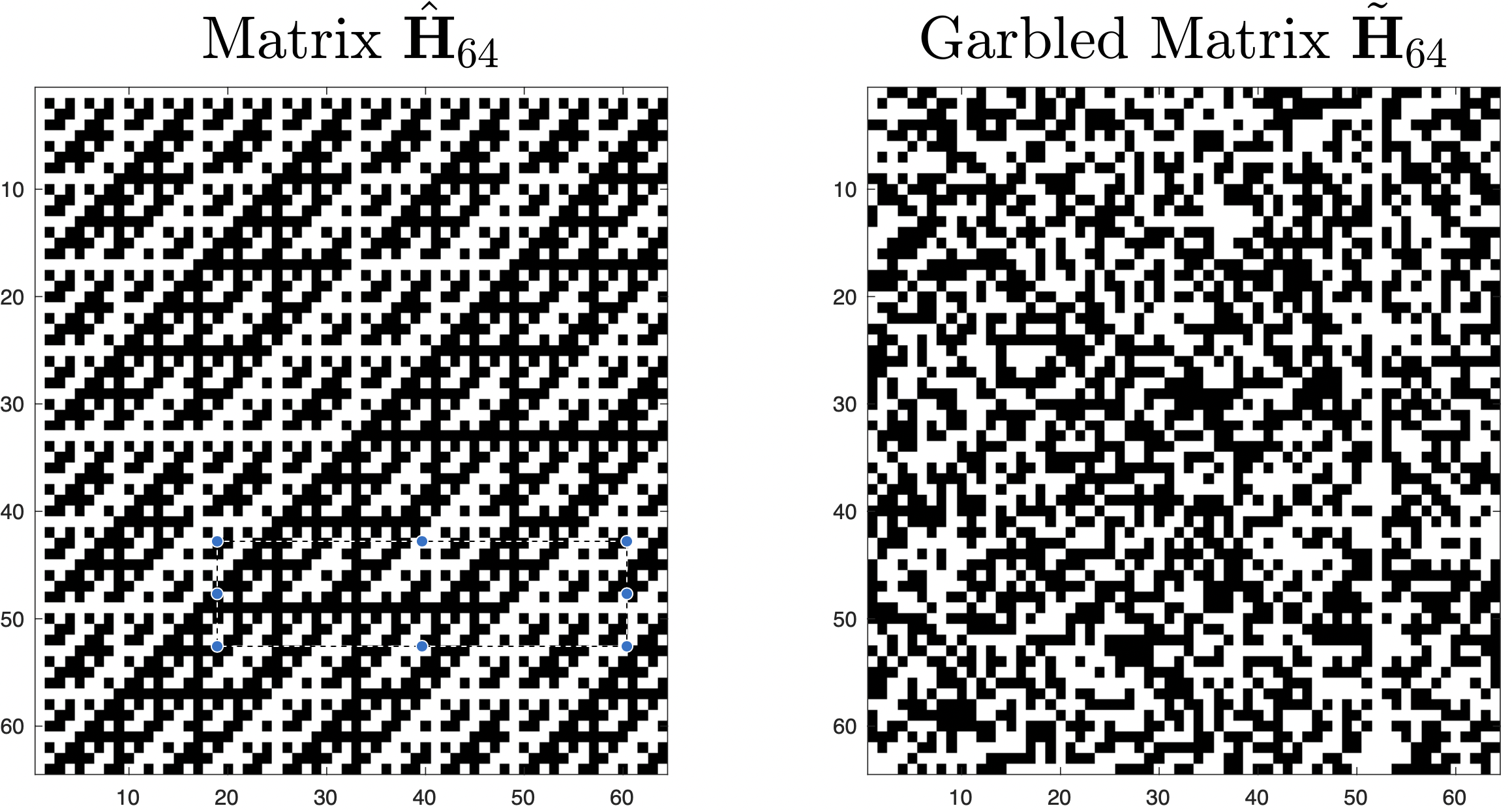}
\caption{Illustration of how $\Pb$ and $\Db$ modify the projection matrix $\Hbh_{64}$.}
  \label{had_transforms_fig}
\end{figure}

By the following Corollary, we deduce that Theorem \ref{subsp_emb_thm_SRHT} also holds for the ``\textit{garbled block-SRHT}'' (an analogous result is used to prove Lemma \ref{bd_block_lvg_Had}). Thus, we can apply any $\Pibt\getsU\Ht_N$ in Algorithm \ref{alg_orthog_sketch}, and get a valid sketch.

\begin{Cor}
\label{cor_fl_lem}
  For $\yb\in\R^N$ a fixed (orthonormal) column vector of $\Ub$, and $\Db\in\{0,\rpm1\}^{N\times N}$ with random equi-probable diagonal entries of $\rpm1$, we have:
  \begin{equation}
  \label{flat_lem_id_tilde}
    \Pr\left[\|\Hbt\Db\cdot\yb\|_\infty> C\sqrt{\log(Nd/\delta)/N}\right]\leqslant\frac{\delta}{2d}
  \end{equation}
  for $0<C\leqslant \sqrt{2+\log(16)/\log(Nd/\delta)}$ a constant.
\end{Cor}

Moreover, Corollary \ref{cor_fl_lem} also holds true for random projections $\Rb$ whose entries are rescaled Rademacher random variables, \textit{i.e.} $\Rb_{ij}=\rpm1/\sqrt{N}$ with equal probability. The advantage of this is that we have a larger set of projections
\begin{equation*}
  \Rt_N\coloneqq\Big\{\Rb\in\{\rpm1/\sqrt{N}\}^{N\times N}:\Pr\big[\Rb_{ij}=+1/\sqrt{N}\big]=1/2\Big\}
\end{equation*}
to draw from. This makes it harder for an adversary to determine which projection was applied. Specifically $|\Rt_N|=2^{N^2}$, which is significantly larger than $|\Ht_N|$. Two drawbacks of applying a random Rademacher projection $\Rb$ is that it is much slower than its Hadamard-based counterpart, and the resulting gradients $\gh^{[t]}$ are not unbiased.

Next, we provide a computationally secure guarantee for the garbled block-SRHT $\SbPt=\Ombwt\Pibt$, where $\Pibt\getsU\Ht_N$. The guarantee of Theorem \ref{SRHT_comp_sec_thm} against computationally bounded adversaries, relies heavily on the assumption that \textit{strong pseudorandom permutations} (s-PRPs) and \textit{one-way functions} (OWFs) exist. Through a long line of work, it was shown that s-PRPs exist if and only if OWFs exist. Even though OWFs are minimal cryptographic objects, it is not known whether such functions exist \cite{KL14}. Proving their existence is non-trivial, as this would then imply that $\textsf{P}\neq\textsf{NP}$. In practice however, this is not unreasonable to assume. The proof of Theorem \ref{SRHT_comp_sec_thm} entails a reduction to inverting the s-PRP $\Pb$. In practice, \textit{block ciphers} are used for s-PRPs.

\begin{Thm}
\label{SRHT_comp_sec_thm}
Assume that $\Pb$ is a s-PRP. Then, $\SbPt\Ab$ is computationally secure against polynomial-bounded adversaries, for $\SbPt=\Ombwt\Pibt$ the garbled block-SRHT.
\end{Thm}

As discussed in \ref{block_SRHT_sec}, the Hadamard matrix satisfies the desired properties (b), (c), (d), (g), while any other form of a discrete Fourier transform would violate (c). By applying $\Pb$ to $\Hbh$, the matrix $\Pb\Hbh$ still satisfies the aforementioned properties, while also incorporating security; \textit{i.e.} property (f). It would be interesting to see if other structured matrices exist which also satisfy (b)-(g). Similar to what we saw with the block-SRHT, if (b) is met; then we can achieve (a) through uniform sampling.

% - - - - - - - - - - - - - - -
\subsection{Exact Gradient Recovery}
\label{exact_grad_subsec}

In the case where the \textit{exact} gradient is desired, one can use the proposed orthonormal projections to encrypt the data from the servers, while requiring that the computations from \textit{all} servers are received; \textit{i.e.} $r=N$. That is, we instead consider the modified objective function
\begin{equation}
\label{encr_mod_ls_prob}
  L_{\Pibold}(\Ab,\bb;\xb^{[t]}) \coloneqq \|\Pibold(\Ab\xb-\bb)\|_2^2 = \|\Abg\xb-\bbg\|_2^2,
\end{equation}
for which each partial gradient $\gh_i^{[t]}$ is equal to $g_i^{[t]}$, hence $\gh_i^{[t]}=g^{[t]}$ for all $t$. From Theorems \ref{Shan_secr_thm} and \ref{SRHT_comp_sec_thm}, we know that under certain assumptions $\Ab$ is therefore secure. One can utilize $\Pibold$ to encrypt other distributive computations; \textit{e.g.} matrix multiplication or inversion\footnote{We note that the encryption through an arbitrary $\Pibold$ can cost just as much as the computation itself. Selecting $\Pibt\getsU\Ht_N$ though is reasonable.}, and logistic regression, which are discussed in Appendix \ref{orth_encr_distr_tasks}. This resembles a homomorphic encryption scheme, but is by no means fully-homomorphic.

In order to recover the exact gradient while also tolerating stragglers, as with any other GCS; the data needs to be replicated across the network. Straggler resiliency can be achieved by replicating each block $\lceil\frac{K}{m}(s+1)\rceil$ times, and assigning each server $(s+1)$ distinct blocks, which is precisely what is done in \textit{fractional repetition codes} \cite{TLDK17,CMH20}. This idea generalizes to \textit{any} linear regression GCS, by deploying it on the modified problem \eqref{encr_mod_ls_prob}.

Consider any GCS; exact or approximate, \textit{e.g.} \cite{TLDK17,HASH17,OGU19,CMH20,YA18,CT22,RTTD17,CP18,CPE17,WCP19,BWE19,WLS19,KKR19,HYKM19,CHZP18,CPH20a}. If the servers are given partitions of ($\Abg$,$\bbg$), they will have no knowledge of $(\Ab,\bb)$, unless they learn $\Pibold$. By Theorems \ref{Shan_secr_thm} and \ref{SRHT_comp_sec_thm}, this is not a concern. The servers therefore locally recover the partial gradients of the block pairs they were assigned, perform the encoding of the GCS which is being deployed, and communicate it to the central server. After sufficiently many encodings are received; \textit{i.e.} when the threshold recovery of the GCS is met, the central server can then recover the gradient at that iteration.% once these are encoded and communicated to the central server

% - - - - - - - - - - - - - - -
\section{Experiments}
\label{exper_sec}

We compared our proposed distributed GCSs to analogous approaches where the projection $\Pibold$ is a Gaussian sketch or a Rademacher random matrix. Our approach was found to outperform both these sketching methods in terms of convergence and approximation error, as the resulting gradients through these alternative approaches are biased. In all experiments, the same initialization $\xb^{[0]}$ was selected for each sketching method.

Our approach was also compared to uncoded (regular) SD. Random matrices $\Ab\in\R^{2000\times40}$ with non-uniform block leverage scores were generated for the experiments, and standard Gaussian noise was added to an arbitrary vector from $\image(\Ab)$ to define $\bb$. We considered $K=100$ blocks, thus $\tau=20$. Dimension $N$ was reduced to $r=1000$.% in all experiments.%$\Pibold\sim\mathcal{N}(0,1)\cdot\sqrt{\frac{1}{N}}$

For the experiments in Figure \ref{log_res_err_t_distr} we ran 600 iterations on six different instances for each one, and varied $\xi$ for each experiment by logarithmic factors of the step-size $\xi^\times=2/\sigma_{\max}(\Ab)^2$. The average $\log$ residual errors $\log_{10}\big(\|\xb_{ls}^{\star}-\xbh\|_2\big/\sqrt{N}\big)$ are depicted in Figure \ref{log_res_err_t_distr}. Step-size $\xi^\times$ was considered, as it guarantees descent at each iteration, though it is too conservative.

\begin{figure}[h]
  \centering
    \includegraphics[scale=.2]{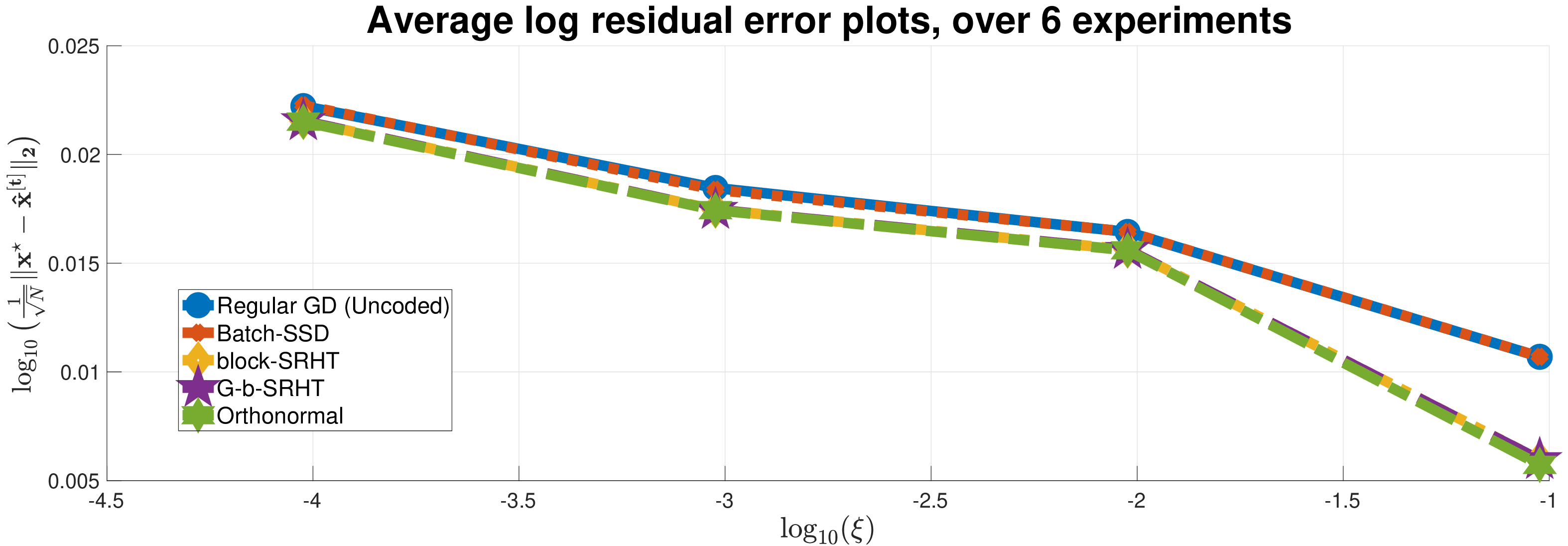}
    \caption{$\log$ residual error, for $\Ab$ following a $t$-distribution.}
  \label{log_res_err_t_distr}
\end{figure}

In contrast to the Gaussian sketch, orthonormal matrices $\Pibold$ also act as preconditioners. One example is the experiment depicted in Figure \ref{precond_figure}, in which the only modification we made from the above experiments, was our initial choice of $\xb^{[0]}$; which was scaled by $1/N$.

\begin{figure}[h]
  \centering
    \includegraphics[scale=.2]{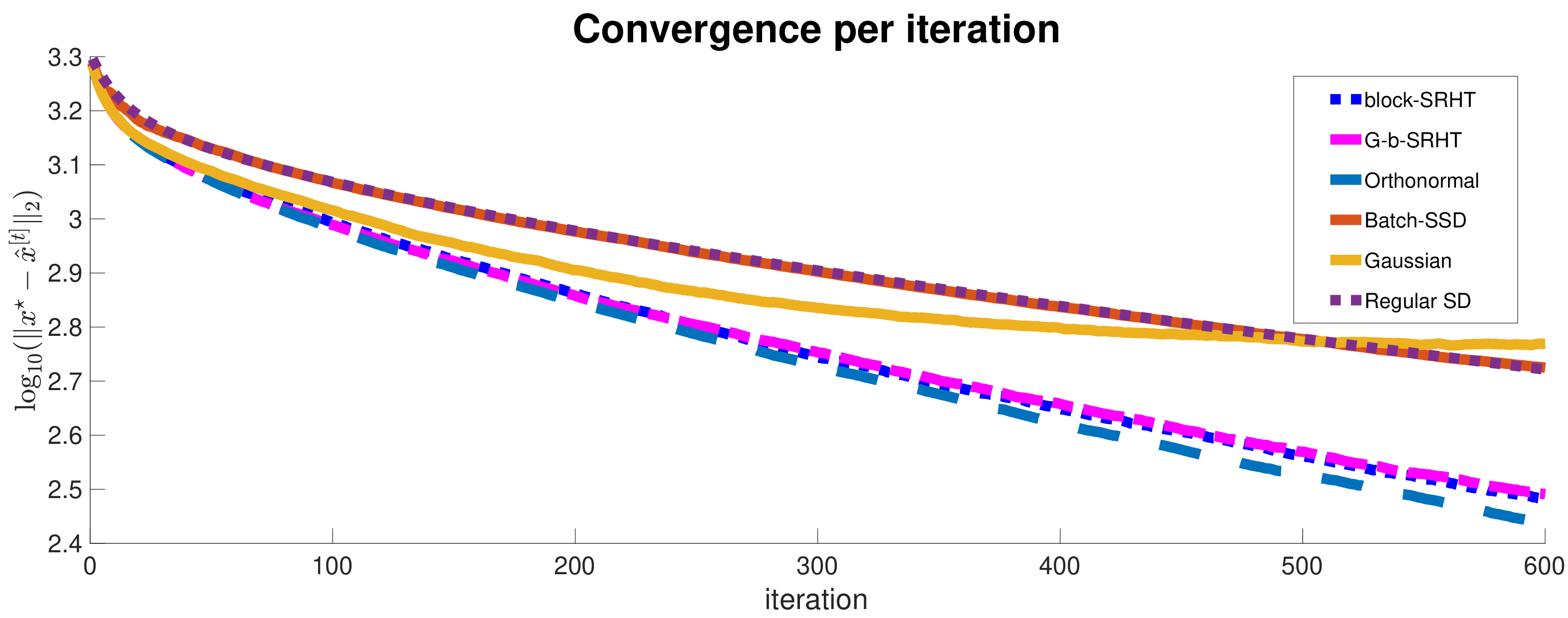}
    \caption{Example where $\Pibold$ also acts as a preconditioner.}
  \label{precond_figure}
\end{figure}

Next, we considered the case where $\xi_t$ was updated according to \eqref{opt_ss}. As above, our sketching approaches outperformed the case where a Gaussian sketch was applied. From Figure \ref{adaptive_GC_t_distr}, our orthonormal sketching approach performs just as well as regular SD for the first 30 iterations, though it slows down afterwards, and is slightly worse than regular SD by the time 50 iterations have been completed. By the discussion in \ref{exact_grad_subsec}, we can achieve the performance of regular SD if we wait until all servers respond; and consider no stragglers, while our security guarantees still hold. This is true also for the block-SRHT and garbled block-SRHT, but not for the Gaussian sketch.

\begin{figure}[h]
  \centering
    \includegraphics[scale=.2]{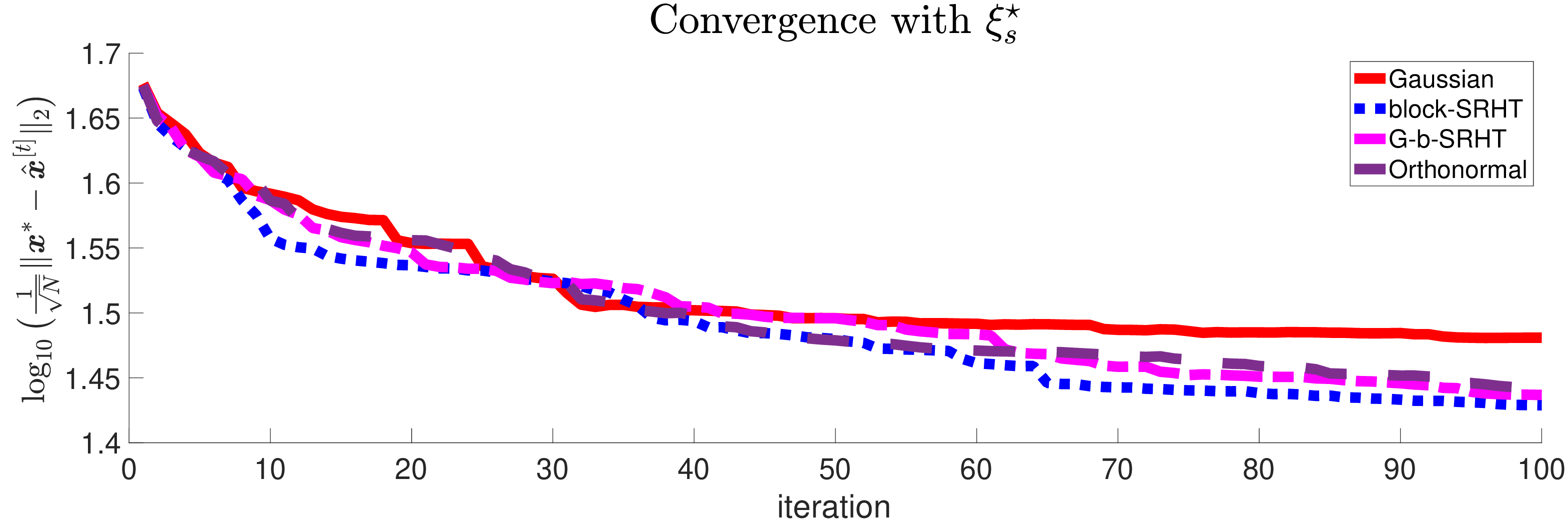}
    \caption{Adaptive step-size update, for $\Ab$ following a $t$-distribution.}
  \label{adaptive_GC_t_distr}
\end{figure}

Lastly, we present an example where it is clear that iterative sketching leads to better convergence than the sketch-and-solve approach. In the experiment depicted in Figure \ref{iter_vs_noniter_fig}, we considered three sketching approaches: the iterative block-SRHT and garbled block-SRHT, and the non-iterative garbled block-SRHT. The step-size was adaptive at each iteration, as was done in the experiment of Figure \ref{adaptive_GC_t_distr}.

\begin{figure}[h]
  \centering
    \includegraphics[scale=.2]{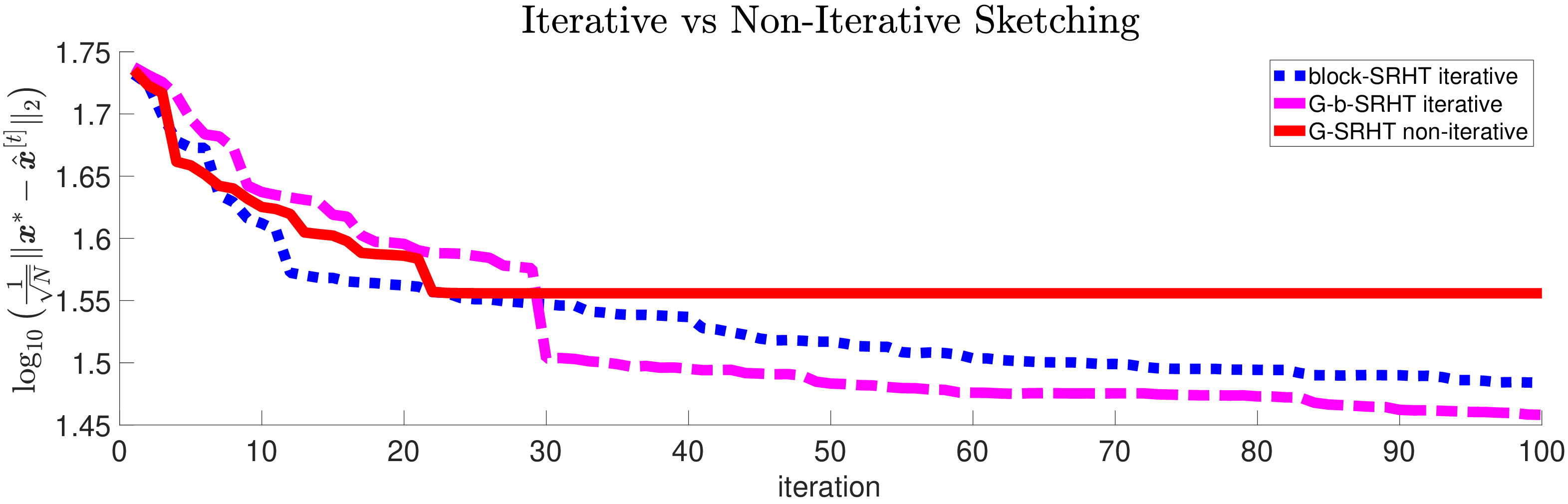}
    \caption{Convergence at each step, for the iterative block-SRHT and garbled block-SRHT, and the non-iterative garbled block-SRHT.}
  \label{iter_vs_noniter_fig}
\end{figure}

We carried out analogous experiments to all of the above; when considering other dense and sparse matrices $\Ab$, with non-uniform block leverage scores. Similar results regarding our approaches were observed, as the ones presented above.

% - - - - - - - - - - - - - - -
\section{Concluding Remarks and Future Work}
\label{concl_sec}

In this work, we proposed approximately solving a linear system by distributively leveraging iterative sketching and performing first-order SD simultaneously. In doing so, we benefit from both approximate GC and RandNLA. A difference to other RandNLA works is that our sketching matrices sample \textit{blocks} uniformly at random, after applying a random orthonormal projection. An additional benefit is that by considering a large ensemble of orthonormal matrices to pick from, under necessary assumptions, we guarantee information-theoretic security while performing the distributed computations. This approach also enables us to not require encoding and decoding steps at every iteration. We also studied the special case where the projection is the randomized Hadamard transform, and discussed its security limitation. To overcome this, we proposed a modified ``\textit{garbled block-SRHT}'', which guarantees computational security.

We note that applying orthonormal random matrices also secures coded matrix multiplication. There is a benefit when applying a garbled Hadamard transform in this scenario, as the complexity of multiplication resulting from the sketching is less than that of regular multiplication. Also, if such a random projection is used before performing $CR$-multiplication distributively \cite{CT19,CPH20c,RCHV23}, the approximate result will be the same. Moreover, our dimensionality reduction algorithm can be utilized by a single server, to store low-rank approximations of very large data matrices.

\textit{Partial stragglers}, have also been of interest in the GC literature. These are stragglers who are able to send back a portion of their requested tasks. Our work is directly applicable, as we can consider smaller blocks, with multiple ones allocated to each server. This extends on the idea described in \ref{exact_grad_subsec}, where sketching is combined with fractional repetition codes.

There are several interesting directions for future work. We observed experimentally in Figure \ref{precond_figure} that $\Pibold$ and $\Pibh$ may act as preconditioners for SSD. This mere observation requires further investigation. Another direction is to see if the proposed ideas could be applied to federated learning scenarios, in which security and privacy are major concerns. Some of the projections we considered, rely heavily on the recursive structure of $\Hbh$ in order to satisfy (g). One thing we overlooked, is whether other efficient multiplication algorithms (\textit{e.g.} Strassen's \cite{Str69}) could be exploited, in order to construct suitable projections. It would be interesting to see if other structured or sparse matrices exist which also satisfy our desired properties (a)-(g).

There has been a lot of work regarding second-order algorithms with iterative sketching, \textit{e.g.} \cite{PW16,LLDP20}. Utilizing iterative Hessian sketching or sketched Newton's method in CC has been explored in a tangential work \cite{GKCMR20}, though the security aspect of these algorithms has not been extensively studied. A drawback here is that the local computations at the servers would be much larger, though we expect the number of iterations to be significantly reduced; for the same termination criterion to be met, compared to first-order methods. Deeper exploration of the theoretical guarantees of iterative sketched first-order methods, along with a comparison to their second-order counterparts, as well as studying their effect in logistic regression and other applications, are also of potential interest.

% - - - - - - - - - - - - - - -
\appendices
% - - - - - - - - - - - - - - -

\section{Orthonormal Encryption for Distributive Tasks}
\label{orth_encr_distr_tasks}

In this appendix, we discuss how applying a random projection $\Pibold$ can be utilized in other existing CC schemes, both approximate and exact, to securely recover other matrix operations. As discussed in \ref{exact_grad_subsec}, the main idea is that after we apply and arbitrary $\Pibold$ to the underlying matrix or matrices, the analysis and conclusions of Theorems \ref{Shan_secr_thm} and \ref{SRHT_comp_sec_thm} still apply. Once the data is encrypted through $\Pibold$, \textit{e.g.} $\Abg=\Pibold\Ab,\bbg=\Pibold\bb$, we can then carry out the CC scheme of choice. We saw that for linear regression GCSs, we will recover the same result as when no encryption took place, without requiring an additional decryption step, and does not increase the system's redundancy. This is also true for \textit{coded matrix multiplication} (CMM) schemes. The drawback of this approach is the additional encryption step, which corresponds to matrix multiplication. Fast matrix multiplication can be used to secure the data \cite{Str69,WXXZ23}, which is faster than computing $\xb_{ls}^{\star}=\Ab^{\dagger}\xb$.

We show how this approach is applied to GCSs for logistic regression through SD, as well as CMM schemes, and an approximate \textit{matrix inversion} CC scheme; which is a non-linear operation \cite{CPH20b}. In the GC and matrix inversion schemes, we utilize the structure of the gradient of the respective objective functions.

The above discussion resembles \textit{Homomorphic Encryption} \cite{Gen09,Gen09a,BGV14}, which allows computations to be performed over encrypted data; and has been used to enable privacy-preserving machine learning. Two main drawbacks of homomorphic encryption though are that it leads to many orders of magnitude slow-down in training, and it allows collusion between a larger number of servers \cite{SGAM19}. Moreover, the privacy guarantees of homomorphic encryption rely on computational assumptions, while we achieved information-theoretic security in Theorem \ref{Shan_secr_thm}. Furthermore, we used (random) orthogonal matrices for encrypting the data, which has been studied in the context of image-processing and message encryption \cite{khan2015hill,AKHK17,AMAK18,Cac14,AID19,reddy2018image}.

% - - - - - - - - - - - - - - -
\begin{comment}
\subsection{Securing GCSs for Linear Regression}
\label{orth_Lin_regr}

As pointed out in \ref{exact_grad_subsec}, for the modified objective function
$$ L_{\Pibold}(\Ab,\bb;\xb^{[t]}) \coloneqq \|\Pibold(\Ab\xb-\bb)\|_2^2 = \|\Abg\xb-\bbg\|_2^2 , $$
we have $\nabla_{\xb}L_{\Pibold}(\Ab,\bb;\xb^{[t]}) = \nabla_{\xb}L_{ls}(\Ab,\bb;\xb^{[t]})$ for all $t$, \textit{i.e.} $\gh^{[t]}=g_{ls}^{[t]}$. It is clear that for $\Pibold$ an orthonormal matrix, there is no need to reverse the transformation to uncover the partial gradients.

Consider any GCS; exact or approximate, \textit{e.g.} \cite{TLDK17,HASH17,OGU19,CMH20,YA18,CT22,RTTD17,CP18,CPE17,WCP19,BWE19,WLS19,KKR19,HYKM19,CHZP18,CPH20a}. If the servers are given partitions of ($\Abg$,$\bbg$), they will have no knowledge of $(\Ab,\bb)$, unless they learn $\Pibold$. From the conclusion of Theorem \ref{Shan_secr_thm}, this is not a concern. The servers therefore locally recover the partial gradients of the block pairs they were assigned, and perform the encoding of the GCS which is being deployed, once these are encoded and communicated to the central server. After sufficiently many encodings are received, \textit{i.e.} when the threshold recovery is met, the central server can then recover the gradient at the given iteration.
\end{comment}
% - - - - - - - - - - - - - - -
\subsection{Securing GCSs for Logistic Regression}
\label{orth_Log_regr}

A popular algorithm whose solution is accelerated though gradient methods is logistic regression, which yields a linear classifier \cite{Mur12}. Applying a random orthonormal matrix can secure the information when GCSs are used to solve logistic regression, though at each iteration the central server will have to apply two encryptions. At each iteration $t$, the servers collectively have $\Abg=\Pibold_1\cdot\Ab$, $\abg_i=\Pibold_2\cdot\big[1\ \ab_i^T\big]^T$ and $\xbg^{[t]}=\Pibold_2\cdot\xb^{[t]}$, for $\Pibold_1\in \Otil_N(\R)$ and $\Pibold_2\in \Otil_{d+1}(\R)$. The gradient update to be recovered is
$$ \gh^{[t+1]}=\Abg^T(\boldsymbol{\mu}-\bb), \ \text{ for } \ \boldsymbol{\mu}_i=\Big(1+\mathrm{exp}\big(-\overbrace{\langle\xbg^{[t]},\abg_i\rangle}^{=\langle\xb^{[t]},[1\ \ab_i^T]\rangle}\big)\Big)^{-1}. $$
Thus, $\gh^{[t+1]}=\Pibold_1^T\cdot g^{[t+1]}$, so we apply $\Pibold_1$ to recover $g^{[t+1]}$. In this problem, the labels $\bb_i\in\{0,1\}$ are not hidden. This is not a concern, as other than basic statistics; nothing meaningful can be inferred from these alone.

% - - - - - - - - - - - - - - -
\subsection{Securing Coded Matrix Multiplication}
\label{orth_CMM}

Consider the matrices $\Ab_1\in\R^{L\times N}$ and $\Ab_2\in\R^{N\times M}$, whose product is to be computed by a CMM scheme. For $\Pibold\in\Otil_N(\R)$, by carrying out the CMM scheme on $\Abg_1=\Ab_1\cdot\Pibold$ and $\Abg_2=\Pibold\cdot\Ab_2$, we recover
\begin{equation}
  \Abg_1\cdot\Abg_2 = \Ab_1\cdot(\Pibold^T\Pibold)\cdot\Abg_2 = \Ab_1\cdot\Ab_2,
\end{equation}
and security guarantees analogous to Theorems \ref{Shan_secr_thm} and \ref{SRHT_comp_sec_thm} hold. In general, this encryption is useful when $N\ll L,M$, as otherwise the cost of encrypting the two matrices could be higher than that of performing the multiplication.

% - - - - - - - - - - - - - - -
\subsection{Securing Coded Matrix Inversion}

In \cite{CPH20b} a CC scheme was used to recover an approximation of the inverse of a matrix $\Ab\in\R^{N\times N}$, by requesting from the servers to approximate as part of their computation a subset of the optimization problems
\begin{equation}
  \bbc_i = \arg\min_{\bb\in\R^N} \big\{\|\Ab\bb-\eb_i\|_2^2\big\}
\end{equation}
for each $i\in\N_N$, for $\{\eb_i\}_{i=1}^N$ the standard basis of $\R^N$. The solutions $\{\bbc\}_{i=1}^N$ comprise the columns of the inverse's estimate $\Abc^{-1}$, \textit{i.e.} $\Ab^{-1}\approx\Abc^{-1}=\big[\bbc_1 \ \cdots \ \bbc_N \big]$, as
$$ \Ib_N = \Ab\Ab^{-1} \approx \Ab\Abc^{-1} = \Ab\big[\bbc_1 \ \cdots \ \bbc_N \big] = \big[\Ab\bbc_1 \ \cdots \ \Ab\bbc_N \big]. $$

In this scheme, each server has knowledge of the entire matrix $\Ab$. To utilize the security results of our work, instead of sharing $\Ab$; we distribute $\Abg\coloneqq\Ab\cdot\Pibold^T$, for a randomly chosen $\Pibold\in\Otil_{N}(\R)$. The servers then approximate
\begin{equation}
  \bbbr_i = \arg\min_{\bb\in\R^N} \big\{\|\Abg\bb-\eb_i\|_2^2\big\} ,
\end{equation}
thus $\Abg^{-1}=\Pibold\cdot\Abc^{-1}=\big[\bbbr_1 \ \cdots \ \bbbr_N \big]$. As in the case of logistic regression, here, we also need an additional decryption step: $\Pibold^T\cdot(\Pibold\cdot\Abc^{-1})=\Abc^{-1}$ at the end of the process, to recover the approximation.

% - - - - - - - - - - - - - - -
\bibliographystyle{IEEEtran}
\bibliography{refs}

\begin{thebibliography}{10}
\providecommand{\url}[1]{#1}
\csname url@samestyle\endcsname
\providecommand{\newblock}{\relax}
\providecommand{\bibinfo}[2]{#2}
\providecommand{\BIBentrySTDinterwordspacing}{\spaceskip=0pt\relax}
\providecommand{\BIBentryALTinterwordstretchfactor}{4}
\providecommand{\BIBentryALTinterwordspacing}{\spaceskip=\fontdimen2\font plus
\BIBentryALTinterwordstretchfactor\fontdimen3\font minus
  \fontdimen4\font\relax}
\providecommand{\BIBforeignlanguage}[2]{{%
\expandafter\ifx\csname l@#1\endcsname\relax
\typeout{** WARNING: IEEEtran.bst: No hyphenation pattern has been}%
\typeout{** loaded for the language `#1'. Using the pattern for}%
\typeout{** the default language instead.}%
\else
\language=\csname l@#1\endcsname
\fi
#2}}
\providecommand{\BIBdecl}{\relax}
\BIBdecl

\bibitem{CMPH22}
N.~Charalambides, H.~Mahdavifar, M.~Pilanci, and A.~O. Hero, ``{Orthonormal
  Sketches for Secure Coded Regression},'' in \emph{2022 IEEE International
  Symposium on Information Theory (ISIT)}, 2022, pp. 826--831.

\bibitem{Vem05}
S.~S. Vempala, \emph{{The Random Projection Method}}.\hskip 1em plus 0.5em
  minus 0.4em\relax American Mathematical Soc., 2005, vol.~65.

\bibitem{Woo14}
D.~P. Woodruff, ``{Sketching as a Tool for Numerical Linear Algebra},''
  \emph{arXiv preprint arXiv:1411.4357}, 2014.

\bibitem{DMMS11}
P.~Drineas, M.~W. Mahoney, S.~Muthukrishnan, and T.~Sarl{\'o}s, ``{Faster Least
  Squares Approximation},'' \emph{Numerische mathematik}, vol. 117, no.~2, pp.
  219--249, 2011.

\bibitem{DM16}
P.~Drineas and M.~W. Mahoney, ``{RandNLA: Randomized Numerical Linear
  Algebra},'' \emph{Communications of the ACM}, vol.~59, no.~6, pp. 80--90,
  2016.

\bibitem{LLPPR17}
K.~Lee, M.~Lam, R.~Pedarsani, D.~Papailiopoulos, and K.~Ramchandran,
  ``{Speeding Up Distributed Machine Learning Using Codes},'' \emph{IEEE
  Transactions on Information Theory}, vol.~64, no.~3, pp. 1514--1529, 2017.

\bibitem{reisizadeh2017coded}
A.~Reisizadeh, S.~Prakash, R.~Pedarsani, and A.~S. Avestimehr, ``{Coded
  Computation over Heterogeneous Clusters},'' in \emph{2017 IEEE International
  Symposium on Information Theory (ISIT)}, 2017, pp. 2408--2412.

\bibitem{li2016coded}
S.~Li, M.~A. Maddah-Ali, and A.~S. Avestimehr, ``{Coded distributed computing:
  Straggling servers and multistage dataflows},'' in \emph{54th Annual Allerton
  Conference}.\hskip 1em plus 0.5em minus 0.4em\relax IEEE, 2016, pp. 164--171.

\bibitem{li2017coding}
------, ``{Coding for Distributed Fog Computing},'' \emph{IEEE Commun. Mag.},
  vol.~55, no.~4, pp. 34--40, 2017.

\bibitem{LSR17}
K.~Lee, C.~Suh, and K.~Ramchandran, ``{High-Dimensional Coded Matrix
  Multiplication},'' in \emph{IEEE International Symposium on Information
  Theory (ISIT)}.\hskip 1em plus 0.5em minus 0.4em\relax IEEE, 2017, pp.
  2418--2422.

\bibitem{dutta2016short}
S.~Dutta, V.~Cadambe, and P.~Grover, ``{``Short-Dot''': Computing Large Linear
  Transforms Distributedly Using Coded Short Dot Products},'' in \emph{Adv. in
  Neural Info. Proc. Systems (NIPS)}, 2016, pp. 2100--2108.

\bibitem{ramamoorthy2019universally}
A.~Ramamoorthy, L.~Tang, and P.~O. Vontobel, ``Universally decodable matrices
  for distributed matrix-vector multiplication,'' \emph{arXiv preprint
  arXiv:1901.10674}, 2019.

\bibitem{YSRKSA18}
Q.~Yu, S.~Li, N.~Raviv, S.~M.~M. Kalan, M.~Soltanolkotabi, and S.~A.
  Avestimehr, ``{Lagrange Coded Computing: Optimal Design for Resiliency,
  Security, and Privacy},'' in \emph{The 22nd International Conference on
  Artificial Intelligence and Statistics}.\hskip 1em plus 0.5em minus
  0.4em\relax PMLR, 2019, pp. 1215--1225.

\bibitem{RRG20}
M.~Rudow, K.~Rashmi, and V.~Guruswami, ``A locality-based lens for coded
  computation,'' in \emph{2021 IEEE International Symposium on Information
  Theory (ISIT)}, 2021, pp. 1070--1075.

\bibitem{CT19}
W.-T. Chang and R.~Tandon, ``{Random Sampling for Distributed Coded Matrix
  Multiplication},'' in \emph{ICASSP 2019-2019 IEEE International Conference on
  Acoustics, Speech and Signal Processing (ICASSP)}.\hskip 1em plus 0.5em minus
  0.4em\relax IEEE, 2019, pp. 8187--8191.

\bibitem{CPH20c}
N.~Charalambides, M.~Pilanci, and A.~O. Hero, ``{Approximate Weighted
  $CR$-Coded Matrix Multiplication},'' in \emph{ICASSP 2021-2021 IEEE
  International Conference on Acoustics, Speech and Signal Processing
  (ICASSP)}.\hskip 1em plus 0.5em minus 0.4em\relax IEEE, 2021, pp. 5095--5099.

\bibitem{CPH20b}
N.~Charalambides, M.~Pilanci, and A.~O. Hero~III, ``{Straggler Robust
  Distributed Matrix Inverse Approximation},'' \emph{arXiv preprint
  arXiv:2003.02948}, 2020.

\bibitem{OUG20}
E.~Ozfatura, S.~Ulukus, and D.~Gunduz, ``Coded distributed computing with
  partial recovery,'' \emph{arXiv preprint arXiv:2007.02191}, 2020.

\bibitem{OBGU20}
E.~Ozfatura, B.~Buyukates, D.~Gunduz, and S.~Ulukus, ``Age-based coded
  computation for bias reduction in distributed learning,'' \emph{arXiv
  preprint arXiv:2006.01816}, 2020.

\bibitem{CMH21}
N.~Charalambides, H.~Mahdavifar, and A.~O. Hero~III, ``{Generalized Fractional
  Repetition Codes for Binary Coded Computations},'' \emph{arXiv preprint
  arXiv:2109.10484}, 2024.

\bibitem{RCHV23}
M.~Rudow, N.~Charalambides, A.~O. Hero~III, and K.~Rashmi,
  ``{Compression-Informed Coded Computing},'' in \emph{IEEE International
  Symposium on Information Theory (ISIT)}, 2023, pp. 2177--2182.

\bibitem{CPH23}
N.~Charalambides, M.~Pilanci, and A.~O. Hero, ``{Securely Aggregated Coded
  Matrix Inversion},'' \emph{IEEE Journal on Selected Areas in Information
  Theory}, vol.~4, pp. 405--419, 2023.

\bibitem{LA20}
S.~Li and S.~Avestimehr, ``{Coded Computing},'' \emph{Foundations and
  Trends{\textregistered} in Communications and Information Theory}, vol.~17,
  no.~1, 2020.

\bibitem{BGV88}
M.~Ben-Or, S.~Goldwasser, and A.~Wigderson, ``{Completeness Theorems for
  Non-Cryptographic Fault-Tolerant Distributed Computations},'' in
  \emph{Proceedings of the $20^{th}$ Annual ACM Symposium on the Theory of
  Computing, 1988}, 1988, pp. 1--10.

\bibitem{CCD88}
D.~Chaum, C.~Cr{\'e}peau, and I.~Damgard, ``{Multiparty Unconditionally Secure
  Protocols},'' in \emph{Proceedings of the $20^{th}$ Annual ACM Symposium on
  the Theory of Computing, 1988}, 1988, pp. 11--19.

\bibitem{Sar06}
T.~Sarl{\'o}s, ``{Improved Approximation Algorithms for Large Matrices via
  Random Projections},'' in \emph{2006 47th annual IEEE symposium on
  foundations of computer science (FOCS '06)}.\hskip 1em plus 0.5em minus
  0.4em\relax IEEE, 2006, pp. 143--152.

\bibitem{Tro11}
J.~A. Tropp, ``{Improved analysis of the subsampled randomized Hadamard
  transform},'' \emph{Advances in Adaptive Data Analysis}, vol.~3, no. 01n02,
  pp. 115--126, 2011.

\bibitem{BG13}
C.~Boutsidis and A.~Gittens, ``{Improved matrix algorithms via the Subsampled
  Randomized Hadamard Transform},'' \emph{SIAM Journal on Matrix Analysis and
  Applications}, vol.~34, no.~3, pp. 1301--1340, 2013.

\bibitem{AC06}
N.~Ailon and B.~Chazelle, ``{Approximate Nearest Neighbors and the Fast
  Johnson--Lindenstrauss Transform},'' in \emph{Proceedings of the
  thirty-eighth annual ACM symposium on Theory of computing}, 2006, pp.
  557--563.

\bibitem{JL84}
W.~B. Johnson and J.~Lindenstrauss, ``{Extensions of Lipschitz mappings into a
  Hilbert space},'' in \emph{Contemp. Math.}, vol.~26, 1984, pp. 189--206.

\bibitem{BP23}
B.~Bartan and M.~Pilanci, ``Distributed sketching for randomized optimization:
  Exact characterization, concentration, and lower bounds,'' \emph{IEEE
  Transactions on Information Theory}, vol.~69, no.~6, pp. 3850--3879, 2023.

\bibitem{TLDK17}
R.~Tandon, Q.~Lei, A.~G. Dimakis, and N.~Karampatziakis, ``{Gradient coding:
  Avoiding stragglers in distributed learning},'' in \emph{International
  Conference on Machine Learning}, 2017, pp. 3368--3376.

\bibitem{HASH17}
W.~Halbawi, N.~Azizan, F.~Salehi, and B.~Hassibi, ``Improving distributed
  gradient descent using {R}eed-{S}olomon codes,'' in \emph{2018 IEEE
  International Symposium on Information Theory (ISIT)}.\hskip 1em plus 0.5em
  minus 0.4em\relax IEEE, 2018, pp. 2027--2031.

\bibitem{OGU19}
E.~Ozfatura, D.~Gunduz, and S.~Ulukus, ``Gradient coding with clustering and
  multi-message communication,'' \emph{arXiv preprint arXiv:1903.01974}, 2019.

\bibitem{CMH20}
N.~Charalambides, H.~Mahdavifar, and A.~O. Hero, ``{Numerically Stable Binary
  Gradient Coding},'' in \emph{2020 IEEE International Symposium on Information
  Theory (ISIT)}, 2020, pp. 2622--2627.

\bibitem{YA18}
M.~Ye and E.~Abbe, ``{Communication-Computation Efficient Gradient Coding},''
  in \emph{International Conference on Machine Learning}.\hskip 1em plus 0.5em
  minus 0.4em\relax PMLR, 2018, pp. 5610--5619.

\bibitem{CT22}
H.~Cao, Q.~Yan, X.~Tang, and G.~Han, ``{Adaptive Gradient Coding},''
  \emph{IEEE/ACM Transactions on Networking}, vol.~30, no.~2, pp. 717--734,
  2022.

\bibitem{RTTD17}
N.~Raviv, I.~Tamo, R.~Tandon, and A.~G. Dimakis, ``{Gradient Coding from Cyclic
  MDS Codes and Expander Graphs},'' \emph{IEEE Transactions on Information
  Theory}, vol.~66, no.~12, pp. 7475--7489, 2020.

\bibitem{CPE17}
Z.~Charles, D.~Papailiopoulos, and J.~Ellenberg, ``{Approximate Gradient Coding
  via Sparse Random Graphs},'' \emph{arXiv preprint arXiv:1711.06771}, 2017.

\bibitem{KKR19}
S.~Kadhe, O.~O. Koyluoglu, and K.~Ramchandran, ``{Gradient Coding Based on
  Block Designs for Mitigating Adversarial Stragglers},'' in \emph{2019 IEEE
  International Symposium on Information Theory (ISIT)}.\hskip 1em plus 0.5em
  minus 0.4em\relax IEEE, 2019, pp. 2813--2817.

\bibitem{BWE19}
R.~Bitar, M.~Wootters, and S.~El~Rouayheb, ``{Stochastic Gradient Coding for
  Straggler Mitigation in Distributed Learning},'' \emph{IEEE Journal on
  Selected Areas in Information Theory}, vol.~1, pp. 277--291, 2020.

\bibitem{SH22}
A.~Sakorikar and L.~Wang, ``{Soft BIBD and Product Gradient Codes},''
  \emph{IEEE Journal on Selected Areas in Information Theory}, vol.~3, no.~2,
  pp. 229--240, 2022.

\bibitem{GW21}
M.~Glasgow and M.~Wootters, ``{Approximate Gradient Coding with Optimal
  Decoding},'' \emph{IEEE Journal on Selected Areas in Information Theory},
  vol.~2, no.~3, pp. 855--866, 2021.

\bibitem{CP18}
Z.~Charles and D.~Papailiopoulos, ``{Gradient Coding Using the Stochastic Block
  Model},'' in \emph{2018 IEEE International Symposium on Information Theory
  (ISIT)}, 2018, pp. 1998--2002.

\bibitem{WCP19}
H.~Wang, Z.~Charles, and D.~Papailiopoulos, ``Erasurehead: Distributed gradient
  descent without delays using approximate gradient coding,'' \emph{arXiv
  preprint arXiv:1901.09671}, 2019.

\bibitem{WLS19}
S.~Wang, J.~Liu, and N.~Shroff, ``{Fundamental Limits of Approximate Gradient
  Coding},'' \emph{Proceedings of the ACM on Measurement and Analysis of
  Computing Systems}, vol.~3, no.~3, pp. 1--22, 2019.

\bibitem{HYKM19}
S.~Horii, T.~Yoshida, M.~Kobayashi, and T.~Matsushima, ``{Distributed
  Stochastic Gradient Descent Using LDGM Codes},'' \emph{arXiv preprint
  arXiv:1901.04668}, 2019.

\bibitem{CHZP18}
L.~Chen, H.~Wang, Z.~Charles, and D.~Papailiopoulos, ``{DRACO:
  Byzantine-resilient Distributed Training via Redundant Gradients},''
  \emph{arXiv preprint arXiv:1803.09877}, 2018.

\bibitem{CPH20a}
N.~Charalambides, M.~Pilanci, and A.~O. Hero, ``{Weighted Gradient Coding with
  Leverage Score Sampling},'' in \emph{ICASSP 2020-2020 IEEE International
  Conference on Acoustics, Speech and Signal Processing (ICASSP)}.\hskip 1em
  plus 0.5em minus 0.4em\relax IEEE, 2020, pp. 5215--5219.

\bibitem{PW16}
M.~Pilanci and M.~J. Wainwright, ``{Iterative Hessian sketch: Fast and accurate
  solution approximation for constrained least-squares},'' \emph{The Journal of
  Machine Learning Research}, vol.~17, no.~1, pp. 1842--1879, 2016.

\bibitem{LLDP20}
J.~Lacotte, S.~Liu, E.~Dobriban, and M.~Pilanci, ``{Optimal iterative sketching
  methods with the subsampled randomized Hadamard transform},'' \emph{Advances
  in Neural Information Processing Systems}, vol.~33, 2020.

\bibitem{ZWL08}
S.~Zhou, L.~Wasserman, and J.~Lafferty, ``{Compressed Regression},'' in
  \emph{Advances in Neural Information Processing Systems}, vol.~20, 2008.

\bibitem{KSD17}
C.~Karakus, Y.~Sun, and S.~Diggavi, ``Encoded distributed optimization,'' in
  \emph{2017 IEEE International Symposium on Information Theory (ISIT)}, 2017,
  pp. 2890--2894.

\bibitem{KSDY19}
\BIBentryALTinterwordspacing
C.~Karakus, Y.~Sun, S.~Diggavi, and W.~Yin, ``{Redundancy Techniques for
  Straggler Mitigation in Distributed Optimization and Learning},''
  \emph{Journal of Machine Learning Research}, vol.~20, no.~72, pp. 1--47,
  2019. [Online]. Available: \url{http://jmlr.org/papers/v20/18-148.html}
\BIBentrySTDinterwordspacing

\bibitem{SKD18}
M.~Showkatbakhsh, C.~Karakus, and S.~Diggavi, ``{Privacy-Utility Trade-off of
  Linear Regression under Random Projections and Additive Noise},'' in
  \emph{2018 IEEE International Symposium on Information Theory (ISIT)}.\hskip
  1em plus 0.5em minus 0.4em\relax IEEE, 2018, pp. 186--190.

\bibitem{LSM23}
H.-P. Liu, M.~Soleymani, and H.~Mahdavifar, ``{Differentially Private Coded
  Computing},'' in \emph{IEEE International Symposium on Information Theory
  (ISIT)}, 2023, pp. 2189--2194.

\bibitem{YA19}
Q.~Yu and A.~S. Avestimehr, ``{Harmonic Coding: An Optimal Linear Code for
  Privacy-Preserving Gradient-Type Computation},'' in \emph{2019 IEEE
  International Symposium on Information Theory (ISIT)}.\hskip 1em plus 0.5em
  minus 0.4em\relax IEEE, 2019, pp. 1102--1106.

\bibitem{FJR15}
M.~Fredrikson, S.~Jha, and T.~Ristenpart, ``{Model Inversion Attacks that
  Exploit Confidence Information and Basic Countermeasures},'' in
  \emph{Proceedings of the 22nd ACM SIGSAC conference on computer and
  communications security}, 2015, pp. 1322--1333.

\bibitem{SSSS17}
R.~Shokri, M.~Stronati, C.~Song, and V.~Shmatikov, ``{Membership Inference
  Attacks Against Machine Learning Models},'' in \emph{2017 IEEE symposium on
  security and privacy (SP)}.\hskip 1em plus 0.5em minus 0.4em\relax IEEE,
  2017, pp. 3--18.

\bibitem{CPH23b}
N.~Charalambides, M.~Pilanci, and A.~O. Hero~III, ``{Gradient Coding through
  Iterative Block Leverage Score Sampling},'' \emph{arXiv preprint
  arXiv:2308.03096}, 2023.

\bibitem{KL14}
J.~Katz and Y.~Lindell, \emph{Introduction to modern cryptography}.\hskip 1em
  plus 0.5em minus 0.4em\relax Chapman and Hall/CRC, 2014.

\bibitem{ERNM22}
\BIBentryALTinterwordspacing
A.~Eshragh, F.~Roosta, A.~Nazari, and M.~W. Mahoney, ``{LSAR: Efficient
  Leverage Score Sampling Algorithm for the Analysis of Big Time Series
  Data},'' \emph{Journal of Machine Learning Research}, vol.~23, no.~22, pp.
  1--36, 2022. [Online]. Available:
  \url{http://jmlr.org/papers/v23/20-247.html}
\BIBentrySTDinterwordspacing

\bibitem{Elf80}
T.~Elfving, ``{Block-Iterative Methods for Consistent and Inconsistent Linear
  Equations},'' \emph{Numerische Mathematik}, vol.~35, no.~1, pp. 1--12, 1980.

\bibitem{Gut06}
M.~H. Gutknecht, ``{Block Krylov Space Methods for Linear Systems with Multiple
  Right-hand Sides: An Introduction},'' \emph{Modern Mathematical
  Models,Methods and Algorithms for Real World Systems}, 2006.

\bibitem{NT14}
{Needell, Deanna and Tropp, Joel A}, ``{Paved with Good Intentions: Analysis of
  a Randomized Block Kaczmarz Method},'' \emph{Linear Algebra and its
  Applications}, vol. 441, pp. 199--221, 2014.

\bibitem{RN20}
E.~Rebrova and D.~Needell, ``{On block Gaussian sketching for the Kaczmarz
  method},'' \emph{Numerical Algorithms}, pp. 1--31, 2020.

\bibitem{DGSX12}
O.~Dekel, R.~Gilad-Bachrach, O.~Shamir, and L.~Xiao, ``{Optimal Distributed
  Online Prediction Using Mini-Batches},'' \emph{Journal of Machine Learning
  Research}, vol.~13, no.~1, 2012.

\bibitem{Bub15}
\BIBentryALTinterwordspacing
S.~Bubeck, ``{Convex Optimization: Algorithms and Complexity},''
  \emph{Foundations and Trends® in Machine Learning}, vol.~8, no. 3-4, pp.
  231--357, 2015. [Online]. Available:
  \url{http://dx.doi.org/10.1561/2200000050}
\BIBentrySTDinterwordspacing

\bibitem{RSS12}
A.~Rakhlin, O.~Shamir, and K.~Sridharan, ``{Making Gradient Descent Optimal for
  Strongly Convex Stochastic Optimization},'' \emph{arXiv preprint
  arXiv:1109.5647}, 2011.

\bibitem{OJXE18}
U.~Oswal, S.~Jain, K.~S. Xu, and B.~Eriksson, ``{Block CUR: Decomposing
  Matrices Using Groups of Columns},'' in \emph{Joint European Conference on
  Machine Learning and Knowledge Discovery in Databases}.\hskip 1em plus 0.5em
  minus 0.4em\relax Springer, 2018, pp. 360--376.

\bibitem{DMM06}
P.~Drineas, M.~W. Mahoney, and S.~Muthukrishnan, ``{Sampling Algorithms for
  $\ell_2$ Regression and Applications},'' in \emph{{Proceedings of the
  seventeenth annual ACM-SIAM Symposium on Discrete Algorithms}}, 2006, pp.
  1127--1136.

\bibitem{DMMW12}
P.~Drineas, M.~Magdon-Ismail, M.~W. Mahoney, and D.~P. Woodruff, ``Fast
  approximation of matrix coherence and statistical leverage,'' \emph{Journal
  of Machine Learning Research}, vol.~13, no. Dec, pp. 3475--3506, 2012.

\bibitem{JHD21}
M.~Jauch, P.~D. Hoff, and D.~B. Dunson, ``{Monte Carlo simulation on the
  Stiefel manifold via polar expansion},'' \emph{Journal of Computational and
  Graphical Statistics}, vol.~30, no.~3, pp. 622--631, 2021.

\bibitem{Osg09}
B.~Osgood, ``{The Fourier Transform and its Applications},'' \emph{Stanford
  University, Lecture Notes}, 2009.

\bibitem{Yas12}
\BIBentryALTinterwordspacing
Y.~Chikuse, \emph{Statistics on Special Manifolds}, ser. Lecture Notes in
  Statistics.\hskip 1em plus 0.5em minus 0.4em\relax Springer New York, 2012.
  [Online]. Available: \url{https://books.google.com.cy/books?id=7lX1BwAAQBAJ}
\BIBentrySTDinterwordspacing

\bibitem{Str69}
V.~Strassen, ``Gaussian elimination is not optimal,'' \emph{Numerische
  mathematik}, vol.~13, no.~4, pp. 354--356, 1969.

\bibitem{GKCMR20}
V.~Gupta, S.~Kadhe, T.~Courtade, M.~W. Mahoney, and K.~Ramchandran,
  ``{OverSketched Newton: Fast Convex Optimization for Serverless Systems},''
  in \emph{2020 IEEE International Conference on Big Data (Big Data)}.\hskip
  1em plus 0.5em minus 0.4em\relax IEEE, 2020, pp. 288--297.

\bibitem{WXXZ23}
V.~V. Williams, Y.~Xu, Z.~Xu, and R.~Zhou, ``{New Bounds for Matrix
  Multiplication: from Alpha to Omega},'' \emph{arXiv preprint
  arXiv:2307.07970}, 2023.

\bibitem{Gen09}
C.~Gentry, ``{A Fully Homomorphic Encryption Scheme},'' Ph.D. dissertation, PhD
  thesis. Stanford University, 2009.

\bibitem{Gen09a}
------, ``{Fully Homomorphic Encryption Using Ideal Lattices},'' in
  \emph{Proceedings of the forty-first annual ACM symposium on Theory of
  computing}, 2009, pp. 169--178.

\bibitem{BGV14}
Z.~Brakerski, C.~Gentry, and V.~Vaikuntanathan, ``{(Leveled) Fully Homomorphic
  Encryption without Bootstrapping},'' \emph{ACM Transactions on Computation
  Theory (TOCT)}, vol.~6, no.~3, pp. 1--36, 2014.

\bibitem{SGAM19}
J.~So, B.~Guler, A.~S. Avestimehr, and P.~Mohassel, ``{CodedPrivateML: A Fast
  and Privacy-Preserving Framework for Distributed Machine Learning},''
  \emph{arXiv preprint arXiv:1902.00641}, 2019.

\bibitem{khan2015hill}
F.~H. Khan, R.~Shams, F.~Qazi, and D.~Agha, ``{Hill Cipher Key Generation
  Algorithm by using Orthogonal Matrix},'' \emph{Int. J. Innov. Sci. Mod. Eng},
  vol.~3, no.~3, pp. 5--7, 2015.

\bibitem{AKHK17}
J.~Ahmad, M.~A. Khan, S.~O. Hwang, and J.~S. Khan, ``A compression sensing and
  noise-tolerant image encryption scheme based on chaotic maps and orthogonal
  matrices,'' \emph{Neural computing and applications}, vol.~28, no.~1, pp.
  953--967, 2017.

\bibitem{AMAK18}
J.~Ahmad, M.~A. Khan, F.~Ahmed, and J.~S. Khan, ``A novel image encryption
  scheme based on orthogonal matrix, skew tent map, and xor operation,''
  \emph{Neural Computing and Applications}, vol.~30, no.~12, pp. 3847--3857,
  2018.

\bibitem{Cac14}
Y.~C. Santana, ``{Orthogonal Matrix in Cryptography},'' \emph{arXiv preprint
  arXiv:1401.5787}, 2014.

\bibitem{AID19}
\BIBentryALTinterwordspacing
S.~Alhassan, M.~M. Iddrisu, and M.~I. Daabo, ``Perceptual video encryption
  using orthogonal matrix,'' \emph{International Journal of Computer
  Mathematics: Computer Systems Theory}, vol.~4, no. 3-4, pp. 129--139, 2019.
  [Online]. Available: \url{https://doi.org/10.1080/23799927.2019.1645210}
\BIBentrySTDinterwordspacing

\bibitem{reddy2018image}
K.~M. Reddy, A.~Itagi, S.~Dabas, and B.~K. Prakash, ``{Image Encryption Using
  Orthogonal Hill Cipher Algorithm},'' \emph{International Journal of
  Engineering \& Technology}, vol.~7, no. 4.10, pp. 59--63, 2018.

\bibitem{Mur12}
K.~P. Murphy, \emph{{Machine Learning: A Probabilistic Perspective}}.\hskip 1em
  plus 0.5em minus 0.4em\relax MIT Press, 2012.

\bibitem{Mah16}
M.~W. Mahoney, ``{Lecture Notes on Randomized Linear Algebra},'' \emph{arXiv
  preprint arXiv:1608.04481}, 2016.

\bibitem{Wan15}
S.~Wang, ``A practical guide to randomized matrix computations with matlab
  implementations,'' \emph{arXiv preprint arXiv:1505.07570}, 2015.

\end{thebibliography}
% - - - - - - - - - - - - - - -

% - - - - - - - - - - - - - - -
\section{Proofs of Section \ref{bl_orth_sk_sec}}

\subsection{Subsection \ref{distr_grad_desc}}

Note that in Lemma \ref{lemma_exp}:
$$ \E\left[\Ombwt_{[t]}^T\Ombwt_{[t]}\right]=\Ib_N \quad \implies \quad \E\left[\Sb_{[t]}^T\Sb_{[t]}\right]=\Ib_N, $$
as
$$ \E\left[\Sb_{[t]}^T\Sb_{[t]}\right]=\Pibold^T\E\left[\Ombwt_{[t]}^T\Ombwt_{[t]}\right]\Pibold=\Pibold^T\Pibold=\Ib_N. $$
We provide both derivations separately in order to convey the respective importance behind the use of the Lemma in subsequent arguments, even though the main idea is the same. Furthermore, the proof of Theorem \ref{GC_SGD_thm} is very similar to that of Lemma \ref{lemma_exp}.

\begin{proof}{[Lemma \ref{lemma_exp}]}
The only difference in $\SbPi^{[t]}$ at each iteration, is $\Scal^{[t]}$ and $\Ombwt_{[t]}$. This corresponds to a uniform random selection of $q$ out of $K$ batches of the data which determine the gradient at iteration $t$ --- all blocks are scaled by the same factor $\sqrt{K/q}$ in $\Ombwt_{[t]}$. Let $\Qcal$ be the set of all subsets of $\N_K$ of size $q$. Then
\begin{align*}
  \E\big[\Sb_{[t]}^T\Sb_{[t]}\big] &= \sum_{\Scal^{[t]}\in\Qcal}\frac{1}{{K\choose q}}\cdot\left(\Sb_{[t]}\cdot\Sb_{[t]}\right)\\
  &= \frac{1}{{K\choose q}}\sum_{\Scal^{[t]}\in\Qcal}\sum_{i\in\Scal^{[t]}}{\left(\sqrt{K/q}\right)^2}\cdot\Pibold_{(\K_i)}^T\Pibold_{(\K_i)}\\
  &= \frac{{K-1\choose q-1}}{{K\choose q}}\sum_{i=1}^K\frac{K}{q}\cdot\Pibold_{(\K_i)}^T\Pibold_{(\K_i)}\\
  &= \frac{{K-1\choose q-1}\cdot\frac{K}{q}}{{K\choose q}}\sum_{i=1}^K\Pibold_{(\K_i)}^T\Pibold_{(\K_i)}\\
  &= \Pibold^T\Pibold\\
  &= \Ib_N
\end{align*}
where ${K-1\choose q-1}$ is the number of sets in $\Qcal$ which include $i$, for each $i\in\N_K$. This completes the first part of the proof.

Note that the sampling and rescaling matrices $\Ombwt_{[t]}$ of Algorithm \ref{alg_orthog_sketch}, may also be expressed as
$$ \Ombwt_{[t]} = \sqrt{K/q}\cdot \sum_{\iota\in\Scal^{[t]}}\Ib_{(\K_\iota)} . $$
Further notice that $\Ombwt_{[t]}$'s corresponding sampling and rescaling matrix of size $N\times N$, which appears in the expansion the objective function \eqref{x_til_pr_lr}, is
\begin{align*}
  \Ombwt_{[t]}^T\Ombwt_{[t]} &= \left(\sqrt{K/q}\right)^2\cdot\sum_{\iota\in\Scal^{[t]}}\big(\Ib_{(\K_\iota)}\big)^T\Ib_{(\K_\iota)}\\
  &= \frac{K}{q}\cdot\sum_{j\in\bigsqcup\limits_{\iota\in\Scal^{[t]}}\K_\iota} \eb_{j}\eb_j^{T}.
\end{align*}

Let $\Bcal$ denote the set of all possible block sampling and rescaling matrices of size $r\times N$, which sample $q$ out of $K$ blocks. For $\Phib\in\Bcal$, by $\Ib_{(\K_\iota)}\subseteq\Phib$ we denote the condition that $\Ib_{(\K_\iota)}$ is a submatrix of $\Phib$. Note that for each $\iota\in\N_K$, there are ${K-1\choose q-1}$ matrices in $\Bcal$ which have $\Ib_{(\K_\iota)}$ as a submatrix. For our setup, we then have
\begin{align*}
  \E\left[\Ombwt_{[t]}^T\Ombwt_{[t]}\right] &= \sum_{\Phib\in\Bcal}\frac{1}{{K\choose q}}\cdot\left(\Phib^T\Phib\right)\\
  &= \frac{{\left(\sqrt{K/q}\right)^2}}{{K\choose q}}\cdot\sum_{\Phib\in\Bcal}\sum_{\Ib_{(\K_\iota)}\subseteq\Phib}\big(\Ib_{(\K_\iota)}\big)^T\Ib_{(\K_\iota)}\\
  &= \frac{{K-1\choose q-1}\cdot(K/q)}{{K\choose q}}\cdot\sum_{\iota\in\N_K}\big(\Ib_{(\K_\iota)}\big)^T\Ib_{(\K_\iota)}\\
  &= \sum_{\iota\in\N_K}\big(\Ib_{(\K_\iota)}\big)^T\Ib_{(\K_\iota)}\\
  &= \sum_{j\in\N_N}\eb_j^T\eb_j\\
  &= \Ib_N
\end{align*}
and the proof is complete.

\end{proof}

\begin{proof}{[Theorem \ref{GC_SGD_thm}]}
The only difference in $\SbPi^{[t]}$ at each iteration, is $\Scal^{[t]}$ and $\Ombwt_{[t]}$. This corresponds to a uniform random selection of $q$ out of $K$ batches of the data which determine the gradient at iteration $t$ --- all blocks are scaled by the same factor $\sqrt{K/q}$ in $\Ombwt_{[t]}$. By \eqref{gr_update}, the gradient update is equal to that of a batch stochastic steepest descent procedure.

We break up the proof of the second statement by first showing that $\E\left[\gh^{[t]}\right]=\gt^{[t]}$; for $\gt$ the gradient in the basis $\Pibold\Ub$, and then showing that $\E\left[\gt^{[t]}\right]=\frac{q}{K}\cdot g_{ls}^{[t]}$.

Let $\Qcal$ be the set of all subsets of $\N_K$ of size $q$, $\gh_{\Scal^{[t]}}$ the gradient determined by the index set $\Scal^{[t]}$, and $\gt_i^{[t]}$ the respective partial gradients at iteration $t$. Then
\begin{align*}
  \E\left[\gh^{[t]}\right] &= \sum_{\Scal^{[t]}\in\Qcal}\frac{1}{{K\choose q}}\cdot\gh_{\Scal^{[t]}}\\
  &= \frac{1}{{K\choose q}}\sum_{\Scal^{[t]}\in\Qcal}\sum_{i\in\Scal^{[t]}}{\left(\sqrt{K/q}\right)^2}\cdot\gt_i^{[t]}\\
  &= \frac{{K-1\choose q-1}}{{K\choose q}}\sum_{i=1}^K\frac{K}{q}\cdot\gt_i^{[t]}\\
  &= \sum_{i=1}^K\gt_i^{[t]}\\
  &= \gt^{[t]}
\end{align*}
where ${K-1\choose q-1}$ is the number of sets in $\Qcal$ which include $i$, for each $i\in\N_K$.

We denote the resulting partial gradient on the sampled index set $\Scal^{[t]}$ of the gradient on \eqref{x_star_pr_lr} at iteration $t$; \textit{i.e.} $g_{ls}^{[t]}$, by $g_{\Scal^{[t]}}$, and the individual partial gradients by $g_i^{[t]}$. Using the same notation as above, we get that
\begin{align*}
  \E\left[\gt^{[t]}\right] &= \sum_{\Scal^{[t]}\in\Qcal}\frac{1}{{K\choose q}}\cdot g_{\Scal^{[t]}}\\
  &= \frac{1}{{K\choose q}}\sum_{\Scal^{[t]}\in\Qcal}\sum_{i\in\Scal^{[t]}}g_i^{[t]}\\
  &= \frac{{K-1\choose q-1}}{{K\choose q}}\sum_{i=1}^Kg_i^{[t]}\\
  &= \frac{q}{K}\cdot\sum_{i=1}^K\gt_i^{[t]}\\
  &= \frac{q}{K}\cdot g^{[t]}
\end{align*}
which completes the proof.
\end{proof}

\begin{proof}{[Lemma \ref{eq_opt_sols}]}
Since $\Pibold$ is an orthonormal matrix, the solution of the least squares problem with the objective $L_{\Gb}(\Ab,\bb;\xb)$ is equal to the optimal solution \eqref{x_star_pr_lr}, as
\begin{align*}
  \xbh &= \arg\min_{\xb\in\R^d}\|\Gb(\Ab\xb-\bb)\|_2^2 \\
  &= \arg\min_{\xb\in\R^d}\|\Pibold(\Ab\xb-\bb)\|_2^2 \\
  &= \arg\min_{\xb\in\R^d}\|\Ab\xb-\bb\|_2^2\\
  &= \xb_{ls}^{\star}.
\end{align*}
\end{proof}

\begin{proof}{[Corollary \ref{eq_SSD_dor}]}
We prove this by induction. From our assumptions we have a fixed starting point $\xb^{[0]}$, for which $\xbh^{[0]}=\xb^{[0]}$. Our base case is therefore $\E[\xbh^{[0]}]=\E[\xb^{[0]}]=\xb^{[0]}$. For the inductive hypothesis, we assume that $\E[\xbh^{[\tau]}]=\xb^{[\tau]}$ for $\tau\in\N$.

It then follows that at step $\tau+1$ we have
\begin{align*}
  \E\big[\xbh^{[\tau+1]}\big] &= \E\big[\xbh^{[\tau]}-\xih_{\tau}\cdot\gh^{[\tau]}\big]\\
  &= \E\big[\xbh^{[\tau]}\big]-\frac{K}{q}\cdot\xi_\tau\cdot\E\big[\gh^{[\tau]}\big]\\
  &= \xb^{[\tau]}-\frac{q}{K}\cdot\left(\frac{K}{q}\cdot\xi_\tau\right)\cdot g_{ls}^{[\tau]}\\
  &= \xb^{[\tau]}-\xi_\tau\cdot g_{ls}^{[\tau]}\\
  &= \xb^{[\tau+1]}
\end{align*}
which completes the inductive step.
\end{proof}

% - - - - - - - - - - - - - - -
\subsection{Subsection \ref{subsp_emb_alg1}}

In this appendix, we provide the proofs of Lemma \ref{bd_block_lvg_Unif} and Theorem \ref{subsp_emb_thm_Unif}. First, we need to provide Lemmas \ref{Lemma_exp_lev_i} and \ref{norm_lvg_bd_lemma}, and Hoeffding's inequality; which we use to prove the latter Lemma. Throughout this subsection, by $\ell_i$ we denote the $i^{th}$ leverage score of $\Pibold\Ab$ for $\Pibold$ a random orthonormal matrix, \textit{i.e.}
\begin{equation}
\label{lvg_sc_expr}
  \ell_i=\|\Ubt_{(i)}\|_2^2=\|\eb_i^T\Ubt\|_2^2=\eb_i^T\Ubt\Ubt^T\eb_i
\end{equation}
where $\Ubt=\Pibold\Ub$; for $\Ub$ the reduced left orthonormal matrix of $\Ab$. By $\eb_i$ we denote the $i^{th}$ standard basis vector of $\R^N$.

\begin{Lemma}
\label{Lemma_exp_lev_i}
For each $i\in\N_N$, we have $\E[\ell_i]=\frac{d}{N}$.
\end{Lemma}

\begin{proof}
By \eqref{lvg_sc_expr}, we have
\begin{align*}
  \E[\ell_i] &= \E\left[\tr(\eb_i^T\Ubt\Ubt^T\eb_i)\right]\\
  &= \E\left[\tr(\eb_i\eb_i^T\cdot\Ubt\Ubt^T)\right]\\
  &= \sum_{j=1}^N\frac{1}{N}\cdot\tr(\eb_j\eb_j^T\cdot\Ubt\Ubt^T)\\
  &= \frac{1}{N}\cdot\tr\left(\sum_{j=1}^N\eb_j\eb_j^T\cdot\Ubt\Ubt^T\right)\\
  &= \frac{1}{N}\cdot\tr\left(\Ib_N\cdot\Ubt\Ubt^T\right)\\
  &= \frac{1}{N}\cdot\tr\left(\Ubt\Ubt^T\right)\\
  &= \frac{d}{N}.
\end{align*}
\end{proof}

Let $\ellb_i$ denote the $i^{th}$ normalized leverage score, \textit{i.e.} $\ellb_i=\frac{\ell_i}{d}$. The $\iota^{th}$ normalized block leverage score of $\Ab$ is denoted by $\ellg_\iota$, \textit{i.e.}
\begin{equation}
\label{norm_bl_lvg_sc_expr}
  \ellg_\iota=\frac{1}{d}\cdot\|\Ib_{(\K_\iota)}\Ubt\|_F^2=\frac{1}{d}\cdot\Big(\sum_{j\in\K_\iota}\ell_j\Big)=\sum_{j\in\K_\iota}\ellb_j.
\end{equation}
To prove Lemma \ref{bd_block_lvg_Unif}, we first recall Hoeffding's inequality.

\begin{Thm}[Hoeffding's Inequality, \cite{Mah16}]
\label{Hoef_in}
Let $\{X_i\}_{i=1}^m$ be independent random variables such that $X_i\in[a_i,b_i]$ for all $i\in\N_m$, and let $X=\sum_{i=1}^mX_i$. Then
\begin{equation*}
\label{Hoeffding_id}
  \Pr\left[\big|X-\E[X]\big|\geqslant t\right] \leqslant 2\cdot\exp\left\{\frac{-2t^2}{\sum_{j=1}^m(a_i-b_i)^2}\right\}.
\end{equation*}
\end{Thm}

\begin{Lemma}
\label{norm_lvg_bd_lemma}
The normalized leverage scores $\{\ellb_i\}_{i=1}^N$ of $\Pibold\Ab$ satisfy
$$ \Pr\left[|\ellb_i-1/N|<\rho\right] > 1-2\cdot e^{-2\rho^2} $$
for any $\rho>0$.
\end{Lemma}

\begin{proof}
We know that $\ell_i\in[0,d]$ for each $i\in\N_N$, thus $\ellb_i\in[0,1]$ for each $i$. By Lemma \ref{Lemma_exp_lev_i}, it follows that
$$ \E[\ellb_i]=\E[\ell_i/d]=\frac{1}{d}\cdot\E[\ell_i]=\frac{1}{N}. $$
Now, fix a constant $\rho>0$. By applying Theorem \ref{Hoef_in} with $m=1$, we get
\begin{equation*}
\label{norm_lvg_bd}
  \Pr\left[|\ellb_i-1/N|\geqslant\rho\right] \leqslant 2\cdot e^{-2\rho^2}
\end{equation*}
thus
\begin{equation*}
  \Pr\left[|\ellb_i-1/N|<\rho\right] > 1-2\cdot e^{-2\rho^2}.
\end{equation*}
\end{proof}

Next, we complete the proof of the ``flattening Lemma of block leverage scores'' (Lemma \ref{bd_block_lvg_Unif}).

\begin{proof}{[Lemma \ref{bd_block_lvg_Unif}]}
To show that the two probability events of expression \eqref{bl_lvg_unif_flatten} are equal, note that:
\begin{enumerate}
  \item $\ellg_\iota-\frac{1}{K}<\frac{N}{K}\rho \ \ \iff \ \ \ellg_\iota<(1+N\rho)\frac{1}{K}$
  \item $\frac{1}{K}-\ellg_\iota<\frac{N}{K}\rho \ \ \iff \ \ \ellg_\iota>(1-N\rho)\frac{1}{K}$.
\end{enumerate}
By combining the two inequalities, we conclude that
\begin{equation}
\label{approx_unif_error}
  (1-N\rho)\cdot\frac{1}{K} < \ellg_\iota < (1+N\rho)\cdot\frac{1}{K} \ \ \iff \ \ \ellg_\iota<_{N\rho} 1/K.
\end{equation}
By Lemma \ref{norm_lvg_bd_lemma}, it follows that
\begin{align*}
  \Pr\left[\big|\ellg_\iota-1/K\big|<\tau\rho\right] &> \Pr\left[\bigwedge_{j\in\K_\iota}\big\{|\ellb_i-1/N|<\rho\big\}\right]\\
  &> \left(1-2\cdot e^{-2\rho^2}\right)^\tau\\
  &\overset{\Join}{\approx} 1-2\tau\cdot e^{-2\rho^2}
\end{align*}
where in $\Join$ we applied the binomial approximation. By substituting $\rho\geqslant\sqrt{\log(2\tau/\delta)/2}$, we get
\begin{align*}
  e^{-2\rho^2} &\leqslant e^{-2\frac{\log(2\tau/\delta)}{2}}\\
  &= e^{-\log(2\tau/\delta)}\\
  &= e^{\log(\delta/2\tau)}\\
  &= \delta/2\tau,
\end{align*}
thus $2\tau\cdot e^{-2\rho^2}\leqslant\delta$; and $1-2\tau\cdot e^{-2\rho^2} \geqslant 1-\delta$. In turn, this implies that $\Pr\left[\big|\ellg_\iota-1/K\big|<\tau\rho\right]>1- \delta$.
\end{proof}

The proof of Theorem \ref{subsp_emb_thm_Unif} is a direct consequence of Lemma \ref{bd_block_lvg_Unif} and Theorem \ref{subsp_emb_thm_lvg}. In our statement we make the assumption that $\ellg_\iota=1/K$ for all $\iota$, though this is not necessarily the case, as Lemma \ref{bd_block_lvg_Unif} permits a small deviation. For $\rho\gets\epsilon$, we consider $\epsilon\ll 1/N$ so that the `$N\epsilon$ multiplicative error' in \eqref{approx_unif_error} is small. We note that \cite[Theorem 1]{CPH23b} considers sampling according to \textit{approximate} block leverage scores.%One could generalize Theorem \ref{subsp_emb_thm_lvg} to accommodate sampling according to \textit{approximate} block leverage scores, \textit{e.g.} \cite{DMMW12}. This is not a focus of our work.

\begin{Thm}[$\ell_2$-s.e. of the block leverage score sampling sketch, \cite{CPH23b}]
\label{subsp_emb_thm_lvg}
The sketching matrix $\Sbwt$ constructed by sampling blocks of $\Ab$ with replacement according to their normalized block leverage scores $\{\ellg_\iota\}_{\iota=1}^K$ and rescaling each sampled block by $\sqrt{1\big/\big(q\ellg_\iota\big)}$, guarantees a $\ell_2$-s.e. of $\Ab$; as defined in \eqref{subsp_emb_id}. Specifically, for $\delta>0$ and $q=\Theta(\frac{d}{\tau}\log{(2d/\delta)}/\epsilon^2)$:
\begin{equation*}
  \Pr\big[\|\Ib_d-\Ub^T\Sbwt^T\Sbwt \Ub\|_2\leqslant\epsilon\big]\geqslant 1-\delta.
\end{equation*}
\end{Thm}

Before we prove Proposition \ref{appr_GC_prop}, we first derive \eqref{appr_GC_error}. In \cite{CPE17}, the optimal decoding vector of an approximate GCS was defined as
\begin{equation}
\label{opt_dec_err}
  \ab_\I^\star = \arg\min_{\ab\in\R^{1\times q}}\big\{\|\ab\Gb_{(\I)}-\vec{\bold{1}}\|_2^2\big\}.
\end{equation}
In the case where $q\geqslant K$, it follows that $\ab_\I^\star = \vec{\bold{1}}\Gb_{(\I)}^{\dagger}$. The error can then be quantified as
$$ \err(\Gb_{(\I)}) \coloneqq \|\Ib_K-\Gb_{(\I)}^{\dagger}\Gb_{(\I)}\|_2. $$
The optimal decoding vector \eqref{opt_dec_err} has also been considered in other schemes, \textit{e.g.} \cite{KKR19,SH22}.

Let $\gb^{[t]}$ be the matrix comprised of the transposed exact partial gradients at iteration $t$, \textit{i.e.}
\begin{equation}
\label{g_matrix}
  \gb^{[t]} \coloneqq {\begin{pmatrix} g_1^{[t]} & g_2^{[t]} & \hdots & g_K^{[t]} \end{pmatrix}}^T \in \R^{K\times d}.
\end{equation}
Then, for a GCS $(\Gb,\ab_\I)$ satisfying $\ab_\I\Gb_{(\I)}=\vec{\bold{1}}$ for any $\I$, it follows that $\left(\ab_\I\Gb_{(\I)}\right)\gb^{[t]}=\vec{\bold{1}}\gb^{[t]}=\big(g^{[t]}\big)^T$. Hence, the gradient can be recovered exactly. Considering an optimal approximate scheme $(\Gb,\ab_\I^\star)$ which recovers the gradient estimate $\grave{g}^{[t]}=\left(\ab_\I^\star\Gb_{(\I)}\right)\gb^{[s]}$, the error in the gradient approximation is
\begin{align*}
  \big\|g^{[s]}-\grave{g}^{[s]}\big\|_2 &= \left\|\big(\vec{\bold{1}}-\ab_\I^\star\Gb_{(\I)}\big)\gb^{[s]}\right\|_2\\
  &= \left\|\vec{\bold{1}}\big(\Ib_K-\Gb_{(\I)}^{\dagger}\Gb_{(\I)}\big)\gb^{[s]}\right\|_2\\
  &\leqslant \|\vec{\bold{1}}\|_2\cdot \left\|\Ib_K-\Gb_{(\I)}^{\dagger}\Gb_{(\I)}\right\|_2\cdot\big\|\gb^{[s]}\big\|_2\\
  &\overset{\pounds}{\leqslant} \sqrt{K}\cdot \left\|\Ib_K-\Gb_{(\I)}^{\dagger}\Gb_{(\I)}\right\|_2\cdot\big\|g^{[s]}\big\|_2\\
  &\overset{\$}{\leqslant} 2\sqrt{K}\cdot\underbrace{\left\|\Ib_K-\Gb_{(\I)}^{\dagger}\Gb_{(\I)}\right\|_2}_{\err(\Gb_{(\I)})}\cdot\|\Ab\|_2\cdot\|\Ab\xb^{[s]}-\bb\|_2
\end{align*}
where $\pounds$ follows from the facts that $\|\gb^{[s]}\|_2\leqslant\|g^{[s]}\|_2$ and $\|\vec{\bold{1}}\|_2=\sqrt{K}$, and $\$$ from \eqref{gr_ls} and sub-multiplicativity of matrix norms. This concludes the derivation of \eqref{appr_GC_error}.

\begin{proof}{[Proposition \ref{appr_GC_prop}]}
%Let $\Gamma_t=g^t-\gh^{[t]}$,
Let $\gh^{[t]}$ be the approximated gradient of our scheme at iteration $t$. Since we are considering linear regression, it follows that
\begin{align*}
  \big\|g^{[t]}-\gh^{[t]}\big\|_2 &= 2\|\Ab^T(\Ab\xb^{[t]}-\bb)-\Ab^T(\SbPi^T\SbPi)(\Ab\xb^{[t]}-\bb)\|_2 \\
  &= 2\|\Ab^T(\Ib_N-\SbPi^T\SbPi)(\Ab\xb^{[t]}-\bb)\|_2\\
  &\leqslant 2\|\Ab\|_2\cdot\|\Ib_N-\SbPi^T\SbPi\|_2\cdot\|\Ab\xb^{[t]}-\bb\|_2\\
  &= 2\|\Ab\|_2\cdot\|\Ub^T(\Ib_N-\SbPi^T\SbPi)\Ub\|_2\cdot\|\Ab\xb^{[s]}-\bb\|_2\\
  &= 2\|\Ab\|_2\cdot\|\Ib_d-\Ub^T\SbPi^T\SbPi\Ub\|_2\cdot\|\Ab\xb^{[s]}-\bb\|_2\\
  &\overset{\flat}{\leqslant} 2\epsilon\cdot\|\Ab\|_2\cdot\|\Ab\xb^{[t]}-\bb\|_2
\end{align*}
where in $\flat$ we invoked the fact that $\SbPi$ satisfies \eqref{subsp_emb_id}. Our approximate GC approach therefore (w.h.p.) satisfies \eqref{appr_GC_error}, with $\err(\Gb_{(\I)})=\epsilon/\sqrt{K}$
\end{proof}

% - - - - - - - - - - - - - - -
\section{Proofs of Section \ref{block_SRHT_sec}}

In this appendix, we present two lemmas which we use to bound the entries of $\Vbh\coloneqq\Hbh\Db\Ub$, and its \textit{leverage scores} $\ell_i\coloneqq\|\Vbh_{(i)}\|_2^2$, for which $\sum_{i=1}^N\ell_i=d$. Leverage scores induce a sampling distribution which has proven to be useful in linear regression \cite{DMMW12,Woo14,Mah16,Wan15} and GC \cite{CPH20a}. From these lemmas, we deduce that the leverage scores of $\Hbh\Db\Ab$ are close to being uniform, implying that the \textit{block leverage scores}\cite{OJXE18,CPH20a} are also uniform, which is precisely what Lemma \ref{bd_lvg_Had} states.

Lemma \ref{fl_lem} is a variant of the Flattening Lemma \cite{AC06,Mah16}, a key result to Hadamard based sketching algorithms, which justifies uniform sampling. In the proof, we make use of the Azuma-Hoeffding inequality; a concentration result for the values of martingales that have bounded differences. We also recall a matrix Chernoff bound, which we apply to prove our $\ell_2$-s.e. guarantees. Finally, we present proofs of Lemma \ref{bd_block_lvg_Had} and Theorem \ref{subsp_emb_thm_SRHT}.

\begin{Lemma}[Azuma-Hoeffding Inequality, \cite{Mah16}]
\label{Az_Hoef_in}
For zero mean random variable $Z_i$ (or $Z_0,Z_1,\cdots,Z_m$ a martingale sequence of random variables), bounded above by $|Z_i|\leqslant \beta_i$ for all $i$ with probability 1, we have
$$ \Pr\bigg[\big|\sum_{j=0}^m Z_j\big|>t\bigg] \leqslant 2\exp\left\{\frac{-t^2}{2\cdot\big(\sum_{j=0}^m\beta_j^2\big)}\right\}. $$
\end{Lemma}

\begin{Thm}[Matrix Chernoff Bound, {\cite[Fact 1]{Woo14}}]
\label{matr_Chern}
  Let $\Xb_1,\cdots,\Xb_q$ be independent copies of a symmetric random matrix $\Xb\in\R^{d\times d}$, with $\E[\Xb]=\bold{0}, \|\Xb\|_2\leqslant \gamma$, $\|\E[\Xb^T\Xb]\|_2\leqslant \sigma^2$. Let $\Zb=\frac{1}{q}\sum_{i=1}^q\Xb_i$. Then, $\forall\epsilon>0$:
  \begin{equation}
  \label{matr_Chern_expr}
    \Pr\Big[\|\Zb\|_2>\epsilon\Big]\leqslant2d\cdot\exp\left(-\frac{q\epsilon^2}{\sigma^2+\gamma\epsilon/3}\right).
  \end{equation}
\end{Thm}

\begin{Lemma}[Flattening Lemma]
\label{fl_lem}
  For $\yb\in\R^N$ a fixed (orthonormal) column vector of $\Ub$, and $\Db\in\{0,\rpm1\}^{N\times N}$ with random equi-probable diagonal entries of $\rpm1$, we have:
  %$$ \Pr\left[\|\Hbh\Db\cdot\yb\|_\infty> C\sqrt{\log(Nd/\delta)/N}\right]\leqslant\frac{\delta}{2d} $$
  \begin{equation}
  \label{flat_lem_id}
    \Pr\left[\|\Hbh\Db\cdot\yb\|_\infty> C\sqrt{\log(Nd/\delta)/N}\right]\leqslant\frac{\delta}{2d}
  \end{equation}
  for $0<C\leqslant \sqrt{2+\log(16)/\log(Nd/\delta)}$ a constant.
\end{Lemma}

\begin{proof}{[Lemma \ref{fl_lem}]}
Fix $i$ and define $Z_j=\Hbh_{ij}\Db_{jj}\yb_j$ for each $j\in\N_N$, which are independent random variables. Since $\Db_{jj}=\vec{D}_j$ are i.i.d. entries with zero mean, so are $Z_j$. Furthermore $|Z_j| \leqslant |\Hbh_{ij}|\cdot|\Db_{jj}|\cdot|\yb_j| = \frac{|\yb_j|}{\sqrt{N}}$, and note that
$$ \sum_{j=1}^NZ_j=(\Hbh\Db\yb)_{i}=\sum_{j=1}^N\Hbh_{ij}\Db_{jj}\yb_j=\langle\Hbh_{(i)}\odot \overbrace{\diag(\Db)}^{\vec{D}},\yb\rangle $$
where $\odot$ is the Hadamard product. By Lemma \ref{Az_Hoef_in}
\begin{align}
\label{pr_sum_Zj}
  \Pr\bigg[\Big|\sum_{j=1}^N Z_j\Big|&>\rho\bigg] \leqslant 2\exp\left\{\frac{-\rho^2/2}{\sum_{j=1}^N(\yb_j/\sqrt{N})^2}\right\}\notag\\
  &= 2\exp\left\{\frac{-N\rho^2}{2\cdot\langle\yb,\yb\rangle}\right\} \overset{\flat}{=} 2\cdot e^{-N\rho^2/2}
\end{align}
where $\flat$ follows from the fact that $\yb$ is a column of $\Ub$. By setting $\rho=C\sqrt{\frac{\log(Nd/\delta)}{N}}$, we get
\begin{align*}
  \Pr\left[\Big|\sum_{j=1}^N Z_j\Big|>C\sqrt{\frac{\log(Nd/\delta)}{N}}\right] &\leqslant 2\exp\left\{-\frac{C^2\log(Nd/\delta)}{2}\right\}\\
  &= 2\left(\frac{\delta}{Nd}\right)^{C^2/2} \overset{\natural}{\leqslant} \frac{\delta}{2Nd}
\end{align*}
where $\natural$ follows from the upper bound on $C$. By applying the union bound over all $i\in\N_N$, we attain \eqref{flat_lem_id}.
\end{proof}

\begin{Lemma}
\label{bd_lvg_Had}  % bound the leverage scores of the SRHT
  For all $i\in\N_N$ and $\{\eb_i\}_{i=1}^N$ the standard basis:
  $$ \Pr\left[\sqrt{\ell_i}\leqslant C\sqrt{d\log(Nd/\delta)/N}\right]\geqslant 1-\delta/2 $$
  for $\ell_i=\|\Vbh_{(i)}\|_2^2$ the $i^{th}$ leverage score of $\Vbh=\Hbh\Db\Ub$.
  %$$ \Pr\left[\|\eb_i^T\cdot(\Hbh\Db\Ub)\|_2\leqslant C\sqrt{\frac{d\log(Nd/\delta)}{N}}\right]\geqslant 1-\delta/2 . $$
  %$$ \Pr\left[\|\overbrace{\eb_j^T\cdot(\Hbh\Db\Ub)}^{\Vbh_{(j)}}\|_2\leqslant C\sqrt{\frac{d\log(Nd/\delta)}{N}}\right]\geqslant 1-\delta/2 . $$
\end{Lemma}

\begin{proof}{[Lemma \ref{bd_lvg_Had}]}
It is straightforward that the columns of $\Vbh$ form an orthonormal basis of $\Ab$, thus Lemma \ref{fl_lem} implies that for $j\in\N_d$% fixed
$$ \Pr\left[\|\Vbh\cdot\eb_j\|_\infty> C\sqrt{\log(Nd/\delta)/N}\right] \leqslant \frac{\delta}{2d} . $$
By applying the union bound over all entries of $\Vbh^{(j)}=\Vbh\cdot\eb_j$%$\Vbh\cdot\eb_j=\Vbh^{(j)}$
\begin{equation}
\label{small_entries}
  \Pr\Bigg[\overbrace{|\eb_i^T\cdot\Vbh\cdot\eb_j|}^{|(\Hbh\Db\Ub)_{ij}|}> C\sqrt{\frac{\log(Nd/\delta)}{N}}\Bigg] \leqslant d\cdot\frac{\delta}{2d} = \delta/2.
\end{equation}
%By the definition of the $\ell_2$-norm, we manipulate the argument of the above bound to get
We manipulate the argument of the above bound to obtain
$$ \|\eb_i^{T}\cdot\Vbh\|_2=\Big(\sum_{j=1}^d(\Hbh\Db\Ub)_{ij}^2\Big)^{1/2} > { C}\sqrt{d\cdot\frac{\log(Nd/\delta)}{N}} , $$
which can be viewed as a scaling of the random variable entries of $\Vbh$. The probability of the complementary event is therefore
$$ \Pr\left[\|\eb_i^T\cdot\Vbh\|_2\leqslant C\sqrt{d\log(Nd/\delta)/N}\right]\geqslant 1-\delta/2 $$
and the proof is complete.
\end{proof}

\begin{Rmk}
The complementary probable event of \eqref{small_entries} can be interpreted as `every entry of $\Vbh$ is small in absolute value'.
\end{Rmk}

\begin{proof}{[Lemma \ref{bd_block_lvg_Had}]}
For $\alpha\coloneqq \eta d\cdot\log(Nd/\delta)/N$
$$ \Pr\big[\tilde{\ell}_\iota\leqslant\tau\cdot\alpha\big]>\Pr\big[\{\ell_j\leqslant\alpha:\forall j\in\K_\iota\}\big]\overset{\diamondsuit}{>}(1-\delta/2)^{\tau}$$
where $\eta=C^2$ and $\diamondsuit$ follows from Lemma \ref{bd_lvg_Had}. By the binomial approximation, we have $(1-\delta/2)^{\tau}\approx 1-\tau\delta/2$.
\end{proof}

Define the symmetric matrices
\begin{equation}
\label{def_Xi}
  \Xb_i = \left(\Ib_d-\frac{N}{\tau}\cdot\Vbh_{(\K^i)}^T\Vbh_{(\K^i)}\right) = \left(\Ib_d-K\cdot \Vbh_{(\K^i)}^T\Vbh_{(\K^i)}\right)
\end{equation}
where $\Vbh_{(\K^i)}=\Vbh_{(\K_\iota)}$ is the submatrix of $\Vbh$ corresponding to the $i^{th}$ sampling trial of our algorithm. Let $\Xb$ be the matrix r.v. of which the $\Xb_i$'s are independent copies. Note that the realizations $\Xb_i$ of $\Xb$ correspond to the sampling blocks of the event in \eqref{subsp_emb_id}. To apply Theorem \ref{matr_Chern}, we show that the $\Xb_i$'s have zero mean, and we bound their $\ell_2$-norm and variance. Their $\ell_2$-norms are upper bounded by
\begin{align}
  \|\Xb_i\|_2 &\leqslant \|\Ib_d\|_2+\|\frac{N}{\tau}\cdot \Vbh_{(\K^i)}^T\Vbh_{(\K^i)}\|_2\notag\\
  &= 1+\frac{N}{\tau}\cdot\|\Vbh_{(\K_\iota)}\|_2^2\notag\\
  &\leqslant 1+\frac{N}{\tau}\cdot\max_{\iota\in\N_K}\left\{\|\Ib_{(\K_\iota)}\cdot\Vbh\|_2^2\right\}\notag\\
  &\leqslant 1+\frac{N}{\tau}\cdot\max_{\iota\in\N_K}\left\{\|\Ib_{(\K_\iota)}\cdot\Vbh\|_F^2\right\}\notag\\
  &\overset{\$}{\leqslant} 1+\frac{N}{\tau}\cdot\left(|\K_\iota|\cdot\max_{j\in\N_N}\left\{\|\eb_j^T\cdot\Vbh\|_2^2\right\}\right)\notag\\
  &\leqslant 1+\frac{N}{\tau}\cdot\big(\tau\cdot(\eta \cdot d\log(Nd/\delta)/N)\big) \tag*{[Lemma \ref{fl_lem}]}\notag\\
  &= 1+\eta \cdot d\log(Nd/\delta) \label{bound_gamma}\\
  &=1+N\alpha\notag
\end{align}
for $\alpha=\eta d\cdot\log(Nd/\delta)/N$ where in $\$$ we used the fact that
$$ \|\Ib_{(\K_\iota)}\cdot\Vbh\|_F^2 = \sum_{j\in\K_\iota}\|\eb_j^T\cdot\Vbh\|_2^2 \leqslant |\K_\iota|\cdot\max_{j\in\K_\iota}\left\{\|\eb_j^T\cdot\Vbh\|_2^2\right\} . $$
From the above derivation, it follows that
\begin{align*}
  \|\Vbh_{(\K^i)}\|_2^2 &= \|\Vbh_{(\K^i)}^T\Vbh_{(\K^i)}\|_2 \\
  &\leqslant \frac{\tau}{N}\cdot\left(1+\eta \cdot d\log(Nd/\delta)-\|\Ib_d\|_2\right)\\
  &=\tau \eta d/N\cdot\log(Nd/\delta)\\
  &=\tau\alpha 
\end{align*}
for all $\iota\in\N_K$. By setting $\tau=1$, we get an upper bound on the squared $\ell_2$-norm of the rows of $\Vbh$:
\begin{equation}
\label{bound_out_prod_norm}
  \|\Vbh_l\|_2^2 =\|\Vbh_l\Vbh_l^T\|_2 =\|\Vbh_l^T\Vbh_l\|_2 \leqslant \alpha% C^2d/N\cdot\log(Nd/\delta)
\end{equation}
where $\Vbh_l=\Vbh_{(l)}$, for all $l\in\N_N$.

Next, we compute $\Eb\coloneqq\E[\Xb^T\Xb+\Ib_d]$ and its eigenvalues. By the definition of $\Xb$ and its realizations:
{\small
\begin{align*}
  \Xb_i^T\Xb_i &= \left(\Ib_d-N/\tau\cdot \Vbh_{(\K^i)}^T\Vbh_{(\K^i)}\right)^T \cdot \left(\Ib_d-N/\tau\cdot \Vbh_{(\K^i)}^T\Vbh_{(\K^i)}\right)\\
  &= \Ib_d-2\cdot\frac{N}{\tau}\cdot \Vbh_{(\K^i)}^T\Vbh_{(\K^i)} + \left(\frac{N}{\tau}\right)^2\cdot \Vbh_{(\K^i)}^T\Vbh_{(\K^i)}\Vbh_{(\K^i)}^T\Vbh_{(\K^i)}
\end{align*}
}
thus $\Eb$ is evaluated as follows:
{\small
\begin{subequations}
\label{E_eval}  % evaluation of E
\begin{align*}
  \E[&\Xb^T\Xb+\Ib_d] = 2\Ib_d - 2\cdot\left(N/\tau\right)\cdot\E\left[ \Vbh_{(\K^i)}^T\Vbh_{(\K^i)}\right] \\ &{\white=}+ \left(N/\tau\right)^2\cdot \E\left[\Vbh_{(\K^i)}^T\Vbh_{(\K^i)}\Vbh_{(\K^i)}^T\Vbh_{(\K^i)}\right]\\
  &= 2\Ib_d - 2\cdot\left(N/\tau\right)\cdot\left({\textstyle\sum_{j=1}^K}K^{-1}\cdot \Vbh_{(\K_j)}^T\Vbh_{(\K_j)}\right) \\ &{\white=}+ \left(N/\tau\right)^2\cdot\left({\textstyle\sum_{j=1}^K}K^{-1}\cdot \Vbh_{(\K_j)}^T\left(\Vbh_{(\K_j)}\Vbh_{(\K_j)}^T\right)\Vbh_{(\K_j)}\right)\\
  &= 2\Ib_d - 2\cdot\left({\textstyle\sum_{l=1}^N} \Vbh_l^T\Vbh_l\right) + (N/\tau)\cdot\left({\textstyle\sum_{l=1}^N} \Vbh_l^T\left(\Vbh_l\Vbh_l^T\right)\Vbh_l\right)\\
  &= K\cdot\left({\textstyle\sum_{l=1}^N} \langle \Vbh_l,\Vbh_l\rangle\cdot \Vbh_l^T\Vbh_l\right)
\end{align*}
\end{subequations}
}
where in the last equality we invoked $\sum_{l=1}^N \Vbh_l^T\Vbh_l=\Ib_d$.% the fact that $\sum_{l=1}^N \Vbh_l^T\Vbh_l=\Ib_d$.

In order to bound the variance of the matrix random variable $\Xb$, we bound the largest eigenvalue of $\Eb$; by comparing it to the matrix
$$ \Fb=K\alpha\cdot\left(\sum_{l=1}^N\Vbh_l^T\Vbh_l\right)=K\alpha\cdot\Ib_d $$
%$$ \Fb=K\cdot\left(C^2d/N\cdot\log(Nd/\delta)\cdot\sum_{l=1}^N\Vbh_l^T\Vbh_l\right). $$
whose eigenvalue $K\alpha$ is of algebraic multiplicity $d$. It is clear that $\Eb$ and $\Fb$ are both real and symmetric; thus they admit an eigendecomposition of the form $\Qb\bold{\Lambda}\Qb^T$. Note also that for all $\yb\in\R^d$:
\begin{align}
  \yb^T\Eb\yb &= K\cdot \yb^T\left(\sum_{l=1}^N \Vbh_l^T\left(\Vbh_l\Vbh_l^T\right)\Vbh_l\right)\yb\notag\\
  &\overset{\sharp}{=} K\cdot\sum_{l=1}^N \langle \yb,\Vbh_l\rangle^2\cdot\|\Vbh_l\|_2^2\notag\\
  &\overset{\flat}{\leqslant} K\alpha\cdot\sum_{l=1}^N \langle \yb,\Vbh_l\rangle^2 \label{bound_quadr_E}\\
  &= K\alpha\cdot\sum_{l=1}^N\yb^T\Vbh_l^T\cdot\Vbh_l\yb\notag\\
  &= \yb^T\left(K\alpha\cdot\sum_{l=1}^N\Vbh_l^T\cdot\Vbh_l\right)\yb\notag\\
  &= \yb^T\Fb\yb\notag
\end{align}
where in $\flat$ we invoked \eqref{bound_out_prod_norm}. By $\sharp$ we conclude that $\yb^T\Eb\yb\geqslant0$, thus $\Fb\succeq\Eb\succeq0$.

Let $\wb_i,\zb_i$ be the unit-norm eigenvectors of $\Eb,\Fb$ corresponding to their respective $i^{th}$ largest eigenvalue. Then
$$ \wb_i^T\left(\Qb_\Eb\bold{\Lambda}_\Eb\Qb_\Eb^T\right)\wb_i = \eb_i^T\cdot\bold{\Lambda}_\Eb\cdot\eb_i = \lambda_i \quad {\white\implies} $$
and by \eqref{bound_quadr_E} we bound this as follows:
\begin{align*}
  \lambda_i = \wb_i^T\Eb\wb_i \leqslant K\alpha\cdot\sum_{l=1}^N \langle \wb_i,\Vbh_l\rangle^2.
  %\lambda_i = \wb_i^T\Eb\wb_i &= K\cdot\left(\sum_{l=1}^N \langle \Vbh_l,\Vbh_l\rangle\cdot\langle \wb_i,\Vbh_l\rangle^2\right)\\
  %&\leqslant K\cdot\left(\sum_{l=1}^N C^2d/N\cdot\log(Nd/\delta)\cdot\langle \wb_i,\Vbh_l\rangle^2\right)\\
  %&= K\cdot C^2d/N\cdot\log(Nd/\delta)\cdot\left(\sum_{l=1}^N\langle \wb_i,\Vbh_l\rangle^2\right) \geqslant 0.
\end{align*}
Since
$$ \wb_1 = \arg\max_{\substack{\vb\in\R^d\\ \|\vb\|_2=1}}\big\{\vb^T\Eb\vb\big\} \ \implies \ \|\Eb\|_2 = \lambda_1=\wb_1^T\Eb\wb_1, $$
and $\Fb\succeq\Eb\geqslant0$, it follows that
\begin{align*}
  \|\Eb\|_2&=\wb_1^T\Eb\wb_1 \leqslant \wb_1^T\Fb\wb_1\\
  &\leqslant \arg\max_{\substack{\vb\in\R^d\\ \|\vb\|_2=1}}\big\{\vb^T\Fb\vb\big\} = \|\Fb\|_2 = K\alpha.
\end{align*}
%$$ \|\Eb\|_2=\wb_1^T\Eb\wb_1 \leqslant \wb_1^T\Fb\wb_1 \leqslant \arg\max_{\substack{\vb\in\R^d\\ \|\vb\|_2=1}}\big\{\vb^T\Fb\vb\big\} = \|\Fb\|_2 = K\alpha. $$
In turn, this gives us 
\begin{align}
  \|\E[\Xb^T\Xb]\|_2 &= \|\Eb-\Ib_d\|_2\notag\\
  &\leqslant \|\Eb\|_2+\|\Ib_d\|_2\notag\\
  &\leqslant \|\Fb\|_2+1\notag\\
  &= K\alpha+1\notag\\
  &\leqslant \eta K\frac{d}{N}\log(Nd/\delta)+1 \notag\\
  &= \eta \frac{d}{\tau}\log(Nd/\delta)+1 \label{bound_var}
\end{align}
hence $\|\E[\Xb^T\Xb]\|_2=O\big(\frac{d}{\tau}\log(Nd/\delta)\big)$.

We now have everything we need to apply Theorem \ref{matr_Chern}.

\begin{Prop}
\label{prop_SRHT_b}
The block-SRHT $\SbPh$ guarantees %from Algorithm \ref{alg_SRHT_b} guarantees
$$ \Pr\Big[\|\bold{I}_d-\Ub^T\SbPh^T\SbPh\Ub\|_2>\epsilon\Big] \leqslant 2d\cdot \exp\left\{\frac{-\epsilon^2\cdot q}{\Theta\left(\frac{d}{\tau}\cdot\log(Nd/\delta)\right)}\right\} $$
for any $\epsilon>0$, and $q=r/\tau>d/\tau$.
%The block-SRHT $\SbPi$ guarantees that $\|\bold{I}_d-\Ub^T\SbPi^T\SbPi\Ub\|_2>\epsilon$; with probability at most $2d\cdot \exp\left\{-\epsilon^2\cdot q\big/\Theta\left(\frac{d}{\tau}\cdot\log(Nd/\delta)\right)\right\}$, for any $\epsilon>0$, and $q=r/\tau>d/\tau$.
\end{Prop}

\begin{proof}{[Proposition \ref{prop_SRHT_b}]}
Let $\{\Xb_i\}_{i=1}^q$ as defined in \eqref{def_Xi} denote $q$ block samples. Let $j(i)$ denote the index of the submatrix which was sampled at the $i^{th}$ random trial, \textit{i.e.} $\K_{j(i)}=\K_{j(i)}^i$. Then
{\small
\begin{align*}
  \Zb &= \frac{1}{q}\sum_{i=1}^t \Xb_{j(i)}\\
  &= \frac{1}{q}\cdot\sum_{i=1}^q\left(\Ib_d-\frac{N}{\tau}\cdot \Vbh_{(\K_{j(i)})}^T\Vbh_{(\K_{j(i)})}\right)\\
  %&= \Ib_d-\frac{1}{q}\cdot\frac{N}{\tau}\cdot\sum_{i=1}^q\left(\Vbh_{(\K_{j(i)})}^T\Vbh_{(\K_{j(i)})}\right)\\
  %&= \Ib_d-\frac{N}{r}\cdot\sum_{i=1}^q\left(\Vbh_{(\K_{j(i)})}^T\Vbh_{(\K_{j(i)})}\right)\\
  &= \Ib_d-\sum_{i=1}^q\left(\sqrt{N/r}\cdot \Vbh_{(\K_{j(i)})}\right)^T\cdot\left(\sqrt{N/r}\cdot \Vbh_{(\K_{j(i)})}\right)\\
  &= \Ib_d-\sum_{i=1}^q\left(\sqrt{N/r}\cdot\Ib_{(\K_{j(i)})}\cdot \Vbh\right)^T\cdot\left(\sqrt{N/r}\cdot\Ib_{(\K_{j(i)})}\cdot \Vbh\right)\\
  &= \Ib_d-\left(\Ombwt\Hbh\Db\Ub\right)^T\cdot\left(\Ombwt\Hbh\Db\Ub\right)\\
  &= \Ib_d - \Ub^T\SbPh^T\SbPh\Ub.
\end{align*}
}

We apply Lemma \ref{matr_Chern} by fixing the terms we bounded: \eqref{bound_gamma} $\gamma=\eta d\log(Nd/\delta)+1$, \eqref{bound_var} $\sigma^2=\eta \frac{d}{\tau}\log(Nd/\delta)+1$, and fix $q$ and $\epsilon$. The denominator of the exponent in \eqref{matr_Chern_expr} is then
\begin{align*}
  \big(\eta &d/\tau\cdot\log(Nd/\delta)+1\big)+\big((\eta d\log(Nd/\delta)+1)\cdot\epsilon/3\big) = \\
  &= \eta d/\tau\cdot\log(Nd/\delta)\cdot\big(1+\epsilon\tau/3\big)+(1+\epsilon/3)\\
  &= \Theta\left(\frac{d}{\tau}\log(Nd/\delta)\right) 
\end{align*}
and the proof is complete.
\end{proof}

\begin{proof}{[Theorem \ref{subsp_emb_thm_SRHT}]}
By substituting $q$ in the bound of Proposition \ref{prop_SRHT_b} and taking the complementary event, we attain the statement.
\end{proof}

\subsection{The Hadamard Transform}

\begin{Rmk}
The Hadamard matrix is a real analog of the discrete Fourier matrix, and there exist matrix multiplications algorithms for the Hadamard transform which resemble the FFT algorithm. Recall that the Fourier matrix represents the characters of the cyclic group of order $N$. In this case, $\Hbh_N$ represents the characters of the group $(\Z_2^N,+)$, where $\Z_2^n\cong\Z_{N}$. For both of these transforms, it is precisely through this algebraic structure which one can achieve a matrix-vector multiplication in $\ow(N\log N)$ arithmetic operations.
\end{Rmk}

Recall that the characters of a group $G$, form an orthonormal basis of the vector space of functions over the Boolean hypercube, \textit{i.e.} $\mathcal{F}_n=\left\{f:\{0,1\}^n\to\R\right\}$, and it is the Fourier basis. Furthermore, when working over groups of characteristic 2, \textit{e.g.} $\F_{2^q}\cong\F_2^q$ for $q\in\Z_+$, we can move everything so that the underlying field is $\R$. Specifically, we map the elements of the binary field to $\R$ by applying $f(y)=1-2y$. This gives us $f:\{0,1\}\mapsto\{+1,-1\}\subseteq \R$, and we can work with addition and multiplication over $\R$. 
%Specifically, by transforming the elements of $\{0,1\}$ through $f(y)=1-2y$; $f:\{0,1\}\mapsto\{+1,-1\}\subseteq \R$, we can work with addition and multiplication over $\R$.
%(so $0\mapsto+1$ and $1\mapsto-1$)

We note that there is a bijective correspondence between the characters of $\Z_m$ and the $m^{th}$ root of unity, which is precisely how we get an orthonormal (Fourier) basis. In the case where $m$ is not a power of two, we have a basis with complex elements, which violates (c) in the list of properties we seek to satisfy. This is why the Hadamard matrix is appropriate for our application, and why we do not consider a general discrete Fourier transform.

\subsection{Recursive Kronecker Products of Orthonormal Matrices}

In this subsection, we show that multiplying a vector of length $N$ with $\Piba=\Pibold_k^{\otimes \ceil{\log_k(N)}}$ for $\Pibold_k\in O_k(\R)$ and $k\in\Z_{>2}$, takes $\ow(Nk^2\log_k N)$ elementary operations. Therefore, multiplying $\Ab\in\R^{N\times d}$ with $\Piba$ takes $\ow(Ndk^2\log_k N)$ operations. We follow a similar analysis to that of \cite[Section 6.10.2]{Osg09}.

For $C(N)$ the number of elementary operations involved in carrying out the above matrix-vector multiplication, the basic recursion relationship is
\begin{equation}
\label{rec_relation}
  C(N) = k(C/k) + NC(1)
  %C(N) = k(C/k) + \gamma N
\end{equation}
where $C(1)=\zeta k^2$, for $\zeta >0$ a constant.%for $\gamma>0$ a constant, such that $C(1)=\gamma$.

For $p=\ceil{\log_k(N)}$, we have the following relationship:
\begin{equation}
\label{expr_C}
  T(p) = \frac{C(N)}{N} \quad \implies \quad C(N) = NT(p).
\end{equation}
Then, $p-1=\ceil{\log_k(N/k)}$, which gives us
\begin{equation}
\label{expr_T}
  T(p-1)=\frac{C(N/k)}{N/k}=k\frac{C(N/k)}{N} \quad \implies \quad NT(p-1) = kC(N/k) .
\end{equation}
By substituting \eqref{expr_T} into \eqref{rec_relation}, we get
$$ C(N) = NT(p) = NT(p-1)+NC(1), $$
thus $T(p)=T(p-1)+C(1)$, which implies that $T(p)$ is linear. Therefore
$T(p)=pC(1)=\zeta k^2p$, and from \eqref{expr_C} we conclude that the total number of elementary operations is
$$ C(N) = NT(p) = N\zeta k^2p = N\zeta k^2\ceil{\log_k(N)} = \ow(Nk^2\log_k N). $$
%$$ C(N)=\zeta Nk^2\log_k N = \ow(Nk^2\log_k N). $$

% - - - - - - - - - - - - - - -
\section{Proofs of Section \ref{opt_ss_sec}}

In this appendix, we present the proofs of Proposition \ref{prop_opt_ss} and Corollary \ref{cor_opt_ss}.

\begin{proof}{[Proposition \ref{prop_opt_ss}]}
Note that the optimization problem \eqref{ss_opt_prob} is equivalent to
\begin{equation}
\label{ss_opt_prob_alt}  % alternative formulation
  \xi_t^{\star} = \arg\min_{\xi\in\R}\Big\{\|\Ab\xb^{[t+1]}-\bb\|_2^2\Big\}.
\end{equation}
If we cannot decrease further, the optimal solution to \eqref{ss_opt_prob_alt} will be 0, and we can never have $\xi_{t}<0$, as this would imply that
\begin{align*}
  \|\Ab\xb^{[t+1]}-\bb\|_2^2 &= \|\Ab(\xb^{[t]}-\xi_t\cdot g^{[t]})\bb\|_2^2 > \|\Ab\xb^{[t]}-\bb\|_2^2
\end{align*}
which contradicts the fact that we are minimizing the objective function of \eqref{ss_opt_prob_alt}. Specifically, if $\xi_t<0$, we get an ascent step in \eqref{par_upd}, and a step-size $\xi_t=0$ achieves a lower value. It therefore suffices to prove the given statement by solving \eqref{ss_opt_prob_alt}.

We will first derive \eqref{opt_ss} for $L_{ls}(\Ab,\bb;\xb^{[t]})$, and then show it is the same for the optimization problems $L_{\Pibold}(\Ab,\bb;\xb^{[t]})$ and $L_{\Gb}(\Ab,\bb;\xb^{[t]})$.
  
Recall that $\xb^{[t+1]}\gets\xb^{[t]}-\xi_t\cdot g_{ls}^{[t]}$ for the least squares objective $L_{ls}(\Ab,\bb;\xb^{[t]})$. From here onward, we denote the gradient update of the underlying objective function by $g_t$. We then reformulate the objective function of \eqref{ss_opt_prob} as follows
$$ \Delta_{t+1} \coloneqq \|\Ab\xb^{[t+1]}-\bb\|_2^2=\|\Ab(\xb^{[t]}-\xi\cdot g^{[t]})-\bb\|_2^2. $$
By expanding the above expression, we get
\begin{align*}
  \Delta_{t+1} = \xi^2\cdot&\|\Ab g_t\|_2^2-2\xi\cdot\left(g_t^T\Ab^T\Ab\xb^{[t]}-g_t^T\Ab^T\bb\right) +\\
  &+ \left(\|\Ab\xb^{[t]}\|_2^2-2\langle\Ab\xb^{[t]},\bb\rangle+\langle\bb,\bb\rangle\right)
\end{align*}
and by setting $\frac{\partial\Delta_{t+1}}{\partial\xi}=0$ and solving for $\xi$, it follows that
$$ \frac{\partial\Delta_{t+1}}{\partial\xi} = 2\xi\cdot\left(g_t^T\Ab^T\Ab g_t\right)-2\cdot\left(g_t^T\Ab^T\right)(\Ab\xb^{[t]}-\bb) = 0 $$
\begin{equation}
\label{ss_update}
  \implies \quad \xi_t^{\star} = \frac{\langle\Ab g_t,\Ab\xb^{[t]}-\bb\rangle}{\|\Ab g_t\|_2^2},
\end{equation}
which is the updated step-size we use at the next iteration. Since $\partial^2\Delta_{t+1}/\partial\xi^2=2\|\Ab g_t\|_2^2\geqslant0$, we know that $\Delta_{t+1}$ is convex. Therefore, $\xi_t^{\star}$ derived in \eqref{ss_update} is indeed the minimizer of $\Delta_{t+1}$.

Now consider SD with the objective function $L_{\Pibold}(\Ab,\bb;\xb^{[t]})$. The only thing that changes in the derivation, is that now we have $(\Abg=\Pibold\Ab,\bbg=\Pibold\bb)$ instead of $(\Ab,\bb)$. By replacing $(\Abg,\bbg)\gets(\Ab,\bb)$ in \eqref{ss_update}, it follows that
\begin{equation}
\label{ss_update_Pi}
  \frac{\langle\Abg g_t,\Abg\xb^{[t]}-\bbg\rangle}{\|\Abg g_t\|_2^2} = \frac{\langle\Ab g_t,\Ab\xb^{[t]}-\bb\rangle}{\|\Ab g_t\|_2^2}
\end{equation}
as $\Abg^T\Abg=\Ab^T\Ab$ and $\Abg^T\bbg=\Ab^T\bb$ , since $\Pibold\in O_N(\R)$. The step-sizes for the corresponding iterations are therefore identical.

Moreover, the only difference between the objective functions $L_{\Pibold}(\Ab,\bb;\xb^{[t]})$ and $L_{\Gb}(\Ab,\bb;\xb^{[t]})$ is the factor of $\sqrt{N/r}$. Let $\Abt=\Gb\Ab$ and $\bbt=\Gb\bb$. Therefore, the step-size at iteration $t+1$ when considering the objective function $L_{\Gb}(\Ab,\bb;\xb^{[t]})$ is
\begin{align*}
  \frac{\langle\Abt g_t,\Abt\xb^{[t]}-\bbt\rangle}{\|\Abt g_t\|_2^2} &= \frac{N/r}{N/r}\cdot\frac{\langle\Abg g_t,\Abg\xb^{[t]}-\bbg\rangle}{\|\Abt g_t\|_2^2} \\
  &\overset{\diamond}{=} \frac{\langle\Ab g_t,\Ab\xb^{[t]}-\bb\rangle}{\|\Ab g_t\|_2^2}
\end{align*}
where $\diamondsuit$ follows from \eqref{ss_update_Pi}.
\end{proof}

\begin{proof}{[Corollary \ref{cor_opt_ss}]}
We want to show that $\xi_t^{\star}$ according to \eqref{opt_ss}, is a solution to \eqref{opt_exp_ss}. We know that the only difference in the induced sketching matrices $\SbPi^{[t]}$ at each iteration are the resulting index sets $\Scal^{[t]}$, and the corresponding sampling and rescaling matrices $\Ombwt_{[t]}$.

To prove the given statement, since $\SbPi^{[t]}=\Ombwt_{[t]}\Pibold$; and by Proposition \eqref{prop_opt_ss} $\xi_t^{\star}$ is a solution to
\begin{equation*}
  \arg\min_{\xi\in\R}\Big\{\big\|\Pibold\big(\Ab\xbh^{[t+1]}-\bb\big)\big\|_2^2\Big\},
\end{equation*}
it suffices to show that $\E\left[\Ombwt_{[t]}^{T}\Ombwt_{[t]}\right]=\Ib_N$. This was proven in Lemma \ref{lemma_exp}. Hence, the proof is complete. 
\end{proof}

% - - - - - - - - - - - - - - -
\section{Proofs of Section \ref{security_sec}}

In this appendix, we present the proofs of Theorems \ref{Shan_secr_thm} and \ref{SRHT_comp_sec_thm}, and Corollary \ref{cor_fl_lem}. We also present a counterexample to perfect secrecy of the SRHT.

\begin{proof}{[Theorem \ref{Shan_secr_thm}]}
Denote the application of $\Pibold$ to a matrix $\Mb$ by $\Enc_\Pibold(\Mb)=\Pibold\Mb$. We will prove secrecy of this scheme, which then implies that a subsampled version of the transformed information is also secure. Let $\Abg=\Enc_\Pibold(\Ab)$ and $\bbg=\Enc_\Pibold(\bb)$.

The adversaries' goal is to reveal $\Ab$. To prove that $\Enc_\Pibold$ is a well-defined security scheme, we need to show that an adversary cannot recover $\Ab$; with only knowledge of $(\Abg,\bbg)$.

For a contradiction, assume an adversary is able to recover $\Ab$ after only observing $(\Abg,\bbg)$. This means that it was able to obtain $\Pibold^{-1}$, as the only way to recover $\Ab$ from $\Abg$ is by inverting the transformation of $\Pibold$: $\Ab=\Pibold^{-1}\cdot\Abg$. This contradicts the fact that only $(\Abg,\bbg)$ were observed. Thus, $\Enc_\Pibold$ is a well-defined security scheme.

It remains to prove perfect secrecy according to Definition \ref{Sh_secr}. Observe that for any $\bar{\Ub}\in\M$ and $\bar{\Qb}\in\Cc$
\begin{equation} \Pr_{\Pibold\getsU\K}\left[\Enc_\Pibold(\bar{\Ub})=\bar{\Qb}\right] = \Pr_{\Pibold\getsU\K}\left[\Pibold\cdot\bar{\Ub}=\bar{\Qb}\right] = \ind \end{equation}
\begin{equation} \ind = \Pr_{\Pibold\getsU\K}\left[\Pibold=\bar{\Qb}\cdot\bar{\Ub}^{-1}\right] \overset{\sharp}{=} \frac{1}{|\Otil_\Ab|} = \frac{1}{|\K|} \end{equation}
where $\sharp$ follows from the fact that $\bar{\Qb}\cdot\bar{\Ub}^{-1}$ is fixed. Hence, for any $\Ub_0,\Ub_1\in\M$ and $\bar{\Qb}\in\Cc$ we have
$$ \Pr_{\Pibold\getsU\K}\left[\Enc_\Pibold(\Ub_0)=\bar{\Qb}\right] = \frac{1}{|\K|} = \Pr_{\Pibold\getsU\K}\left[\Enc_\Pibold(\Ub_1)=\bar{\Qb}\right] $$
as required by Definition \ref{Sh_secr}. This completes the proof.
\end{proof}

We note that through the SVD of $\Abg$, the adversaries can learn the singular values and right singular vectors of $\Ab$, since
\begin{equation}
  \Abg=(\Pibold\cdot\Ub_\Ab)\cdot\Sigb_\Ab\cdot\Vb_\Ab^T=\Ub_\Abg\cdot\Sigb_\Ab\cdot\Vb_\Ab^T .
\end{equation}
Recall that the singular values are unique and, for distinct positive singular values, the corresponding left and right singular vectors are also unique up to a sign change of both columns. We assume w.l.o.g. that $\Vb_\Abg=\Vb_\Ab$ and $\Ub_\Abg=\Pibold\cdot\Ub_\Ab$.

Geometrically, the encoding $\Enc_\Pibold$ changes the orthonormal basis of $\Ub_\Ab$ to $\Ub_\Abg$, by rotating it or reflecting it; when $\det(\Pibold)$ is +1 or -1 respectively. Of course, there are infinitely many ways to do so, which is what we are relying the security of this approach on.

Furthermore, unless $\Ub_\Ab$ has some special structure (\textit{e.g.}, triangular, symmetric, etc.), one cannot use an off-the-shelf factorization to reveal $\Ub_\Ab$. Even though a lot can be revealed about $\Ab$, \textit{i.e.} $\Sigma_\Ab$ and $\Vb_\Ab$, we showed that it is not possible to reveal $\Ub_\Ab$; hence nor $\Ab$, without knowledge of $\Pibold$.

\begin{proof}{[Corollary \ref{cor_fl_lem}]}
The proof is identical to that of Lemma \ref{fl_lem}. The only difference is that the random variable entries $\tilde{Z}_j=\Hbt_{ij}\Db_{jj}\yb_j$ for $j\in\N_N$ and the fixed $i$ now differ, though they still meet the same upper bound
$$|\tilde{Z}_j| \leqslant |\Hbt_{ij}|\cdot|\Db_{jj}|\cdot|\yb_j| = \frac{|\yb_j|}{\sqrt{N}}. $$
Since \eqref{pr_sum_Zj} holds true, the guarantees implied by flattening lemma also do, thus the sketching properties of the SRHT are maintained.
\end{proof}

\begin{Rmk}
Since the Lemma \ref{fl_lem} and Corollary \ref{cor_fl_lem} give the same result for the block-SRHT and garbled block-SRHT respectively, it follows that Theorem \ref{subsp_emb_thm_SRHT} also holds for the garbled block-SRHT.
\end{Rmk}

\begin{proof}{[Theorem \ref{SRHT_comp_sec_thm}]}
Assume w.l.o.g. that a computationally bounded adversary observes $\Pibt\Ab$, for which
$\widetilde{\Ab}_r=\SbPi\cdot\Ab=\Ombwt\cdot(\Pibt\Ab)$ is the resulting sketch of Algorithm \ref{alg_orthog_sketch}, for $\Pibt\in\Ht_N$. To invert the transformation of $\Pibt$, the adversary needs knowledge of the components of $\Pibt$, \textit{i.e.} $\Hbh$ and $\Pb$. Assume for a contradiction that there exists a probabilistic polynomial-time algorithm which, is able to recover $\Ab$ from $\Pibt\Ab$. This means that it has revealed $\Pb$, so that it can compute
$$ \overbrace{(\Db\Hbh\Pb^T)}^{\Pibt^T=\Pibt^{-1}}\cdot(\Pb\Hbh\Db)\cdot\Ab = \Pibt^{-1}\cdot\Pibt\cdot\Ab = \Ab, $$
which contradicts the assumption that the permutation $\Pb$ is a s-PRP. Specifically, recovering $\Ab$ by observing $\Pibt\Ab$ requires finding $\Pb$ in polynomial time.
\end{proof}

Finally, we show that $\gh^{[t]}=g_{ls}^{[t]}$, which we claimed in Subsection \ref{exact_grad_subsec}. Since $\Pibold\in O_N(\R)$ for the suggested projections (except that random Rademacher projection), we have $\Pibold^T\Pibold=\Ib_N$. It then follows that
\begin{align}
  \gh^{[t]} &= 2\sum\limits_{j=1}^K\Abt_j^T\left(\Abt_j\xb^{[t]}-\bbt_j\right) \notag \\
  &= \left(\Pibold\Ab\right)^T\left(\Pibold\Ab\xb^{[t]}-\Pibold\bb\right) \notag \\
  &= \Ab^T\left(\Pibold^T\Pibold\right)\left(\Ab\xb^{[t]}-\bb\right) \notag \\
  &= g_{ls}^{[t]} \notag %\label{eq_orth_encr}  % equality orthonormal encryption
\end{align}
and this completes the derivation.

\subsection{Counterexample to Perfect Secrecy of the SRHT}
\label{SRHT_counter_example}

Here, we present an explicit example for the SRHT (which also applies to the block-SRHT), which contradicts Definition \ref{Sh_secr}. Therefore, the SRHT cannot provide perfect secrecy.

Consider the simple case where $N=2$, and assume that $\Hbh_2\in\Otil_\Ab$. Since $(\Otil_\Ab,\cdot)$ is a multiplicative subgroup of $\GL_2(\R)$, we have $\Ib_2\in\Otil_\Ab$. Let $\Ub_0=\Ib_2$ and $\Ub_1=\Hbh_2$.

For $d_1,d_2$ i.i.d. Rademacher random variables and
$$ \Db = \begin{pmatrix} d_1 & 0\\ 0 & d_2 \end{pmatrix} , $$
it follows that
\begin{align*}
  \Cb_0=\left(\Hbh_2\Db\right)\cdot\Ub_0 = \Hbh_2\Db = \frac{1}{2}\begin{pmatrix} d_1 & -d_2\\ d_1 & d_2\end{pmatrix}
\end{align*}
and
\begin{align*}
  \Cb_1=\left(\Hbh_2\Db\right)\cdot\Ub_1 &= \frac{1}{2} \begin{pmatrix} 1 & -1\\ 1 & 1\end{pmatrix} \begin{pmatrix} d_1 & 0\\ 0 & d_2 \end{pmatrix} \begin{pmatrix} 1 & -1\\ 1 & 1\end{pmatrix}\\
  &=\frac{1}{2} \begin{pmatrix} 1 & -1\\ 1 & 1\end{pmatrix} \begin{pmatrix} d_1 & -d_1\\ d_2 & d_2 \end{pmatrix}\\
  &= \frac{1}{2}\begin{pmatrix} d_1-d_2 & -d_1-d_2\\ d_1+d_2 & -d_1+d_2 \end{pmatrix} .
\end{align*}
It is clear that $\Cb_0$ always has precisely two distinct entries, while $\Cb_1$ has three distinct entries; with 0 appearing twice for any pair $d_1,d_2\in\{\rpm1\}$. Therefore, depending on the observed transformed matrix, we can disregard one of $\Ub_0$ and $\Ub_1$ as being a potential choice for $\Pibold$.

For instance, if $\bar{\Cb}$ is the observed matrix and it has two zero entries, then
$$ \Pr_{\Pibold\getsU \hat{H}_2}\left[\Pibold\cdot\Ub_1=\bar{\Cb}\right] > \Pr_{\Pibold\getsU \hat{H}_2}\left[\Pibold\cdot\Ub_0=\bar{\Cb}\right]=0 $$
which contradicts \eqref{perf_secrecy_id}.

Note that even if we apply a permutation, as in the case of the garbled block-SRHT, we still get the same conclusion. Hence, the garbled block-SRHT also does not achieve perfect secrecy.

% - - - - - - - - - - - - - - -
\subsection{Analogy with the One-Time Pad}
\label{OTP_app}

It is worth noting that the encryption resulting by the multiplication with $\Pibold$; under the assumptions made in Theorem \ref{Shan_secr_thm}, bares a strong resemblance with the one-time pad (OTP), which is the optimum cryptosystem with theoretically perfect secrecy. This is not surprising, as it is one of the few known perfectly secret encryption schemes.

The main difference between the two, is that the spaces we work over are the multiplicative group $(\Otil_\Ab,\cdot)$ whose identity is $\Ib_N$ in Theorem \ref{Shan_secr_thm}, and the additive group $\big((\Z_2)^\ell,+\big)$ in the OTP; whose identity is the zero vector of length $\ell$.

As in the OTP, we make the assumption that $\K,\M,\Cc$ are all equal to the group we are working over; $\Otil_\Ab$, which it is closed under multiplication. In the OTP, a message is revealed by applying the key on the ciphertext: if $c=m\oplus k$ for $k$ drawn from $\K$, then $c\oplus k=m$. Analogously here, for $\Pibold$ drawn from $\Otil_\Ab$: if $\bar{\Cb}=\Pibold\cdot\Ub_\Ab$, then $\bar{\Cb}^T\cdot\Pibold=(\Ub_\Ab^T\cdot\Pibold^T)\cdot\Pibold=\Ub_\Ab^T$. An important difference here is that the multiplication is not commutative.

Also, for two distinct messages $m_0,m_1$ which are encrypted with the same key $k$ to $c_0,c_1$ respectively, it follows that $c_0\oplus c_1=m_1\oplus m_2$ which reveals the XOR of the two messages. In our case, for the bases $\Ub_0,\Ub_1$ encrypted to $\Cb_0=\Pibold\Ub_0$ and $\Cb_0=\Pibold\Ub_1$ with the same projection matrix $\Pibold$, it follows that $\Cb_0^T\cdot\Cb_1=\Ub_0^T\cdot\Ub_1$.

%\balance
\end{document}